\documentclass[a4paper,11pt]{article}

\usepackage{a4wide}
\usepackage[margin=1in]{geometry}
\usepackage{times}

\usepackage[noadjust]{cite}
\usepackage{amsmath}
\usepackage{amssymb}
\usepackage{amsthm} 
\usepackage{algorithm}
\usepackage[noend]{algpseudocode}

\usepackage{array}
\usepackage{enumerate}

\usepackage{tikz}
\usetikzlibrary{fit,positioning}
\usepackage{varwidth}
\usetikzlibrary{shapes.geometric,arrows,decorations.markings,positioning,calc}
\newcommand*{\StrikeThruDistance}{0.2cm}
\tikzset{
    >=stealth',
    dummy/.style={},
    cProb/.style={
           text width =8em,
           minimum height=1em,
           text centered,
           font=\normalsize},
    rightarr/.style={
           ->,
           thick,
           shorten <=2pt,
           shorten >=2pt,},
    leftarr/.style={
           <-,
           thick,
           shorten <=2pt,
           shorten >=2pt,},
    leftrightarr/.style={
           <->,
           thick,
           shorten <=2pt,
           shorten >=2pt,},
    seprightarr/.style={
    	   ->,
           thick,
           shorten <=2pt,
           shorten >=2pt,
    	decoration={markings, mark=at position 0.5 with {
        	\draw [thick,-] 
            	++ (-\StrikeThruDistance,-\StrikeThruDistance) 
            	-- ( \StrikeThruDistance, \StrikeThruDistance);}
    	},
    	postaction={decorate},
	},
	sepleftarr/.style={
    	   <-,
           thick,
           shorten <=2pt,
           shorten >=2pt,
    	decoration={markings, mark=at position 0.5 with {
        	\draw [thick,-] 
            	++ (-\StrikeThruDistance,-\StrikeThruDistance) 
            	-- ( \StrikeThruDistance, \StrikeThruDistance);}
    	},
    	postaction={decorate},
	}
}
\tikzset{}

\theoremstyle{plain}
\newtheorem{thm}{Theorem}
\newtheorem{cor}[thm]{Corollary}
\newtheorem{lem}[thm]{Lemma}
\newtheorem{obs}[thm]{Observation}

\theoremstyle{definition}
\newtheorem{defn}[thm]{Definition}
\newtheorem{constraint}[thm]{Constraint}
\newtheorem{remark}[thm]{Remark}

\newenvironment{myitemize}
{ \begin{itemize}
    \setlength{\itemsep}{0pt}
    \setlength{\parskip}{0pt}
    \setlength{\parsep}{0pt}     }
{ \end{itemize}                  }

\newcommand{\Hom}[1]{\mathrm{\#\textsc{Hom}}(#1)}
\newcommand{\SHom}[1]{\mathrm{\#\textsc{SHom}}(#1)}
\newcommand{\Comp}[1]{\mathrm{\#\textsc{Comp}}(#1)}
\newcommand{\Ret}[1]{\mathrm{\#\textsc{Ret}}(#1)}

\newcommand{\LHom}[1]{\mathrm{\#\textsc{LHom}}(#1)}

\newcommand{\LSHom}[1]{\mathrm{\#\textsc{LSHom}}(#1)}
\newcommand{\LComp}[1]{\mathrm{\#\textsc{LComp}}(#1)}
\newcommand{\OALHom}[1]{\Ret{#1}}
\newcommand{\cOALHom}[1]{\OALHom{#1}^\text{conn}}
\newcommand{\DirRet}[1]{\mathrm{\#\textsc{Dir-Ret}}(#1)}
\newcommand{\bis}{\#\mathrm{\textsc{BIS}}}
\newcommand{\is}{\#\mathrm{\textsc{IS}}}
\newcommand{\sat}{\#\mathrm{\textsc{SAT}}}
\newcommand{\TCut}[1]{\mathrm{\#\textsc{MultiterminalCut}(#1)}}
\newcommand{\largecut}{\mathrm{\#\textsc{LargeCut}}}

\newcommand{\csp}{\mathrm{\#\textsc{CSP}}}

\newcommand{\DRet}[1]{\mathrm{\textsc{Ret}}(#1)}

\newcommand{\Zivny}{{\v{Z}}ivn{\'y}}

\renewcommand{\P}{\mathrm{P}}
\newcommand{\E}{\mathbf{E}}
\newcommand{\FP}{\mathrm{FP}}
\newcommand{\NP}{\mathrm{NP}}
\newcommand{\numP}{\#\mathrm{P}}
\newcommand{\leap}{\le_\mathrm{AP}}
\newcommand{\geap}{\ge_\mathrm{AP}}
\newcommand{\eqap}{\equiv_\mathrm{AP}}
\newcommand{\WR}[1]{\mathrm{WR}_{#1}}

\newcommand{\calH}{\mathcal{H}}

\newcommand{\calM}{\mathcal{M}}
\newcommand{\calL}{\mathcal{L}}

\newcommand{\calY}{\mathcal{Y}}
\newcommand{\boldS}{\mathbf{S}}

\newcommand{\hatN}{\widehat{N}}

\newcommand{\TGS}{T_{G,\boldS}}
\newcommand{\tGS}{t_{G,\boldS}}
\newcommand{\OmGSi}{\Omega_{G,\boldS,i}}
\newcommand{\OmGSk}{\Omega_{G,\boldS,k}}
\newcommand{\OmPlus}{\Omega^+_{G,\boldS}}
\newcommand{\OmCup}{\Omega_{G,\boldS}}

\newcommand{\sx}{s_\alpha}
\newcommand{\sy}{s_\beta}
\newcommand{\sz}{s_\gamma}

\newcommand{\Imp}{\ensuremath{\mathrm{Imp}}}

\newcommand{\IIv}{I_{\mathrm{v}}}
\newcommand{\IIe}{I_{\mathrm{e}}}
\newcommand{\IIf}{I_{\mathrm{f}}}
\newcommand{\IIb}{I_{\mathrm{b}}}
\newcommand{\Cv}{C_{\mathrm{v}}}
\newcommand{\Ce}{C_{\mathrm{e}}}
\newcommand{\Cf}{C_{\mathrm{f}}}
\newcommand{\Cb}{C_{\mathrm{b}}}
\newcommand{\Hve}{H_{\IIv,\IIe}}
\newcommand{\Hvfb}{H_{\IIv,\IIf,\IIb}}

\newcommand{\type}{T}

\newcommand{\abs}[1]{\left\vert #1 \right\vert}
\newcommand{\tv}[2]{\mathrm{d}_{\mathrm{TV}}\left(#1,#2\right)}
\newcommand*\from{\colon}
\let\epsilon=\varepsilon
\newcommand{\eps}{\ensuremath{\varepsilon}}
\DeclareRobustCommand{\stirling}{\genfrac\{\}{0pt}{}}
\newcommand{\ucp}[2]{#1 \times #2}

\let\originalleft\left
\let\originalright\right
\renewcommand{\left}{\mathopen{}\mathclose\bgroup\originalleft}
\renewcommand{\right}{\aftergroup\egroup\originalright}

\renewcommand{\hom}[3][]{{N^{#1}\bigl(#2 \rightarrow #3\bigr)}}
\newcommand{\sur}[2]{\hom[\mathrm{sur}]{#1}{#2}}
\newcommand{\comp}[2]{\hom[\mathrm{comp}]{#1}{#2}}

\newcommand{\prob}[3]{
\vbox{
  \begin{description}
    \item[\bf Name:] #1
    \vspace{-1.75ex}
    \item[\bf Input:] #2  
    \vspace{-1.75ex}
    \item[\bf Output:] #3
  \end{description}
}}

\newcommand{\sampleHom}[1]{\textsc{SampleHom}_{H, \calL^*}(#1)}
\newcommand{\countHom}[1]{\textsc{CountHom}_{H, \calL^*}(#1)}

\usepackage{filecontents}

\usepackage[hidelinks]{hyperref}

\title{The Complexity of Approximately Counting Retractions}
\author{Jacob Focke, Leslie Ann Goldberg and  Stanislav \Zivny 
\thanks{A preliminary version of this paper (without the proofs) appeared in the proceedings of the Thirtieth Annual {ACM-SIAM} Symposium on Discrete
	Algorithms, {SODA} 2019, San Diego, California, USA, January 6-9,
	2019. This preliminary version states a trichotomy for approximately counting retractions to trees rather than (more generally) to graphs of girth at least $5$.
The research leading to these results has received funding from 
the European Research Council under the European Union's Seventh Framework Programme (FP7/2007-2013) ERC grant agreement no.\ 334828 and under the European Union's Horizon 2020 research and innovation programme (grant agreement No 714532). Jacob Focke has received funding from the Engineering and Physical Sciences Research Council (grant ref: EP/M508111/1). Stanislav \Zivny\ was supported by a Royal Society University Research Fellowship. The paper 
reflects only the authors' views and not the views of the ERC or the European Commission. The European Union is not liable for any use that may be made of the information contained therein.}}
\date{12 March 2020}

\begin{document}
\maketitle
\begin{abstract}
Let $G$ be a graph that 
contains an induced subgraph $H$.
A  \emph{retraction} from $G$ to $H$ is a homomorphism from $G$ to $H$ 
that is the identity function on $H$. 
Retractions are very well-studied:
Given $H$, the complexity of deciding whether there is a retraction from an input graph~$G$ to~$H$
is completely classified, in the sense that it is known for which~$H$ this problem is tractable (assuming $\mathrm{P}\neq \mathrm{NP}$). 
Similarly, the complexity of (exactly) counting retractions from~$G$ to~$H$ is classified (assuming 
$\mathrm{FP}\neq \#\mathrm{P}$).  
However,   almost nothing is known about approximately counting retractions. 
Our first contribution is to give a complete trichotomy for approximately counting retractions to graphs without short cycles.
The result is as follows: (1)~Approximately counting retractions to a graph~$H$ of girth at least $5$ is in $\mathrm{FP}$ if every connected component of $H$ is a star, a single looped vertex, or an edge with two loops. 
(2)~Otherwise, if every component is an irreflexive caterpillar or a partially bristled reflexive path, then 
approximately counting retractions to~$H$ is
equivalent to approximately counting the independent sets of a bipartite graph --- a problem
  which is complete in the approximate counting complexity class $\mathrm{RH}\Pi_1$. (3)~Finally, if none of these hold, then approximately counting retractions to $H$ is equivalent to approximately counting the satisfying assignments of a Boolean formula.
  
Our second contribution is to locate the retraction counting problem for each $H$ in the complexity landscape of related approximate counting problems. Interestingly,  our results are in contrast to the situation in the exact counting context. We show that the problem of approximately counting retractions is separated both from the problem of approximately counting homomorphisms and 
from the problem of approximately counting list homomorphisms --- whereas for exact counting all three of these problems are interreducible. We also show that  the number of retractions is at least as hard to approximate as both the number of surjective homomorphisms and the number of compactions. In contrast, exactly counting compactions is the hardest of all of these exact counting problems.
\end{abstract}

\section{Introduction}\label{sec:Intro}
A \emph{homomorphism} from a graph~$G$ to 
a graph~$H$ is a function $h \from V(G) \to V(H)$ 
such that, for all $\{u,v\} \in E(G)$, we have $\{h(u),h(v)\}\in E(H)$. For example, suppose that $H$ is a path $a,b,c$. Then a homomorphism from $G$ to $H$ is a $3$-colouring of $G$ in which the colour classes $\{a,c\}$ and $\{b\}$ induce a bipartition of $G$. Suppose that $G$ itself contains a path $A,B,C$. Then a \emph{retraction} from $G$ to $H$ is a homomorphism from $G$ to $H$ that maps $A$ to $a$, $B$ to $b$ and $C$ to $c$. In general, let $G$ and $H$ be graphs such that $G$ contains a fixed copy of $H$ as an induced subgraph. Then a retraction from $G$ to $H$ is a homomorphism from $G$ to $H$ that is the identity function on this fixed copy of $H$ in $G$.

Retractions have been  studied over a long period of time~\cite{HellThesis, Hell1974, HR1987, Pesch1988}. 
In particular, the computational decision problem of determining whether there is a retraction from~$G$
to~$H$ is  well-studied~\cite{HellNesetrilBook, VikasCompRetCSP, Vikas4Vertex, Vikas2017, FederPseudoForest}. 
Retractions have also been studied under the name of \emph{one-or-all list homomorphisms}, \emph{pre-colouring extensions} or simply \emph{extensions}, see, e.g.,~\cite{FederLHomRefl, FederLHomIrrefl, BHT1992, KS1997, Tuza1997, Marx2006, FHH2009}. 
See Hell and Ne\v{s}et\v{r}il's review article~\cite{HN2008} for a more extensive list of such work.

Homomorphism \emph{counting} problems have been researched extensively as well~\cite{DGGJApprox, GJTreeHoms, DGJSampling, GGJList, GGJBIS, KelkThesis, GKP2004, DG2000, Hell2004b, BorgsCountingSurvey, 
FGZ2017}. The problem of exactly counting retractions has been studied recently and a complete complexity classification is given in~\cite{FGZ2017}. However, very little is known about \emph{approximately} counting retractions. In this work we do two things. First we give a complexity trichotomy for approximately counting retractions for the class of graphs that have girth at least $5$. Second we relate the complexity of approximately counting retractions to other approximate counting graph homomorphism problems.

\subsection{First Contribution: A Trichotomy for Approximately Counting Retractions to Graphs of Girth at least 5}\label{sec:FirstContribution}
A \emph{(self-)loop} is an edge from a vertex to itself. 
A \emph{cycle} is a walk $w_0 w_1 \cdots w_k w_0$ where $k>1$ and all vertices in $\{w_0,\ldots,w_k\}$ are distinct. The \emph{length} of the cycle is $k+1$.
We sometimes refer to length-$3$ cycles as ``triangles'' and to length-$4$ cycles as ``squares''.
The \emph{girth} of a graph $H$ is the length of a shortest cycle in $H$. If $H$ is acyclic (that is if $H$ is a forest with possibly some loops) then its girth is infinity.
In this work we give a complete complexity classification for 
the problem of approximately counting retractions to graphs that have a girth of at least $5$ (Theorem~\ref{thm:RetMain}). Thus, our classification applies to all graphs $H$ except to those that contain $3$-cycles or $4$-cycles.
We now informally introduce some notation and concepts in order to state this result.  

Given a  graph~$H$, we use $\Ret{H}$ to denote the problem of counting retractions
to~$H$, given as input a graph~$G$ containing a fixed copy of~$H$.

To investigate the complexity of approximate counting problems, Dyer, Goldberg, Greenhill and Jerrum~\cite{DGGJApprox} introduce the concept of an approximation-preserving reduction (AP-reduction). Intuitively, an AP-reduction from a problem $A$ to a problem $B$ is an algorithm that is a ``good'' approximation to $A$ if it has oracle access to a ``good'' approximation to $B$. We write $A\leap B$ if such an AP-reduction exists. Two problems 
that are studied in this paper
appear frequently as benchmark problems in this line of research. $\sat$ is the problem of counting the satisfying assignments of a Boolean formula.
This problem is complete for $\numP$ with respect to AP-reductions. $\bis$ is the problem of counting the independent sets of a bipartite graph. This problem is complete for the approximate counting complexity class  
$\mathrm{RH}\Pi_1$ (with respect to AP-reductions).
While it is not believed that there is an efficient approximation algorithm for~$\bis$, it is also not
believed that it is complete for $\numP$ with respect to AP-reductions. 

While, in general, the vertices of~$H$ may or may not have loops,
we will consider two special cases.
We say that a graph is \emph{irreflexive} if it does not contain any loops.
We say that it is  \emph{reflexive} if every vertex has a loop. 
A \emph{tree} may be irreflexive, reflexive, or neither, but it may not have any cycles.
A \emph{caterpillar} is an irreflexive tree which contains a path $P$ such that all vertices outside of $P$ have degree $1$. A \emph{partially bristled reflexive path} (formally defined in Definition~\ref{def:PBRP}) is a tree consisting of a reflexive path~$P$, together with a (possibly empty) set of 
unlooped ``bristle'' vertices~$U$  
and a matching connecting all of the vertices of $U$ to 
``internal'' vertices of~$P$ (vertices of~$P$ that are not endpoints of the path). A more formal definition, along with an example, is given in Section~\ref{sec:Preliminaries}.

\newcommand{\ThmRetMain}{
Let $H$ be a graph of girth at least $5$.
{
\renewcommand{\theenumi}{\roman{enumi})}
\renewcommand{\labelenumi}{\theenumi}
\begin{enumerate}            
\item If every connected component of $H$ is an irreflexive star, a single looped vertex, or an edge with two loops, then $\Ret{H}$ is in $\FP$.
\item Otherwise, if every connected component of $H$ is an irreflexive caterpillar or a partially bristled reflexive path, then $\Ret{H}$ is approximation-equivalent to $\bis$.
\item Otherwise, $\Ret{H}$ is approximation-equivalent to $\sat$.
\end{enumerate}
}
}
\begin{thm}\label{thm:RetMain}
\ThmRetMain
\end{thm}

Since there has been prior work on the problem of counting homomorphisms to trees~\cite{GJTreeHoms}, it is worth noting the special case of Theorem~\ref{thm:RetMain} where $H$ is an irreflexive tree. In this case, $\Ret{H}$ is in $\FP$ if $H$ is a star. If $H$ is a caterpillar but not a star, then $\Ret{H} \eqap \bis$. For all other irreflexive trees~$H$, the problem $\Ret{H}$ is $\sat$-equivalent under AP-reductions. 
The special case where $H$ is a reflexive tree is also easy to state. In this case, $\Ret{H}$ is in $\FP$ if $H$ is a single looped vertex or an edge with two loops. If $H$ is a reflexive path with at least three vertices, then $\Ret{H} \eqap \bis$. Otherwise, $\Ret{H}$ is $\sat$-equivalent with respect to AP-reductions. 

The special case of our classification where $H$ is irreflexive is given in Theorem~\ref{thm:RetIrreflexive}. It is slightly more general than what is given in Theorem~\ref{thm:RetMain} since it covers all irreflexive square-free graphs $H$, including those that have $3$-cycles. The proof of Theorem~\ref{thm:RetMain} is given in Section~\ref{sec:MainTheorem}.

\subsection{Second Contribution: Locating $\Ret{H}$ in the Approximate Counting Landscape} 
We locate the retraction counting problem in the complexity landscape of related homomorphism counting problems. Interestingly, it turns out that the complexity landscape for approximate counting looks very different from the one for exact counting.

We use $\calH(G,H)$ to denote the set of homomorphisms from $G$ to $H$ and $\hom{G}{H}$ to denote the size of $\calH(G,H)$. The following is the well-known homomorphism counting problem.

\prob
{
$\Hom{H}$.
}
{
An irreflexive graph $G$.
}
{
$\hom{G}{H}$.
}

Note that, in the problem $\Hom{H}$, the input graph~$G$ is
required to be irreflexive. This is standard in the field, and the reason for it is to
make results stronger --- typically it is the hardness results that are most challenging.
In the problem $\Ret{H}$, as we have informally defined it,
it does not make sense to force~$G$ to be irreflexive, since it contains an induced copy of~$H$, which may have loops.
However, we can insist that 
$G$   have no loops outside of the induced copy of~$H$.
Theorem~\ref{thm:RetMain} is still true  under this restriction, and 
we incorporate this restriction into our formal definitions below.
In order to give formal definitions, it is more natural to re-cast the retraction problem in terms of
list homomorphisms, so we define these next.

Let $\boldS=\{S_v\subseteq V(H) \mid v\in V(G)\}$ be a set of ``lists'' indexed by the vertices of $G$. 
Each list $S_v$ is a subset of $V(H)$.
We say that a function $h \from V(G) \to V(H)$ is a homomorphism from $(G,\boldS)$ to $H$ (also called a \emph{list homomorphism}) if $h$ is a homomorphism from $G$ to $H$ and, for each vertex $v$ of $G$, we have $h(v)\in S_v$. We use $\calH((G,\boldS),H)$ to denote the set of homomorphisms from $(G,\boldS)$ to $H$ and we use $\hom{(G,\boldS)}{H}$ to denote the size of $\calH((G,\boldS),H)$. We will be interested in the following generalisation of $\Hom{H}$.

\prob
{
$\LHom{H}.$
}
{
An irreflexive graph $G$ and a collection of lists $\boldS=\{S_v\subseteq V(H)\mid v\in V(G)\}$.
}
{
$\hom{(G,\boldS)}{H}$.
}
 
As noted earlier, we will find it convenient to formally define the computational problem $\Ret{H}$ in terms of list homomorphisms.

\prob
{
$\OALHom{H}$.
}
{
An irreflexive graph $G$ and a collection of lists $\boldS=\{S_v\subseteq V(H) \mid v\in V(G)\}$ such that, for all $v\in V(G)$, $\abs{S_v}\in \{1,\abs{V(H)}\}$.
}
{
$\hom{(G,\boldS)}{H}$.
}

The polynomial-time interreducibility between the problem $\Ret{H}$ which we defined informally 
(with the restriction on loops in~$G$)
and the one defined here, is demonstrated by Feder and Hell~\cite[Theorem 4.1]{FederLHomRefl} who give a parsimonious  reduction between them. (A reduction is parsimonious if it preserves the number of solutions.) This reduction also shows that the corresponding decision problems are polynomial-time interreducible.

We consider two more related counting problems, namely $\SHom{H}$, the problem of counting vertex-surjective homomorphisms, and $\Comp{H}$, the problem of counting edge-surjective homomorphisms, which are called \emph{compactions}. We give their formal definitions at the beginning of Section~\ref{sec:ComplexityLandscape}. Both problems are well-studied in the decision setting~\cite{BodirskySurvey, GolovachFinding, GolovachTrees, GolovachNewHardness, MartinPaulusmaSHomC4, VikasReflCycle, VikasIrreflCycle, Vikas2013, VikasHexagon}. All three of the problems $\SHom{H}$, $\Comp{H}$ and $\Ret{H}$ can be interpreted as problems requiring one to count homomorphisms with some kind of surjectivity constraint. $\LSHom{H}$ and $\LComp{H}$ are the corresponding list homomorphism problems and these are also formally defined in Section~\ref{sec:ComplexityLandscape}.

A \emph{separation} between two homomorphism-counting problems $A$ and $B$ is given by a parameter $H$ for which $A$ and $B$ are of different complexity, subject to some complexity-theory assumptions.

Before stating our results we give an overview of the approximate counting complexity landscape in Figure~\ref{fig:ApproxCountingLandscape}.
 The results summarised in this figure   are  consistent with 
 the results that are known
 concerning   the corresponding decision problems, as surveyed by Bodirsky, K\'{a}ra and Martin~\cite{BodirskySurvey}, 
 but they are in contrast to the situation in the exact counting world.
For exact counting, $\Hom{H}$, $\Ret{H}$ and $\LHom{H}$ are interreducible.
Also, all of the exact counting problems that we have mentioned
   reduce to $\Comp{H}$ and $\LComp{H}$~\cite{FGZ2017}, as depicted in Figure~\ref{fig:ExactCountingLandscape}. Moreover, $\Comp{H}$ and $\LComp{H}$ are separated from the remaining problems.

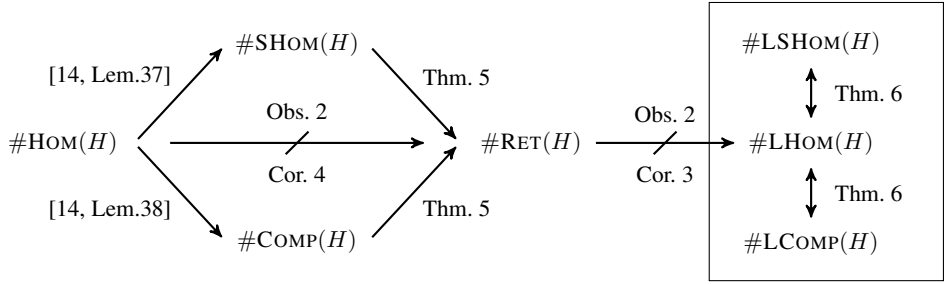
\begin{figure}[t]
\centering
\begin{minipage}[t]{\linewidth}
\hspace*{+0em}\raisebox{0em}{
\begin{tikzpicture}[node distance=.5cm, auto, scale=0.8, every node/.style={transform shape}]

 \node[cProb](Hom){$\Hom{H}$};
 
 \node[cProb,right= of Hom](dummy){\vphantom{(}};
 \node[cProb,above=1cm of dummy](SHom){$\SHom{H}$};
 \node[cProb,below=1cm of dummy](Comp){$\Comp{H}$};
 \node[cProb,right= of dummy](Ret){$\Ret{H}$};
 
 \node[cProb,right= 1.25cm of Ret](LHom){$\LHom{H}$};
 \node[cProb,above=1cm of LHom](LSHom){$\LSHom{H}$};
 \node[cProb,below=1cm of LHom](LComp){$\LComp{H}$};
 \draw ($(LHom.north west)+(0cm,2cm)$) rectangle ($(LComp.south east)+(.5cm,-.3cm)$);
 

 \draw[rightarr] ($(Hom.east)+(-0.5cm,0cm)$)--($(SHom.west)+(0.5cm,0cm)$) node [midway] {\cite[Lem.37]{FGZ2017}};
 \draw[seprightarr] ($(Hom.east)$)--($(Ret.west)$) node [midway,above=0.25cm] {Obs.~\ref{obs:HomToRetToLHom}} node [midway,below=0.25cm] {Cor.~\ref{cor:HomSepRet}};
 \draw[rightarr] ($(Hom.east)+(-0.5cm,0cm)$)--($(Comp.west)+(0.5cm,0cm)$) node [midway,below left] {\cite[Lem.38]{FGZ2017}};
 \draw[rightarr] ($(SHom.east)+(-0.5cm,0cm)$)--($(Ret.west)+(0.5cm,0cm)$)  node [midway] {Thm.~\ref{thm:SHomCompToRet}};
 \draw[rightarr] ($(Comp.east)+(-0.5cm,0cm)$)--($(Ret.west)+(0.5cm,0cm)$) node [midway,below right] {Thm.~\ref{thm:SHomCompToRet}};
 
 \draw[seprightarr] ($(Ret.east)+(-0.7cm,0cm)$)--($(LHom.west)+(0.5cm,0cm)$)  node [midway, above=0.2cm] {Obs.~\ref{obs:HomToRetToLHom}} node [midway,below=0.25cm] {Cor.~\ref{cor:RetSepLHom}};

 \draw[leftrightarr] (LSHom.south)--(LHom.north) node [midway, right=.25cm] {Thm.~\ref{thm:LHomEquivalent}};
 \draw[leftrightarr] (LHom.south)--(LComp.north) node [midway, right=.25cm] {Thm.~\ref{thm:LHomEquivalent}};
\end{tikzpicture}}
\end{minipage}
\caption{Approximate counting complexity landscape. An arrow from a  problem $A$ to a problem $B$ means that there exists an AP-reduction from $A$ to $B$. A struck through arrow corresponds to a reduction with a separation. The references for the reduction and the separation are given above and below the arrow, respectively.}
\label{fig:ApproxCountingLandscape}
\end{figure}

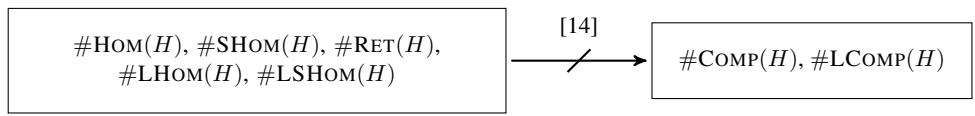
\begin{figure}[h!]
\centering
\begin{minipage}[t]{\linewidth}
\centering
\hspace*{+0em}\raisebox{0em}{
\begin{tikzpicture}[node distance=.5cm, auto, scale=0.8, every node/.style={transform shape}]

 \node[text width =19em, text centered](LHS){$\Hom{H}$, $\SHom{H}$, $\Ret{H}$, $\LHom{H}$, $\LSHom{H}$};
 \node[right=3cm of LHS](RHS){$\Comp{H}$, $\LComp{H}$};
 
 \draw ($(LHS.north west)+(-.3cm,.3cm)$) rectangle ($(LHS.south east)+(.3cm,-.3cm)$);
 \draw ($(RHS.north west)+(-.3cm,.3cm)$) rectangle ($(RHS.south east)+(.3cm,-.3cm)$);
 

 \draw[seprightarr] ($(LHS.east)+(.3cm,0cm)$)--($(RHS.west)+(-.3cm,0cm)$) node [midway, above= .25cm] {\cite{FGZ2017}};
\end{tikzpicture}}
\end{minipage}
\caption{Exact counting complexity landscape. All problems in the same box are interreducible with respect to polynomial-time Turing reductions. The arrow means that each problem in the box on the left-hand side reduces to each problem on the right-hand side using a polynomial-time Turing reduction. The arrow is struck through as there exists a separation between each problem on the left and each problem on the right.}
\label{fig:ExactCountingLandscape}
\end{figure}

We now give our results in more detail. We start off with the following simple observation which follows immediately from the problem definitions.
\begin{obs}\label{obs:HomToRetToLHom}
Let $H$ be a graph. Then $\Hom{H}\leap\OALHom{H}\leap\LHom{H}$.
\end{obs}

As we will see later, the complexity of approximately counting homomorphisms is still open (despite a lot of work on the problem) --- even if restricted to trees $H$. The complexity of approximately counting list homomorphisms is known, due to Galanis, Goldberg and Jerrum~\cite{GGJList}. Thus, 
Observation~\ref{obs:HomToRetToLHom} 
 indicates that
$\OALHom{H}$ 
is an important intermediate problem, between the solved $\LHom{H}$ and the wide-open $\Hom{H}$.

The first interesting consequence of Theorem~\ref{thm:RetMain} is a separation between $\Ret{H}$ and $\LHom{H}$.
\newcommand{\CorRetSepLHom}{
$\Ret{H}$ and $\LHom{H}$ are separated subject to the assumption that $\bis$ and $\sat$ are not AP-interreducible. In particular, if $H$ is a partially bristled reflexive path with at  least one unlooped vertex, then $\Ret{H} \eqap \bis$, whereas $\LHom{H} \eqap \sat$. 
}
\begin{cor}\label{cor:RetSepLHom}
\CorRetSepLHom
\end{cor}

The fact that $\Ret{H} \eqap \bis$ for partially bristled reflexive paths follows from Theorem~\ref{thm:RetMain}. The fact that $\LHom{H} \eqap \sat$ is from~\cite{GGJList}, see Theorem~\ref{thm:LHomTricho} in the Related Work section.

As a second consequence, Theorem~\ref{thm:RetMain} separates $\Ret{H}$ from $\Hom{H}$, but in a different sense. For $q\ge 3$ let $J_q$ be the irreflexive graph obtained from the $q$-leaf star by subdividing each edge. The graph $J_3$ is depicted in Figure~\ref{fig:J3} on page~\pageref{fig:J3}. From Goldberg and Jerrum~\cite{GJTreeHoms} it is known that the problem $\Hom{J_q}$ is AP-interreducible with the task of computing the partition function of the $q$-state ferromagnetic Potts model~\cite{Potts1952} --- a well-studied model from statistical physics. Despite extensive work on this problem~\cite{GJTreeHoms, GJ2012, Galanis2016} it is only known to be $\bis$-hard but 
is not   known to be $\bis$-easy or to  be $\sat$-hard (with respect to AP-reductions).
\begin{cor}\label{cor:HomSepRet}
Let $q$ be an integer with $q\ge 3$.
$\Hom{H}$ and $\Ret{H}$ are separated subject to the assumption that approximately computing the partition function of the $q$-state ferromagnetic Potts model is not $\sat$-hard. In particular, it follows from Theorem~\ref{thm:RetMain} that $\sat\leap \Ret{J_q}$.
\end{cor}

In addition to these separations, we show that approximately counting retractions is at least as hard as approximately counting surjective homomorphisms and also at least as hard as approximately counting compactions. The latter is surprising as it is in contrast to known results for the corresponding exact counting problems (see Figure~\ref{fig:ExactCountingLandscape}). Our proof uses an interesting Monte Carlo approach to AP-reductions and more details on this method are given in Section~\ref{sec:Methods}. The approach gives analogous reductions for the list versions of these problems for free.

\newcommand{\ThmSHomCompToRet}{
Let $H$ be a graph. Then $\SHom{H}\leap \Ret{H}$ and $\Comp{H}\leap \Ret{H}$.
}
\begin{thm}\label{thm:SHomCompToRet}
\ThmSHomCompToRet
\end{thm}

\newcommand{\ThmLHomEquivalent}{
Let $H$ be a graph. Then $\LSHom{H}\eqap \LHom{H}$ and $\LComp{H}\eqap \LHom{H}$.
}
\begin{thm}\label{thm:LHomEquivalent}
\ThmLHomEquivalent
\end{thm}

Using Theorem~\ref{thm:SHomCompToRet} and Corollary~\ref{cor:RetSepLHom} we can deduce that $\SHom{H}$ and $\LHom{H}$ are also separated subject to the assumption that $\bis$ and $\sat$ are not AP-interreducible. The same holds for $\Comp{H}$ and $\LHom{H}$. Moreover, from Theorem~\ref{thm:LHomEquivalent} it follows that we can replace the problem $\LHom{H}$ with $\LSHom{H}$ or $\LComp{H}$ in these separations.

Our reductions $\SHom{H}\leap \Ret{H}$ and $\Comp{H}\leap \Ret{H}$ allow us to state new $\bis$-easiness results which are not limited to graphs of girth at least $5$, namely the $\bis$-easiness results in the following corollary.

\newcommand{\CorBisEasiness}{
Let $H$ be one of the following:
\begin{itemize}
\item A reflexive proper interval graph but not a complete graph.
\item An irreflexive bipartite permutation graph but not a complete bipartite graph.
\end{itemize} 
Then $\SHom{H}$, $\Comp{H}$ and $\Ret{H}$ are $\bis$-equivalent.
}
\begin{cor}\label{cor:BisEasiness}
\CorBisEasiness
\end{cor}

The $\bis$-easiness results in Corollary~\ref{cor:BisEasiness} come from our Theorem~\ref{thm:SHomCompToRet} together with Observation~\ref{obs:HomToRetToLHom} and the $\bis$-easiness results for $\LHom{H}$ given in Theorem~\ref{thm:LHomTricho} on page~\pageref{thm:LHomTricho}. The corresponding $\bis$-hardness comes from~\cite[Theorem 35]{FGZ2017}).

The proof of Theorem~\ref{thm:RetMain} is given in Section~\ref{sec:MainTheorem}. Theorem~\ref{thm:SHomCompToRet} is proved in Section~\ref{sec:MonteCarlo} in the form of Corollaries~\ref{cor:CompToRet} and~\ref{cor:SHomToRet}. The proof of Theorem~\ref{thm:LHomEquivalent} is in Section~\ref{sec:AdditionalReductions}.

\subsection{Methods}\label{sec:Methods}
In the proof of Theorem~\ref{thm:RetMain} we use several different techniques. In the $\bis$-easiness proof for partially bristled reflexive paths (Lemma~\ref{lem:BristledPathEasiness}) we build upon a technique that was introduced by Dyer et al.~\cite{DGGJApprox} and extended by Kelk~\cite{KelkThesis} to reduce the problem of approximately  counting homomorphisms to the problem of approximately counting the downsets of a partial order. In order to obtain more general results, we formalise this technique and use it in the context of the constraint satisfaction framework. 
This framework is convenient 
for generating
$\bis$-easiness results, not only for counting homomorphisms but also for counting retractions, both in the setting of undirected graphs (as used in this work) and even in the setting of directed graphs.

In order to obtain the $\sat$-hardness part of Theorem~\ref{thm:RetMain}, we analyse, and classify, different local structures
in graphs. 
The  most difficult part is Lemma~\ref{lem:MixedTreeHardness3} which analyses distance-$2$ neighbourhood structures 
as defined in Section~\ref{sec:Preliminaries}. This lemma requires 
intricate gadgets 
(based on simpler versions used in~\cite{DGGJApprox} and~\cite{KelkThesis}), rather careful analysis, and classifying homomorphisms by type.
Other $\sat$-hardness results are easier to come by.
For instance, we use modifications of, and a more careful analysis of, a gadget from~\cite{GJTreeHoms} to prove the $\sat$-hardness for irreflexive square-free graphs that have an induced subgraph $J_3$ (Lemma~\ref{lem:SquareFreeHardness}). 
Some hardness results are also based on $\NP$-completeness results for the retraction decision problem that carry over to the approximate counting version. 

The algorithms captured by the reductions $\SHom{H}\leap\Ret{H}$ and $\Comp{H}\leap\Ret{H}$ are based on a Monte Carlo approach (Lemma~\ref{lem:MCAlgo}). 
We will discuss this approach here in the context of counting surjective homomorphisms
to~$H$. The details, and the  related approach for counting compactions, are described in Section~\ref{sec:MonteCarlo}.
The Monte Carlo approach is applicable 
for a reduction from $\SHom{H}$ to $\Ret{H}$
because $\Ret{H}$ is a so-called self-reducible problem.
Recall that the output of $\Ret{H}$, given a graph $G$ and lists 
$\boldS=\{S_v\subseteq V(H) \mid v\in V(G)\}$ such that, for all $v\in V(G)$, $\abs{S_v}\in \{1,\abs{V(H)}\}$,
is the number of homomorphisms from $(G,\boldS)$ to $H$.
Using an oracle for approximating this number,
it is also possible to \emph{sample} a homomorphism from $(G,\boldS)$ to $H$ approximately uniformly at random
(following the general method of Jerrum, Valiant, and  Vazirani~\cite{JVV1986}).
A naive  approach for sampling surjective homomorphisms (and hence for approximately counting them) is as follows:
Start with an input~$G$ to~$\SHom{H}$.  
Let $\boldS$ be the trivial set of lists
$\boldS=\{S_v = V(H) \mid v\in V(G)\}$.
Using the oracle for $\Ret{H}$, obtain a random homomorphism from $(G,\boldS)$ to~$H$
(which is just a random homomorphism from~$G$ to~$H$).
Reject (and repeat) if this homomorphism is not surjective. Eventually, we obtain a random surjective
homomorphism from $G$ to $H$, as required.
While this approach is certainly straightforward, it does not
lead to an \emph{efficient} algorithm for sampling surjective homomorphisms because  the number of surjective homomorphisms  might be very small compared to the total number of homomorphisms. 

Our method to shrink the sample space is based on the following fact. For every surjective homomorphism $h$ from $G$ to $H$ there exists a constant-size set of vertices $U\subseteq V(G)$ such that the restriction of $h$ to $U$ is already surjective.  
We can enumerate all these constant-size sets $U$ and use single vertex lists to fix their images. Consequently we obtain a (polynomial) number of instances $(G,\boldS^1), \ldots,(G,\boldS^k)$
of the problem $\Ret{H}$. For $i \in \{1,\dots, k\}$ let $R_i$ be the set of homomorphisms from $(G,\boldS^i)$ to~$H$.
Then the set of surjective homomorphisms from $G$ to $H$ is the union $R=\bigcup_{i=1}^k R_i$. The final building block of our reduction is the 
idea that we can sample the union $R$ 
by first sampling
from the disjoint union $R^+=\bigcup_{i=1}^k \{(h,i) \mid h\in R_i\}$. This idea is explained more generally, for instance, in~\cite[Section 11.2.2]{Mitzenmacher2017}. The point is, that we can sample uniformly from $R^+$ by using a $\Ret{H}$ oracle, and the union $R$ is relatively dense in the disjoint union $R^+$ (its size is at least $|R^+|/k$). 
So we can obtain a sample from~$R$. Then the samples can be combined to obtain, with high probability,
an approximate count.
 A lot of AP-reductions are based on the use of gadgets and we have not seen the use of Monte Carlo algorithms  in AP-reductions before.

\subsection{Related Work}\label{sec:RelatedWork}

Let $H$ be a graph. It is well-known that the complexity of $\LHom{H}$ is determined by the maximum complexity $\LHom{C}$ for a connected component $C$ of $H$. In the connected case the complexity is determined by the following theorem by Galanis et al.~\cite{GGJList}. 
\begin{thm}[\cite{GGJList}]\label{thm:LHomTricho}
Let $H$ be a connected graph. 
{
\renewcommand{\theenumi}{(\roman{enumi})}
\renewcommand{\labelenumi}{\theenumi}
\begin{enumerate}            
\item If $H$ is an irreflexive complete bipartite graph or a reflexive complete graph, then
$\LHom{H}$ is in $\FP$.
\item 
Otherwise, if $H$ is an irreflexive bipartite permutation graph or a reflexive proper interval graph, then
$\LHom{H}$ is $\bis$-equivalent under AP-reductions. 
\item Otherwise, $\LHom{H}$ is $\sat$-equivalent under AP-reductions.
\end{enumerate}
}
\end{thm}

From our Theorem~\ref{thm:RetMain} and Theorem~\ref{thm:LHomTricho} we immediately obtain the separation between $\Ret{H}$ and $\LHom{H}$ given in Corollary~\ref{cor:RetSepLHom}. Note that when restricting to irreflexive or reflexive graphs, the classification of $\Ret{H}$ for graphs of girth at least $5$ is identical to the corresponding classification of the problem $\LHom{H}$. A separation between these problems only occurs for graphs with at least one looped and one unlooped vertex.

The complexity of approximately counting homomorphisms in the absence of lists is still far from being resolved. Galanis, Goldberg and Jerrum~\cite{GGJBIS} give a dichotomy for the problem in terms of $\bis$.

\begin{thm}[\cite{GGJBIS}]\label{thm:HomBIS}
Let $H$ be a connected graph. If $H$ is a reflexive complete graph or an irreflexive complete bipartite graph, then $\Hom{H}$ admits an FPRAS. Otherwise, $\bis\leap\Hom{H}$.
\end{thm}

Surprisingly, even for the subclass of problems where $H$ is an irreflexive tree, the complexity of approximately counting homomorphisms is not completely classified. The following partial classification, originally due to Goldberg and Jerrum~\cite{GJTreeHoms}, follows from Theorems~\ref{thm:LHomTricho} and~\ref{thm:HomBIS}. 

\begin{thm}[\cite{GJTreeHoms}]\label{thm:HomTree}
Let $H$ be an irreflexive tree. 
{
\renewcommand{\theenumi}{\roman{enumi})}
\renewcommand{\labelenumi}{\theenumi}
\begin{enumerate}            
\item If $H$ is a star, then $\Hom{H}$ is in $\FP$.
\item Otherwise, if $H$ is a caterpillar, then $\Hom{H}$ is $\bis$-equivalent under AP-reductions.
\item Otherwise, $\Hom{H}$ is $\bis$-hard under AP-reductions.
\end{enumerate}
}
\end{thm}

Note that, in general, for irreflexive trees $H$ that are neither stars nor caterpillars it is open whether approximately counting homomorphisms is $\bis$-equivalent, $\sat$-hard or even none of the two. It is only known that $\bis$ AP-reduces to $\Hom{H}$. However, there exist trees for which $\Hom{H}$ is $\sat$-equivalent with respect to AP-reductions (see~\cite[Section 5]{GJTreeHoms}).

The decision version of the retraction problem is formally defined as follows.\footnote{
The literature is slightly inconsistent in the sense that the decision problem $\DRet{H}$ is often defined without the restriction that $G$ is irreflexive.
It is easy to see that the two versions (with and without the restriction) are polynomial-time interreducible
since a looped vertex $v$ of $G$ with list $S_v$ can be replaced with an irreflexive clique of size $\abs{V(H)}+1$ (all of whose members have list $S_v$) without changing whether or not there is a homomorphism to $H$. Thus, results stated for one version apply to the other.}

\prob
{
$\DRet{H}$. 
}
{
An irreflexive graph $G$ and a collection of lists $\boldS=\{S_v\subseteq V(H) \mid v\in V(G)\}$ such that, for all $v\in V(G)$, $\abs{S_v}\in \{1,\abs{V(H)}\}$.
}
{
Is $\hom{(G,\boldS)}{H}$ positive?
}

$\DRet{H}$ is completely classified as a result of the recent proof of the CSP dichotomy conjecture~\cite{BulatovCSPDichotomy, ZhukCSPDichotomy}.  
However, these proofs do not give a graph-theoretical characterisation. 
Feder, Hell, Jonsson, Krokhin and Nordh~\cite[Corollary 4.2, Theorem 5.1]{FederPseudoForest} give the following graph-theoretical characterisation for pseudotrees, where a pseudotree is a graph with at most one cycle. A graph $H$ is called \emph{loop-connected} if, for every connected component $C$ of $H$, the looped vertices in $C$ induce a connected subgraph of $C$.

\begin{thm}[\cite{FederPseudoForest}]\label{thm:DRetPseudotree}
Let $H$ be a pseudotree. Then $\DRet{H}$ is 
$\NP$-complete if any of  the following hold:
\begin{itemize}
\item $H$ is not loop-connected,
\item $H$  contains a cycle of size at least $5$,
\item $H$  contains a  reflexive cycle of size $4$ or
\item $H$  contains an   irreflexive cycle of size $3$.
\end{itemize} 
Otherwise $\DRet{H}$ is  in $\P$.
\end{thm}

Finally, the complexity of exactly counting retractions is completely classified. It is in $\FP$ if every connected component of $H$ is a reflexive complete graph or an irreflexive complete bipartite graph and $\numP$-complete otherwise~\cite{FGZ2017}.

\subsection{Preliminaries}\label{sec:Preliminaries}

For a positive integer $n$ let $[n]=\{1,\dots, n\}$.
As partially bristled reflexive paths appear in a number of our results we give a more formal definition of this class of graphs. We also give an example in Figure~\ref{fig:PBRP}.

\begin{defn}\label{def:PBRP}
A \emph{partially bristled reflexive path}
is a reflexive path, or a tree with the following form.
Let $Q$ be a positive integer and let $S$ be a non-empty subset of $[Q]$.
Then $V(H) = \{c_0,\ldots,c_{Q+1}\} \cup \bigcup_{i\in S} \{g_i\}$
and $E(H) = \bigcup_{i=0}^{Q} \{c_i,c_{i+1}\} \cup \bigcup_{i=0}^{Q+1} \{c_i,c_i\} \cup \bigcup_{i\in S} \{c_i,g_i\}$.
\end{defn}

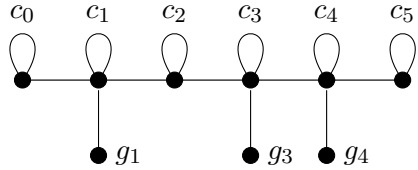
\begin{figure}[h!]\centering
{\def\scaleFactor{1}
\begin{tikzpicture}[scale=1, every loop/.style={min distance=10mm,looseness=10}]

			\filldraw (0,0) node(c0){} circle[radius=3pt] --++ (0:1cm) node(c1){} circle[radius=3pt] --++ (0:1cm) node(c2){} circle[radius=3pt] --++ (0:1cm) node(c3){} circle[radius=3pt] --++ (0:1cm) node(c4){} circle[radius=3pt] --++ (0:1cm) node(c5){} circle[radius=3pt];		
			\filldraw (c1) --++ (-90:1cm) node(g1){} circle[radius=3pt];
			\filldraw (c3) --++ (-90:1cm) node(g3){} circle[radius=3pt];
			\filldraw (c4) --++ (-90:1cm) node(g4){} circle[radius=3pt];
			
			\path[-] (c0.center) edge  [in=125,out=55,loop] node {} ();
			\path[-] (c1.center) edge  [in=125,out=55,loop] node {} ();
			\path[-] (c2.center) edge  [in=125,out=55,loop] node {} ();
			\path[-] (c3.center) edge  [in=125,out=55,loop] node {} ();		
			\path[-] (c4.center) edge  [in=125,out=55,loop] node {} ();
			\path[-] (c5.center) edge  [in=125,out=55,loop] node {} ();
	
			\node[above=.5cm of c0]{$c_0$};
			\node[above=.5cm of c1]{$c_1$};
			\node[above=.5cm of c2]{$c_2$};
			\node[above=.5cm of c3]{$c_3$};
			\node[above=.5cm of c4]{$c_4$};
			\node[above=.5cm of c5]{$c_5$};
			
			\node at ($(g1)+(.4cm,0)$) {$g_1$};
			\node at ($(g3)+(.4cm,0)$) {$g_3$};
			\node at ($(g4)+(.4cm,0)$) {$g_4$};
					
\end{tikzpicture}
}
\caption{Partially bristled reflexive path with $Q=4$ and $S=\{1,3,4\}$.}
\label{fig:PBRP}
\end{figure}

For a graph $H$ and a vertex $u\in V(H)$ we define the \emph{(distance-$1$) neighbourhood of $u$} as $\Gamma(u)=\{v\in V(H) \mid \{v,u\} \in E(H)\}$. Similarly, the \emph{distance-$2$ neighbourhood of $u$} is defined as $\Gamma^2(u)=\{v\in V(H) \mid \exists w \in V(H) : \{v,w\}, \{w,u\} \in E(H)\}$. Let $U$ be a subset of $V(H)$. Then $\Gamma(U)= \bigcap_{u\in U} \Gamma(u)$ is the \emph{set of common neighbours} of the vertices in $U$. 
The set of vertices that have a neighbour in $S$ 
is denoted 
by $\Phi(S)=\bigcup_{v\in S} \Gamma(v)$.

We will also use the concept of \emph{induced graphs}. Given a subset $U$ of $V(H)$, the graph $H[U]=(U,\{\{u_1,u_2\}\in E(H) \mid u_1,u_2\in U\})$ is called \emph{the subgraph of $H$ induced by $U$}. 
The graph
$H'=(V',E')$ is a \emph{subgraph} of $H=(V,E)$ if $V'\subseteq V$ and $E'\subseteq E$.

For the complexity theory part of our work we restate some standard definitions taken from~\cite[Definitions 11.1, 11.2, Exercise 11.3]{Mitzenmacher2017}.
A randomised algorithm gives an \emph{$(\epsilon,\delta)$-approximation}
for the value~$V$ if the output~$X$ of the algorithm satisfies $\Pr(|X-V| \leq \epsilon V) \geq 1-\delta$.
Slightly overloading the notation, an \emph{$(\epsilon,\delta)$-approximation}
for a problem $V$ is a randomised algorithm which, given an input~$x$ and parameters 
$\epsilon,\ \delta \in (0,1)$, outputs an $(\epsilon,\delta)$-approximation for~$V(x)$. A \emph{randomised approximation scheme} (RAS) for a problem $V$ is a $(\epsilon,1/4)$-approximation of $V$. A RAS is called \emph{fully polynomial} (FPRAS) if it runs in time that is polynomial in~$1/\epsilon$ and the size of the input~$x$.

Intuitively, an \emph{approximation-preserving reduction} (AP-reduction) from a problem~$A$ to a problem~$B$ is an algorithm that yields an FPRAS for~$A$ if it has access to an FPRAS for~$B$. We state the technical definition  from~\cite{DGGJApprox}. An approximation-preserving reduction from a problem~$A$ to a problem~$B$ is a probabilistic oracle Turing machine $M$ which takes as input an instance $x$ of $A$ and a parameter $\eps\in (0,1)$, and satisfies the following three properties: 1) Every oracle call made by $M$ is of the form $(y, \delta)$, where $y$ is an instance of $B$ and $\delta\in(0,1)$ with $1/\delta\in \text{poly}(\abs{x},1/\eps)$ specifies the precision of approximation. 2) The Turing machine $M$ is a RAS for $A$ whenever the oracle is a RAS for $B$. 3) The runtime of $M$ is polynomial in $\abs{x}$ and $1/\eps$.

Sometimes it is useful to switch between different notions of accuracy. To this end we use the following observation which follows immediately from the Taylor expansion of the exponential function.
\begin{obs}\label{obs:accuracy}
Let $\eps$ be in $(0,1)$. Then
$1+\eps\le e^\eps\le 1+2\eps$ and $1-\eps\le e^{-\eps}\le 1-\eps/2$.
\end{obs}

We conclude this section with some simple remarks regarding the connectivity of graphs when investigating the complexity of counting retractions. We will show that in the context of approximately counting retractions we can restrict to connected graph without loss of generality. We define the following problem which restricts the input to connected graphs.

\prob
{
$\cOALHom{H}$.
}
{
An irreflexive \emph{connected} graph $G$ and a collection of lists $\boldS=\{ S_v\subseteq  V(H) \mid v\in V(G) \}$
such that, for all $v\in V(G)$, $|S_v| \in \{1, |V(H)| \}$.
}
{
$\hom{(G,\boldS)}{H}$.
}

The following observation is well known.

\begin{obs}\label{obs:connectedOALHom}
Let $H$ be a graph. Then $\OALHom{H} \eqap \cOALHom{H}$.
\end{obs}
\begin{proof}
The fact that $\OALHom{H} \geap \cOALHom{H}$ is trivial. We now show that $\OALHom{H} \leap \cOALHom{H}$.
Let $(G,\boldS)$ be an instance of $\OALHom{H}$ and let $\eps\in (0,1)$ be the desired precision.  Let $C_1, \dots, C_k$ be the connected components of $G$. For each $i\in [k]$, let $\boldS_i = \{S_v \mid v\in V(C_i)\}$. Then $\hom{(G,\boldS)}{H} = \prod_{i=1}^k \hom{(C_i, \boldS_i)}{H}$. The algorithm which, for each $i\in [k]$, makes a $\cOALHom{H}$ oracle call with precision $\delta=\eps/k$ and input $(C_i, \boldS_i)$, and returns the product of outputs,   approximates $\hom{(G,\boldS)}{H}$ with the desired precision.
\end{proof}

\begin{remark}\label{rem:Connectivity}
Let $H$ be a graph with connected components $H_1, \dots, H_k$ and let $(G,\boldS)$ be an input to $\cOALHom{H}$. For $j\in [k]$ let $S^j_v = S_v \cap V(H_j)$ and let $\boldS^j=\{S^j_v \mid v\in V(G)\}$. Then, as $G$ is connected, it holds that
\[
\hom{(G, \boldS)}{H}= \sum_{j\in [k]} \hom{(G,\boldS^j)}{H_j}.
\]
Therefore, given an oracle for $\cOALHom{H_j}$ for each $j\in[k]$, we obtain an algorithm for $\cOALHom{H}$. By Observation~\ref{obs:connectedOALHom} this means that given an oracle for $\OALHom{H_j}$ for each $j\in[k]$, we obtain an algorithm for $\OALHom{H}$.

In the opposite direction, it is straightforward to see that for each $j\in [k]$ we have $\cOALHom{H_j}\leap \cOALHom{H}$ (and therefore  $\Ret{H_j}\leap \Ret{H}$). The details are as follows: Let $(G,\boldS)$ be an input to $\cOALHom{H_j}$. If all lists in $\boldS$ have size $1$, computing $\hom{(G,\boldS)}{H_j}$ is trivial. Otherwise we fix some vertex $v\in V(G)$ with $S_v=V(H_j)$. For each $u\in V(H)$ and $w\in V(G)$ we define 
\[
S^u_w = 
\begin{cases}
\{u\}, &\text{if $w=v$}\\
S_w, &\text{if $\abs{S_w}=1$}\\
V(H), &\text{otherwise}.
\end{cases}
\]
and $\boldS^u=\{S^u_w \mid w\in V(G)\}$. As $G$ is connected, a homomorphism from $G$ to $H$ maps all vertices of $G$ to the same connected component of $H$. Therefore, $\hom{(G,\boldS)}{H_j}=\sum_{u\in V(H_j)} \hom{(G,\boldS^u)}{H}$. This shows that $\abs{V(H_j)}$ calls to a $\cOALHom{H}$ oracle are sufficient to approximate the number of retractions to a component $H_j$ of $H$.
\end{remark}

\section{Approximately Counting Retractions to Graphs without short Cycles}
First we study the complexity of approximately counting retractions to graphs of girth at least $5$. We start off by restricting to irreflexive graphs in Section~\ref{sec:IrreflTrees}. The corresponding classification (Theorem~\ref{thm:RetIrreflexive}) is for irreflexive square-free graphs. Subsequently, in Section~\ref{sec:MixedTrees}, we consider graphs that have at least one loop. 

\subsection{Irreflexive Square-free Graphs}\label{sec:IrreflTrees} 

The goal of this section is to prove Theorem~\ref{thm:RetIrreflexive}.  The most difficult part is 
Lemma~\ref{lem:SquareFreeHardness}, which shows that,
if $H$ is a square-free graph containing an induced~$J_3$, then $\sat \leap \Ret{H}$.

The proof of Lemma~\ref{lem:SquareFreeHardness}
 generalises ideas from the proof of Lemma~3.6 of \cite{GJTreeHoms},
 so we start with some definitions from there. A \emph{multiterminal cut} of a graph $G$ with distinguished vertices $\alpha$, $\beta$ and $\gamma$ (called \emph{terminals}) is a set of edges $E'\subseteq E(G)$ that disconnects the terminals (i.e. ensures that there is no path in $(V(G),E(G)\setminus E')$ that connects any two distinct terminals). The \emph{size} of a multiterminal cut is its cardinality. 
 We consider the following computational problem.

\prob
{
$\TCut{3}$.
}
{
A connected irreflexive graph $G$ with $3$ distinct terminals $\alpha,\ \beta,\ \gamma \in V(G)$ and a positive integer $B$. The input has the property that every multiterminal cut has size at least $B$.
}
{
The number of size-$B$ multiterminal cuts of $G$ with terminals $\alpha$, $\beta$ and $\gamma$.
}

\begin{figure*}[ht]
	\centering
	
	\begin{minipage}[t]{0.45\textwidth}
		\centering
		\begin{tikzpicture}[scale=1, baseline=0.36cm, every loop/.style={min distance=10mm,looseness=10}]

\filldraw (0,0) node(w){} circle[radius=3pt] --++ (210:1cm) node(x0){} circle[radius=3pt] --++ (210:1cm) node(x1){} circle[radius=3pt];	
\filldraw (w.center) --++ (330:1cm) node(y0){} circle[radius=3pt] --++ (330:1cm) node(y1){} circle[radius=3pt];
\filldraw (w.center) --++ (90:1cm) node(z0){} circle[radius=3pt] --++ (90:1cm) node(z1){} circle[radius=3pt];	
	
			\node at ($(.3cm,.3cm)$) {$w$};
			\node at ($(z0) +(.3cm,.3cm)$) {$z_0$};
			\node at ($(z1) +(.3cm,.3cm)$) {$z_1$};
			\node at ($(x0) +(-.3cm,.3cm)$) {$x_0$};
			\node at ($(x1) +(-.3cm,.3cm)$) {$x_1$};
			\node at ($(y0) +(.3cm,.3cm)$) {$y_0$};
			\node at ($(y1) +(.3cm,.3cm)$) {$y_1$};
			
		\end{tikzpicture}
	\end{minipage}%
	\caption{The graph $J_3$.}
\label{fig:J3}
\end{figure*}
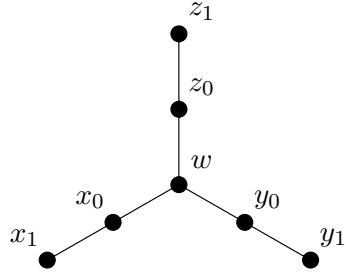

For motivation we consider the case where $H$ is a tree. Suppose that $H$ is an irreflexive tree with an induced~$J_3$, labelled as in Figure~\ref{fig:J3}.
Lemma~3.6 of \cite{GJTreeHoms} gives an AP-reduction
from $\TCut{3}$ to the problem of counting ``weighted'' homomorphisms to $H$.
In fact, the weights used in the proof are Boolean values, so the proof
actually reduces $\TCut{3}$ to the problem of counting list homomorphisms to $H$.
Given an input $G, \alpha, \beta, \gamma$ of $\TCut{3}$, 
 a homomorphism instance is created in which lists ensure that the terminals $\alpha$, $\beta$ and $\gamma$ are mapped to the vertices $x_0$, $y_0$ and $z_0$ of $J_3$, respectively. 
 Lists also ensure  that all other vertices of $G$ are mapped to $\{x_0,y_0,z_0\}$.
 Finally, lists ensure that all remaining vertices of the homomorphism
 instance are mapped to the vertices $\{w,x_1,y_1,z_1\}$ of $J_3$. 
 Our proof shows how to refine the gadgets so that lists have size~$1$ or size~$|V(H)|$.  
 Thus, our reduction is to the more refined problem $\Ret{H}$.
 We also show how to handle graphs~$H$ that are not trees (as long as they are square-free). 
 The details are given in the proof of Lemma~\ref{lem:SquareFreeHardness}. We use the following technical lemma,
 known as Dirichlet's approximation lemma, which 
 bounds the extent to which reals can be approximated by integers. Using this lemma is   a standard technique   in this line of research (see for instance~\cite{GGJBIS}).

\begin{lem}[{\cite[p. 34]{Schmidt1991}}] \label{lem:Dirichlet}
Let $\lambda_1, \dots, \lambda_d > 0$ be real numbers and $N$ be a natural number. Then there exist positive integers $p_1, \dots, p_d, r$ with $r \le N$ such that $\abs{r\lambda_i - p_i}\le 1/N^{1/d}$ for every $i\in [d].$
\end{lem}

Using Lemma~\ref{lem:Dirichlet}, we can prove our main lemma.

\begin{lem}\label{lem:SquareFreeHardness} Let $H$ be a square-free graph that contains an induced $J_3$. Then $\sat \leap\Ret{H}$.
\end{lem}
\begin{proof} Suppose that $H$ is a square-free graph with $q$ vertices and an induced $J_3$, which we label as shown in Figure~\ref{fig:J3}. 
The problem $\TCut{3}$ is shown 
in~\cite[Lemma 3.5]{GJTreeHoms} to be equivalent to $\sat$ with respect to AP-reductions.
We will give a reduction   from $\TCut{3}$ to $\Ret{H}$.

Let $G$, $\alpha$, $\beta$, $\gamma$, $B$ be an instance of $\TCut{3}$ with $n=\abs{V(G)}$ and let $\eps$ be an error bound in $(0,1)$. From this instance we construct an input $(J,\boldS)$ to $\Ret{H}$ as follows. 
Each of the terminals $\alpha$, $\beta$ and $\gamma$ will be vertices of $J$.
$J$ will also have a vertex~$\omega$ which is distinct from $\alpha$, $\beta$ and $\gamma$.
Let $\sx$, $\sy$ and $\sz$ be positive integers (we will give their precise values later).
 For every edge $e=\{u,v\}\in E(G)$ we define the set of vertices 
\[
V'(e)=\{(e,\alpha,1),\dots,(e,\alpha,\sx),(e,\beta,1),\dots,(e,\beta,\sy),(e,\gamma,1),\dots,(e,\gamma,\sz)\}.
\] 
The label ``$\alpha$'' in the name of the vertex $(e,\alpha,i)$ indicates that, in the instance $(J,\boldS)$, 
this vertex will be adjacent to~$\alpha$. The labels ``$\beta$'' and ``$\gamma$'' are similar.
The label ``$e$'' in the name of the vertex $(e,\alpha,i)$ indicates 
that this vertex is in $V'(e)$. 

For each  edge $e=\{u,v\}\in E(G)$ we then define a graph $J(e)$ with vertex set 
$V(J(e))=V'(e)\cup \{\alpha,\beta,\gamma,\omega\}$.
Note that the vertices in $V'(e)$ are distinct for each edge $e$ whereas $\alpha$, $\beta$, $\gamma$ and $\omega$ are identical for all $e$.
The edge set $E(J(e))$ of the graph $J(e)$ is defined as shown in Figure~\ref{fig:EdgeGadgetUnpinned}.

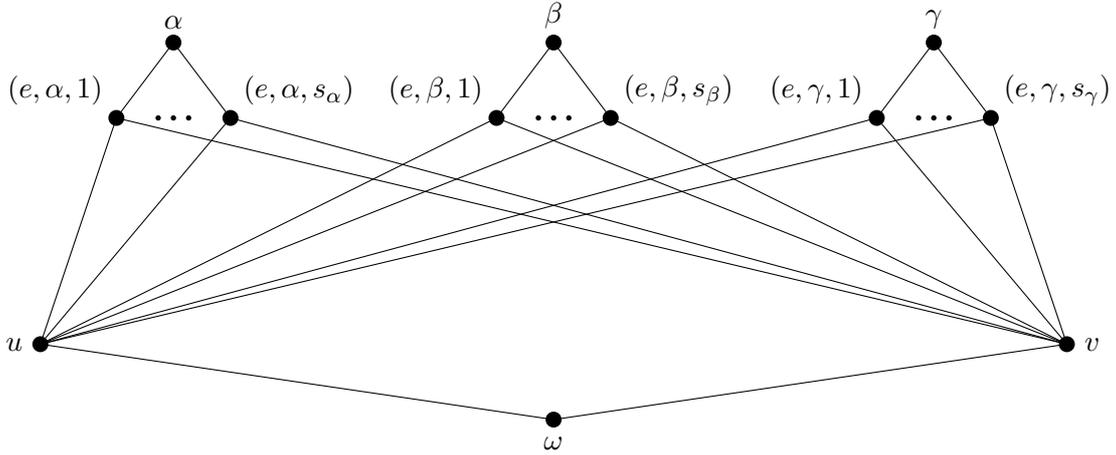
\begin{figure*}[ht]
	\centering
	
	\begin{minipage}[t]{\textwidth}
		\centering
		\begin{tikzpicture}[scale=1, baseline=0.36cm, every loop/.style={min distance=10mm,looseness=10}]

			\coordinate (u)  at (0, 0);
			\node at (u) [left = 1mm of u] {$u$};
			\coordinate (v)  at (13.5, 0);
			\node at (v) [right = 1mm of v] {$v$};
			
			\coordinate (w)  at (6.75, -1);
			\node at (w) [below = 1mm of w] {$\omega $};			
			
			\coordinate (ex1)  at (1, 3);
			\node at (ex1) [above left = .5mm of ex1] {$(e,\alpha,1)$};
			
			\coordinate (vdots1)  at (1.55, 3);
			\coordinate (vdots2)  at (1.75, 3);
			\coordinate (vdots3)  at (1.95, 3);
			
			\coordinate (exsx)  at (2.5, 3);
			\node at (exsx) [above right = .5mm of exsx] {$(e,\alpha,\sx)$};
			
			\coordinate (ey1)  at (6, 3);
			\node at (ey1) [above left = .5mm of ey1] {$(e,\beta,1)$};
			
			\coordinate (vdots4)  at (6.55, 3);
			\coordinate (vdots5)  at (6.75, 3);
			\coordinate (vdots6)  at (6.95, 3);
			
			\coordinate (eysy)  at (7.5, 3);
			\node at (eysy) [above right = .5mm of eysy] {$(e,\beta,\sy)$};
			
			\coordinate (ez1)  at (11, 3);
			\node at (ez1) [above left = .5mm of ez1] {$(e,\gamma,1)$};
			
			\coordinate (vdots7)  at (11.55, 3);
			\coordinate (vdots8)  at (11.75, 3);
			\coordinate (vdots9)  at (11.95, 3);
			
			\coordinate (ezsz)  at (12.5, 3);
			\node at (ezsz) [above right = .5mm of ezsz] {$(e,\gamma,\sz)$};
			
			\coordinate (x0)  at (1.75, 4);
			\node at (x0) [above = .5mm of x0] {$\alpha  $};			
			\coordinate (y0)  at (6.75, 4);
			\node at (y0) [above = .5mm of y0] {$\beta $};
			\coordinate (z0)  at (11.75, 4);
			\node at (z0) [above = .5mm of z0] {$\gamma $};

			\fill (u) circle[radius=3pt];
			\fill (v) circle[radius=3pt];
			\fill (w) circle[radius=3pt];
			\fill (ex1) circle[radius=3pt];
			\fill (exsx) circle[radius=3pt];
			\fill (ey1) circle[radius=3pt];
			\fill (eysy) circle[radius=3pt];
			\fill (ez1) circle[radius=3pt];
			\fill (ezsz) circle[radius=3pt];
			
			\fill (vdots1) circle[radius=1pt];
			\fill (vdots2) circle[radius=1pt];
			\fill (vdots3) circle[radius=1pt];
			\fill (vdots4) circle[radius=1pt];
			\fill (vdots5) circle[radius=1pt];
			\fill (vdots6) circle[radius=1pt];
			\fill (vdots7) circle[radius=1pt];
			\fill (vdots8) circle[radius=1pt];
			\fill (vdots9) circle[radius=1pt];
			
			\fill (x0) circle[radius=3pt];
			\fill (y0) circle[radius=3pt];
			\fill (z0) circle[radius=3pt];			
			
			\draw (u) -- (w) -- (v);
			\draw (u) -- (ex1) -- (v);
			\draw (u) -- (exsx) -- (v);
			\draw (u) -- (ey1) -- (v);
			\draw (u) -- (eysy) -- (v);
			\draw (u) -- (ez1) -- (v);
			\draw (u) -- (ezsz) -- (v);
			\draw (ex1) -- (x0) -- (exsx);
			\draw (ey1) -- (y0) -- (eysy);
			\draw (ez1) -- (z0) -- (ezsz);

		\end{tikzpicture}
	\end{minipage}%
	\caption{The graph $J(e)$ for $e=\{u,v\}$.  }
\label{fig:EdgeGadgetUnpinned}
\end{figure*}

Then we define $J=\left(V(G)\cup\{\omega\}\cup \bigcup_{e\in E(G)} V'(e),\bigcup_{e\in E(G)} E(J(e))\right)$. Intuitively, $J$ is 
constructed from the graph $G$
by replacing each edge $e\in E(G)$ with   the corresponding graph $J(e)$. 
Since $G$ is connected, every vertex $v\in V(G)$ is a member of some edge in $E(G)$.
This ensures that $\{v,\omega\}$ is an edge of~$J$. 
  
Next we define the set of lists $\boldS$ in the instance $(J,\boldS)$. 
We set $S_\omega =\{w\}$, $S_\alpha = \{x_0\}$, $S_\beta =\{y_0\}$, $S_\gamma =\{z_0\}$ and $S_v = V(H)$ for all $v\in V(J)\setminus\{\alpha, \beta, \gamma, \omega\}$. Then 
we define $\boldS = \{S_v \subseteq V(H) \mid v\in V(J)\}$.  
This is depicted in Figure~\ref{fig:EdgeGadget}.

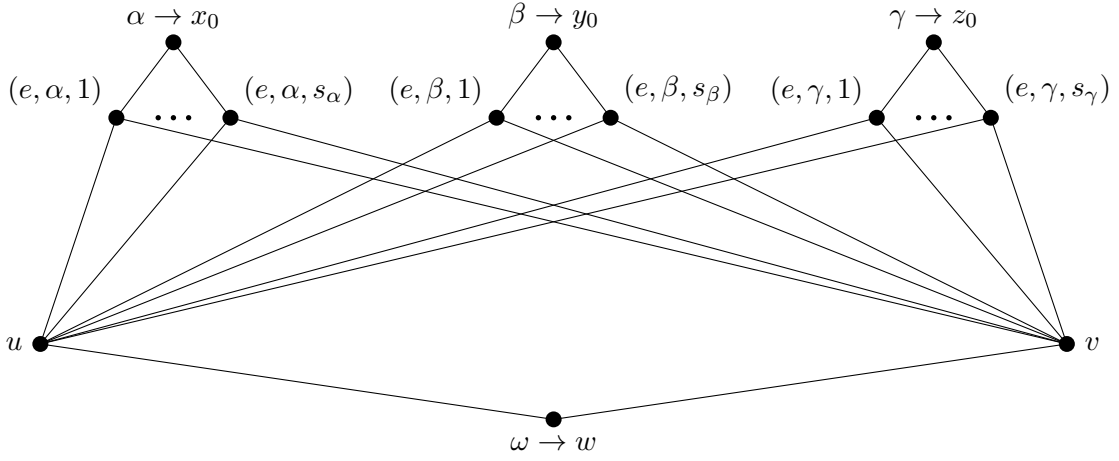
\begin{figure*}[ht]
	\centering
	
	\begin{minipage}[t]{\textwidth}
		\centering
		\begin{tikzpicture}[scale=1, baseline=0.36cm, every loop/.style={min distance=10mm,looseness=10}]

			\coordinate (u)  at (0, 0);
			\node at (u) [left = 1mm of u] {$u$};
			\coordinate (v)  at (13.5, 0);
			\node at (v) [right = 1mm of v] {$v$};
			
			\coordinate (w)  at (6.75, -1);
			\node at (w) [below = 1mm of w] {$\omega\rightarrow w$};			
			
			\coordinate (ex1)  at (1, 3);
			\node at (ex1) [above left = .5mm of ex1] {$(e,\alpha,1)$};
			
			\coordinate (vdots1)  at (1.55, 3);
			\coordinate (vdots2)  at (1.75, 3);
			\coordinate (vdots3)  at (1.95, 3);
			
			\coordinate (exsx)  at (2.5, 3);
			\node at (exsx) [above right = .5mm of exsx] {$(e,\alpha,\sx)$};
			
			\coordinate (ey1)  at (6, 3);
			\node at (ey1) [above left = .5mm of ey1] {$(e,\beta,1)$};
			
			\coordinate (vdots4)  at (6.55, 3);
			\coordinate (vdots5)  at (6.75, 3);
			\coordinate (vdots6)  at (6.95, 3);
			
			\coordinate (eysy)  at (7.5, 3);
			\node at (eysy) [above right = .5mm of eysy] {$(e,\beta,\sy)$};
			
			\coordinate (ez1)  at (11, 3);
			\node at (ez1) [above left = .5mm of ez1] {$(e,\gamma,1)$};
			
			\coordinate (vdots7)  at (11.55, 3);
			\coordinate (vdots8)  at (11.75, 3);
			\coordinate (vdots9)  at (11.95, 3);
			
			\coordinate (ezsz)  at (12.5, 3);
			\node at (ezsz) [above right = .5mm of ezsz] {$(e,\gamma,\sz)$};
			
			\coordinate (x0)  at (1.75, 4);
			\node at (x0) [above = .5mm of x0] {$\alpha \rightarrow x_0$};			
			\coordinate (y0)  at (6.75, 4);
			\node at (y0) [above = .5mm of y0] {$\beta\rightarrow y_0$};
			\coordinate (z0)  at (11.75, 4);
			\node at (z0) [above = .5mm of z0] {$\gamma\rightarrow z_0$};

			\fill (u) circle[radius=3pt];
			\fill (v) circle[radius=3pt];
			\fill (w) circle[radius=3pt];
			\fill (ex1) circle[radius=3pt];
			\fill (exsx) circle[radius=3pt];
			\fill (ey1) circle[radius=3pt];
			\fill (eysy) circle[radius=3pt];
			\fill (ez1) circle[radius=3pt];
			\fill (ezsz) circle[radius=3pt];
			
			\fill (vdots1) circle[radius=1pt];
			\fill (vdots2) circle[radius=1pt];
			\fill (vdots3) circle[radius=1pt];
			\fill (vdots4) circle[radius=1pt];
			\fill (vdots5) circle[radius=1pt];
			\fill (vdots6) circle[radius=1pt];
			\fill (vdots7) circle[radius=1pt];
			\fill (vdots8) circle[radius=1pt];
			\fill (vdots9) circle[radius=1pt];
			
			\fill (x0) circle[radius=3pt];
			\fill (y0) circle[radius=3pt];
			\fill (z0) circle[radius=3pt];			
			
			\draw (u) -- (w) -- (v);
			\draw (u) -- (ex1) -- (v);
			\draw (u) -- (exsx) -- (v);
			\draw (u) -- (ey1) -- (v);
			\draw (u) -- (eysy) -- (v);
			\draw (u) -- (ez1) -- (v);
			\draw (u) -- (ezsz) -- (v);
			\draw (ex1) -- (x0) -- (exsx);
			\draw (ey1) -- (y0) -- (eysy);
			\draw (ez1) -- (z0) -- (ezsz);

		\end{tikzpicture}
	\end{minipage}%
	\caption{The graph $J(e)$ for $e=\{u,v\}$. A label of the form $a \rightarrow b$ means that the vertex $a\in V(G)$ is pinned to $b\in V(H)$ since $S_a=\{b\}$.}
\label{fig:EdgeGadget}
\end{figure*}

We now show how a multiterminal cut of $G,\alpha,\beta,\gamma$
 corresponds to a certain set of homomorphisms from $(J,\boldS)$ to $H$. Every multiterminal cut $E'$ of 
 $G,\alpha,\beta,\gamma$ induces a partition of $V(G)$ into connected components. 
These are the connected components of $(V(G),E(G)\setminus E')$.
Let $\kappa(E')$ be the number of these components. Let $\Gamma(w)$ be the 
set of neighbours of vertex~$w$ in $H$ and let $d_w\ge 3$ be  the degree of~$w$ in~$H$. Let $\Psi(E')$ be the set of functions $\psi\from V(G)\to \Gamma(w)$ such that 
\begin{itemize}
\item $\psi$ maps the vertices of the components containing the terminals $\alpha$, $\beta$, 
and $\gamma$ to $x_0$, $y_0$ and $z_0$, respectively, and
\item  the set of bichromatic edges $\{\{u,v\}\in E(G) \mid \psi(u)\neq \psi(v)\}$ is exactly the cut $E'$. \end{itemize}
Then 
\begin{equation}
\label{eq:juneone}
\abs{\Psi(E')}={d_w}^{\kappa(E')-3}.
\end{equation}

Now, for every $\psi\in\Psi(E')$, let 
  $X(\psi)=\{\{u,v\}\in E(G) \mid \psi(u)=\psi(v)=x_0\}$. Note
  that $X(\psi)$ is   the set of monochromatic edges in $G$ whose endpoints are mapped to $x_0$. Similarly, let  $Y(\psi)=\{\{u,v\}\in E(G) \mid  \psi(u)=\psi(v)=y_0\}$  and $Z(\psi)=\{\{u,v\}\in E(G) \mid \psi(u)=\psi(v)=z_0\}$. 

Given a multiterminal cut~$E'$ of~$G,\alpha,\beta,\gamma$ 
and a map $\psi \in \Psi(E')$,
we say that a homomorphism $\sigma\in \calH((J,\boldS),H)$ \emph{agrees} with $\psi$ if, for all $v\in V(G)$, we have $\sigma(v)=\psi(v)$. Let $\Sigma(\psi)$ be the set of all $\sigma\in \calH((J,\boldS),H)$ that agree with $\psi$. 
Given a multiterminal cut~$E'$ of~$G,\alpha,\beta,\gamma$
we say that a homomorphism $\sigma\in \calH((J,\boldS),H)$ \emph{agrees} with~$E'$ 
if there is a $\psi\in \Psi(E')$ such that $\sigma$ agrees with~$\psi$.
Let $Z_{E'}=\sum_{\psi \in \Psi(E')}\abs{\Sigma(\psi)}$ be the number of homomorphisms from $(J,\boldS)$ to $H$ that agree with the cut $E'$.

Now  consider a multiterminal cut~$E'$ of~$G,\alpha,\beta,\gamma$. 
We will bound $Z_{E'}$ by considering two cases. Recall that $B$ is a positive integer and part of the instance of $\TCut{3}$.

\medskip
\noindent {\bf Case 1: $\abs{E'}=B$.}

If $\kappa(E')\ge 4$ then,  since $G$ is connected, 
the input $G,\alpha,\beta,\gamma$ has a multiterminal cut of size less than $B$, which contradicts the definition of $\TCut{3}$. 
Hence, it must be the case that $\kappa(E')=3$, which means that  $\abs{\Psi(E')}=1$. 
For the single $\psi\in \Psi(E')$ we have $\abs{X(\psi)}+\abs{Y(\psi)}+\abs{Z(\psi)}=\abs{E(G)}-B$. 
We will consider the possible homomorphisms $\sigma \in \Sigma(\psi)$.

Let $d_x, d_y,d_z \ge 2$ be the degrees of $x_0$, $y_0$ and $z_0$ in $H$, respectively. 
\begin{itemize}
\item Consider any edge $e$ of $G$ that  is not in the cut $E'$.
Then, as $\abs{E'}=B$, $e$ has to be in $X(\psi)$, $Y(\psi)$ or $Z(\psi)$. 
\begin{itemize}
\item Suppose that  $e$ is in $X(\psi)$. 
Consider a vertex $(e,\beta,i)$ of $J(e)$.
This vertex is adjacent to the terminal $\beta$, which is mapped to~$y_0$ by~$\psi$
and also to its endpoints, which are mapped to~$x_0$ by~$\psi$.
Thus, $\sigma$ has to map $(e,\beta,i)$ 
to a mutual neighbour (in $H$) of $x_0$ and $y_0$. Since $H$ is square-free,
the only possibility is vertex~$w$. 
Similarly, $\sigma$ has to map each vertex $(e, \gamma,i)$ of $J(e)$ to~$w$.
There is more choice concerning each vertex $(e,\alpha,i)$ of $J(e)$ ---
the homomorphism~$\sigma$ can map this vertex to any of the $d_x$ neighbours of~$x_0$ in~$H$.
Putting this together, 
the edge $e$ contributes a factor of ${d_x}^{\sx}$
to the number of homomorphisms in $\Sigma(\psi)$.
\item Suppose that $e$ is in $Y(\psi)$. Similarly, the
edge $e$ contributes a factor of ${d_y}^{\sy}$ to  $|\Sigma(\psi)|$.
\item Suppose that $e$ is in $Z(\psi)$.  Similarly, the edge $e$ contributes
a factor of ${d_z}^{\sz}$ to  $|\Sigma(\psi)|$.
\end{itemize}

\item Consider any edge $e=\{u,v\}$ of~$G$ that is in the cut $E'$.
 Then a homomorphism $\sigma\in \Sigma(\psi)$ has to map all vertices of $V'(e)$ to a common neighbour of 
 $\psi(u)$ and $\psi(v)$ in~$H$.
 Since $\psi(u) \neq \psi(v)$ and $H$ is square-free, the only possibility is to map all vertices of $V'(e)$ to~$w$.
 Thus, the edge $e$ contributes a factor of~$1$ to 
 $|\Sigma(\psi)|$. 
\end{itemize}
 Putting all of this together, we have
\[
Z_{E'}={d_x}^{\sx \abs{X(\psi)}}{d_y}^{\sy\abs{Y(\psi)}}{d_z}^{\sz\abs{Z(\psi)}}.
\]

Now our goal is to choose $\sx$, $\sy$ and $\sz$  so that $Z_{E'}$ depends only on the size of $E'$ rather than 
on the sizes of $X(\psi)$, $Y(\psi)$ and $Z(\psi)$. (Intuitively, we want to design our graph $J(e)$ in such a way that it balances out the weights that are induced by the different degrees of $x_0$, $y_0$ and $z_0$.) We are limited by the fact that $\sx$, $\sy$ and $\sz$ have to be integers. We use Lemma~\ref{lem:Dirichlet} to get around this. We set $\delta'=\log_q e^{\eps/42}$ which we will use in the error bound of the Dirichlet approximation (the reasons behind our choice of $\delta'$ will become clear at the end of the proof). Note that $1/{\delta'} \in \text{poly}(\eps^{-1})$. Further, let $s=2+\abs{E(G)} + \lceil\log_2q\rceil\abs{V(G)}$. We use Lemma~\ref{lem:Dirichlet} to approximate the real values $\lambda_1=\log_{d_x}(2^s)$, $\lambda_2=\log_{d_y}(2^s)$ and $\lambda_3=\log_{d_z}(2^s)$ and obtain positive integers $p_1, p_2, p_3$ and $r$ with $r \le (n^2/{\delta'})^3\in \text{poly}(n, \eps^{-1})$ such that for all $i\in \{1,2,3\}$ we have $\abs{r\lambda_i-p_i} \le \delta'/n^2$. Note that $p_1, p_2, p_3 \in \text{poly}(n, \eps^{-1})$. We set $\sx=p_1$, $\sy=p_2$ and $\sz=p_3$ to obtain
\begin{align*}
Z_{E'}&={d_x}^{p_1\abs{X(\psi)}}{d_y}^{p_2\abs{Y(\psi)}}{d_z}^{p_3\abs{Z(\psi)}}\\
&\le {d_x}^{(r\lambda_1 + \delta'/n^2) \abs{X(\psi)}}{d_y}^{(r\lambda_2 + \delta'/n^2)\abs{Y(\psi)}}{d_z}^{(r\lambda_3 + \delta'/n^2)\abs{Z(\psi)}}\\
&\le 2^{sr(\abs{X(\psi)} + \abs{Y(\psi)} + \abs{Z(\psi)})} q^{\delta'/n^2(\abs{X(\psi)} + \abs{Y(\psi)} + \abs{Z(\psi)})}
\end{align*}
where we used the fact that $d_x,d_y,d_z \le q$.
Since $\abs{X(\psi)}+\abs{Y(\psi)}+\abs{Z(\psi)}=\abs{E(G)}-B\le n^2$ it holds that 
\[
Z_{E'}\le q^{\delta'} 2^{sr(\abs{E(G)}-B)}.
\]
Analogously we obtain $q^{-\delta'} 2^{sr(\abs{E(G)}-B)}\le Z_{E'}$. 

Let $Z^*= 2^{sr(\abs{E(G)}-B)}$ (this value will be used later in the proof). Then
\begin{align}
q^{-\delta'} Z^*\le Z_{E'} \le q^{\delta'} Z^*.\label{eq:ZE1}
\end{align}
{(\bf End of Case 1.)}

\medskip
\noindent {\bf Case 2: $\abs{E'} > B$.} 

In this case we have $\kappa(E')\ge 3$. 
For  any $\psi\in\Psi(E')$, as in Case 1, 
each edge in $X(\psi)$
contributes a factor of $d_x^{\sx}$
to $|\Sigma(\psi)|$, each edge in $Y(\psi)$
contributes a factor of  $d_y^{\sy}$, 
and each edge in $Z(\psi)$ contributes a factor of  $d_z^{\sz}$. 

Consider any $\psi \in \Psi(E')$ and
let $E''=E(G)\setminus(X(\psi)\cup Y(\psi) \cup Z(\psi))$. 
Then $E''$ consists of edges in   $E'$ and edges in $\{\{u,v\}\in E(G) \mid \psi(u)=\psi(v) \text{ and } \psi(u)\notin \{x_0,y_0,z_0\}\}$.
Any edge $e$ in $E''$ contributes a factor of~$1$ to $|\Sigma(\psi)|$
as every vertex in $V'(e)$ has to be mapped to~$w$.
Putting all of this together and simplifying as in Case~1, we have 

\begin{align*}
Z_{E'}&=\sum_{\psi \in \Psi(E')}{d_x}^{p_1\abs{X(\psi)}}{d_y}^{p_2\abs{Y(\psi)}}{d_z}^{p_3\abs{Z(\psi)}}\\
&\le \sum_{\psi \in \Psi(E')} q^{\delta'} 2^{sr(\abs{X(\psi)} + \abs{Y(\psi)} + \abs{Z(\psi)})}.
\end{align*} 
We can analogously derive a lower bound for $Z_{E'}$ to obtain
\begin{align}
q^{-\delta'}  \sum_{\psi \in \Psi(E')} 2^{sr(\abs{X(\psi)} + \abs{Y(\psi)} + \abs{Z(\psi)})} \le Z_{E'} \le q^{\delta'}  \sum_{\psi \in \Psi(E')} 2^{sr(\abs{X(\psi)} + \abs{Y(\psi)} + \abs{Z(\psi)})}\label{eq:ZE2}
\end{align} 
{(\bf End of Case 2.)}
\medskip

Let $\calM$ denote the set of multiterminal cuts of $G$ with terminals
$\alpha$, $\beta$ and $\gamma$ and
let $T$ be the number of multiterminal cuts in $\calM$
with size~$B$. 
We would like to show how to estimate~$T$ using an approximation 
for the number of homomorphisms from $(J,\boldS)$ to~$H$. 
Towards this end, define~$Z$ as follows. 
\[
\displaystyle Z=TZ^* + \sum_{E'\in\calM: \abs{E'}> B}\enspace  \sum_{\psi \in \Psi(E')} 2^{sr(\abs{X(\psi)} + \abs{Y(\psi)} + \abs{Z(\psi)})}.
\]
The proof is in two parts.

\medskip \noindent{\bf Part 1: 
 We show that $Z/Z^*\in [T,T+1/4]$. }

Since $\abs{X(\psi)} + \abs{Y(\psi)} + \abs{Z(\psi)}\le \abs{E(G)}-\abs{E'}$, we have
\[
Z \le TZ^* + \sum_{E'\in\calM: \abs{E'}> B} \enspace\sum_{\psi \in \Psi(E')} 2^{sr(\abs{E(G)}-\abs{E'})}.
\]
Using $\abs{\Psi(E')}={d_w}^{\kappa(E')-3}$ 
(from \eqref{eq:juneone})
and the definition of $Z^*$ 
(just before \eqref{eq:ZE1}),
we obtain
\[
Z \le TZ^* + \sum_{E'\in\calM: \abs{E'}> B} {d_w}^{\kappa(E')-3} \frac{Z^*}{2^{sr(\abs{E'}-B)}}.
\] 
Then, in the following expression, the first inequality follows from the definition of $Z$ and the second inequality follows from the fact that there are at most $2^{\abs{E(G)}}$ multiterminal cuts and from the bounds ${d_w}^{\kappa(E')-3}\le q^n$, $\abs{E'}-B\ge 1$ and $r\ge 1$. The third inequality follows from the choice of $s$.
\[
T\le \frac{Z}{Z^*} \le T+\frac{2^{\abs{E(G)}}q^n}{2^s}\le T + 1/4.
\]
We have verified that $Z/Z^*\in [T,T+1/4]$.

\medskip \noindent{\bf Part 2:
 We show that we can obtain a close approximation to $Z/Z^*$ using an oracle for approximating $\Ret{H}$.}

Recall that $\hom{(J,\boldS)}{H}$
is the number of homomorphisms from $(J,\boldS)$ to~$H$ and
note that
$$ \hom{(J,\boldS)}{H} = \sum_{E'\in\calM : \abs{E'}= B}Z_{E'} + \sum_{E'\in\calM : \abs{E'}> B} Z_{E'}.$$ 

Using Inequalities~(\ref{eq:ZE1}) and~(\ref{eq:ZE2}), and the fact that  $G,\alpha,\beta,\gamma$ has $T$ multiterminal cuts of size $B$, we have
\[
q^{-\delta'}Z\le \hom{(J,\boldS)}{H} \le q^{\delta'} Z.
\]
Let $\hat{Q}$ be a solution returned by the $\Ret{H}$ oracle when called with input $\bigl((J,\boldS),\eps/42\bigr)$. Then
\[e^{-\eps/42}q^{-\delta'} Z\le e^{-\eps/42} \hom{(J,\boldS)}{H} \le \hat{Q} \le e^{\eps/42} \hom{(J,\boldS)}{H} \le e^{\eps/42}q^{\delta'} Z.\]
The choice of $\displaystyle\delta'=\log_q e^{\eps/42}$ yields $e^{-\eps/21} \frac{Z}{Z^*} \le \frac{\hat{Q}}{Z^*} \le e^{\eps/21}\frac{Z}{Z^*}$. 
Note that $Z^*$ is easy to compute. The fact that this precision in the approximation of $Z$ suffices to obtain the required accuracy of the output $\hat{Q}/Z^*$ as an approximation of $T$ is derived in~\cite[Proof of Theorem 3]{DGGJApprox}.
\end{proof}

We can now give a classification of $\Ret{H}$ for irreflexive square-free graphs.

\newcommand{\ThmRetIrreflexive}{
Suppose that $H$ is an irreflexive square-free graph.
{
\renewcommand{\theenumi}{(\roman{enumi})}
\renewcommand{\labelenumi}{\theenumi}
\begin{enumerate}            
\item \label{item:SquarefreeFP}  If every connected component of $H$ is a star, then $\Ret{H}$ is in $\FP$.
\item \label{item:SquarefreeBIS} Otherwise, if every connected component of $H$ is a caterpillar, then $\Ret{H}$ approximation-equivalent to $\bis$.
\item \label{item:SquarefreeSAT} Otherwise, $\Ret{H}$ approximation-equivalent to $\sat$.
\end{enumerate}
}
}
\begin{thm}\label{thm:RetIrreflexive}
\ThmRetIrreflexive
\end{thm}
\begin{proof}
We first give the classification assuming that $H$ is a connected graph. Then we use Remark~\ref{rem:Connectivity} to recover the full classification. 

Suppose that $H$ is a connected irreflexive square-free graph. We have $\Hom{H}\leap\Ret{H}$ and $\Ret{H}\leap \LHom{H}$ by Observation~\ref{obs:HomToRetToLHom}.
Therefore, 
$\Ret{H}$ inherits hardness results from $\Hom{H}$ hardness results and it inherits easiness results from $\LHom{H}$ easiness results. 
Thus, since a star is a complete bipartite graph,
item~\ref{item:SquarefreeFP} follows from Theorem~\ref{thm:LHomTricho}.
Since a square-free graph that is not a star cannot be a complete bipartite graph, the $\bis$-hardness part of item~\ref{item:SquarefreeBIS} follows from Theorem~\ref{thm:HomBIS}.
It is known that
a caterpillar is
a bipartite permutation graph~\cite{Koehler1999} (see also~\cite[Appendix A]{GGJList}),
so the $\bis$-easiness part of item~\ref{item:SquarefreeBIS} follows again from Theorem~\ref{thm:LHomTricho}.
Theorem~\ref{thm:LHomTricho} also implies that $\Ret{H}$ is always $\sat$-easy, giving the easiness result in item~\ref{item:SquarefreeSAT}.

It remains to show the hardness result in item~\ref{item:SquarefreeSAT}. If $H$ is not a caterpillar, then it  contains either a cycle or an induced $J_3$~\cite[Theorem 1]{Harary1971}.

\bigskip
\noindent{\bf Case 1: $H$ contains an induced $J_3$.}
The fact that $\sat \leap \Ret{H}$ follows from Lemma~\ref{lem:SquareFreeHardness}.
{\bf End of Case 1.}

\bigskip
\noindent{\bf Case 2: $H$ contains a cycle.}
\begin{itemize}
\item Suppose that $H$ contains a cycle of odd length. Then $H$ is not bipartite and even the problem of deciding whether there exists a homomorphism to $H$ is $\NP$-complete due to Hell and Ne\v{s}et\v{r}il~\cite{HNOriginal}. This homomorphism decision problem reduces to the retraction decision problem $\DRet{H}$~\cite{BodirskySurvey} and therefore $\DRet{H}$ is $\NP$-hard as well. Then, $\NP$-hardness of $\DRet{H}$ implies $\sat$-hardness of the corresponding approximate counting problem $\Ret{H}$ by~\cite[Theorem 1]{DGGJApprox}.
\item Suppose that $H$ contains exactly one cycle of even length. Then $H$ is a pseudotree and, as $H$ is square-free, the cycle has length at least $6$.  Therefore $\DRet{H}$ is $\NP$-complete by Theorem~\ref{thm:DRetPseudotree}. Then, as before, it follows that $\Ret{H}$ is $\sat$-hard under AP-reductions by~\cite[Theorem 1]{DGGJApprox}.
\item Suppose that $H$ contains at least $2$ cycles and all cycles in $H$ have even length. We will show that $H$ contains an induced $J_3$ and therefore is covered by Case~1. Let $C$ be a shortest cycle in $H$. As there are at least two cycles in $H$ and $H$ is connected, there exists a path $P_1=w,z_0,z_1$ such that $w$ is in $C$ and $z_0$ is not in $C$. As $C$ has length at least $6$, there exists a path $P_2$ in $C$ of the form $x_0,x_1,w,y_0,y_1$. As $z_0$ is not in $C$ it does not coincide with any of the vertices of $P_2$. Further, as $C$ is a shortest cycle, $z_1$ cannot coincide with any of the vertices of $P_2$. Therefore the vertices of $P_1$ and $P_2$ form a graph $J_3$ as shown in Figure~\ref{fig:J3}. This subgraph $J_3$ is induced as $H$ does not contain any cycles of length less than $6$.
\end{itemize}
{\bf End of Case 2.}

The theorem now follows easily by Remark~\ref{rem:Connectivity}. If every connected component is easy, so is $H$. If any connected component is hard, so is $H$.
\end{proof}

\subsection{Graphs with Loops}\label{sec:MixedTrees}
In this section we consider graphs that are not irreflexive.

\subsubsection{$\bis$-Easiness Results for Graphs with Loops}\label{sec:MixedTreesEasyness}

The point of this section is to prove the following lemma.
\newcommand{\LemRetBristledPath}{
Let $H$ be a partially bristled reflexive path with at least $3$ vertices. Then $\Ret{H} \eqap \bis$.
}
\begin{lem}\label{lem:RetBristledPath}
\LemRetBristledPath
\end{lem}

This lemma builds on Kelk~\cite[Appendix A.8]{KelkThesis},
who shows $\Hom{H} \eqap \bis$ for partially bristled reflexive paths~$H$.
Thus, our work in this section is generalising Kelk's work
from homomorphism-counting to retraction-counting.
For us, the main interest is actually that we manage to classify \emph{all} graphs of girth at least $5$, rather
than that we show that these particular graphs are $\bis$-equivalent.
Nevertheless, partially bristled reflexive paths
allow us to explore some interesting ideas, providing a convenient setting
for generalising useful techniques.

In particular, in order to reduce $\Ret{H}$ to $\bis$, we generalise a technique that
was  introduced by Dyer et al.~\cite[Lemma 8]{DGGJApprox} in order to reduce homomorphism-counting problems to $\bis$.
Although the graphs $H$ that we consider in this work are undirected, we show that the technique also applies to directed graphs. We expect this to be useful for future work.\footnote{The technique
also  applies if the input~$G$ to $\Ret{H}$ 
is allowed to have loops.
This is the main observation needed to show that Theorem~\ref{thm:RetMain} extends to the setting where $G$ might have loops.
}
A homomorphism from a digraph~$G$ to a digraph~$H$ is simply a function $h\from V(G) \to V(H)$ such that, for all
$(u,v) \in E(G)$, the image $(h(u),h(v))$ is in $E(H)$.
A homomorphism from $(G,\boldS)$ to~$H$ must satisfy $h(v) \in S_v$, as in the undirected case.
As for undirected graphs, we use 
$ \hom{(G,\boldS)}{H}$ to denote the number of homomorphisms from
$(G,\boldS)$ to~$H$.
Thus, we consider the following  directed retraction problem.

\prob
{
$\DirRet{H}$.
}
{
An irreflexive digraph $G$ and a collection of lists $\boldS=\{ S_v\subseteq  V(H) \mid v\in V(G) \}$
such that, for all $v\in V(G)$, $|S_v| \in \{1, |V(H)| \}$.
}
{
$\hom{(G,\boldS)}{H}$.
}
 
The main method used in the literature to prove $\bis$-easiness of
approximate homomorphism-counting problems is to reduce them to 
the problem of counting the downsets of a partial order, which is
known to be $\bis$-equivalent~\cite{DGGJApprox}.  
In order
to obtain more general results, we formalise the technique introduced in the proof of~\cite[Lemma 8]{DGGJApprox} and expanded by Kelk~\cite{KelkThesis}, and use it in the context of the 
constraint satisfaction framework. 
Let  $\calL$ 
be a set of  Boolean relations 
(called a constraint language).
The counting constraint satisfaction problem (CSP) with parameter~$\calL$ is defined as follows.

\prob
{
$\csp(\calL)$.
}
{
A set of variables $X$ and a set of constraints $C$, where each constraint applies a relation from~$\calL$ to a list of variables from~$X$.
}
{
The number of assignments $\sigma\from X\to \{0,1\}$ that satisfy all constraints in $C$.
} 
 
The constraint language that we will use consists of 
the two unary  Boolean relations  
$\delta_0=\{(0)\}$ and $\delta_1=\{(1)\}$ 
and the   arity-two Boolean relation  $\Imp=\{(0,0), (0,1), (1,1)\}$.
Note that the constraint $\delta_0(x)$ forces a satisfying assignment to 
assign the value~$0$ to the variable~$x$
and the constraint $\delta_1(x)$ forces a satisfying assignment to 
assign the value~$1$ to~$x$.
The constraint $\Imp(x,y)$  
ensures that, in any satisfying assignment~$\sigma$, we have
$\sigma(x) \implies \sigma(y)$ (that is, if $\sigma(x)=1$ then $\sigma(y)=1$). It is known that 
the counting constraint satisfaction problem is $\bis$-equivalent when the constraint language contains
(exactly) these three relations.

\begin{lem} \cite[Theorem 3]{DGJBooleanCSP}
\label{lem:CSPImpliesEqBIS}
$\csp(\{ \Imp, \delta_0, \delta_1\}) \eqap \bis$.
\end{lem}

We now formalise the  downsets reduction technique from~\cite[Lemma 8]{DGGJApprox} and state it as a 
technique for reducing homomorphism-counting problems
to $\csp(\{ \Imp, \delta_0, \delta_1\})$.
We generalise the original technique in two ways.
First, we allow size-$1$ and size-$|V(H)|$ lists in the input,
so we obtain $\bis$-easiness results for $\Ret{H}$ and not merely for $\Hom{H}$. Second, even though the main focus of this work is on undirected graphs,
we set up the machinery to enable (stronger) $\bis$-easiness results for the directed problem $\DirRet{H}$.
  
The main idea is as follows.
Given \emph{any} instances 
$\IIv$, $\IIe$, $\IIf$ and $\IIb$ 
of $\csp(\{\Imp\})$ 
on a variable set~$X$ we will define (Definition~\ref{def:Hve})
an undirected graph~$\Hve$ and (Definition~\ref{def:dirHve})
a digraph~$\Hvfb$.
Then Lemma~\ref{lem:OALHomToCSPImplies}
will show that the problems
$\Ret{\Hve}$
and $\DirRet{\Hvfb}$ both reduce to the $\bis$-easy problem
$\csp(\{ \Imp, \delta_0, \delta_1\})$.
Finally, to prove the $\bis$-easiness of $\Ret{H}$ when
$H$ is a partially bristled reflexive path 
(in order to achieve our goal of proving  Lemma~\ref{lem:RetBristledPath}), 
we have to show, given a partially bristled reflexive path~$H$,
how to set up the corresponding instances 
$\IIv$ and $\IIe$
of $\csp(\{\Imp\})$ so
that $\Hve = H$.

Before defining the graph $\Hve$ and the digraph $\Hvfb$,
it helps to explain the notation.
The subscript ``v'' stands for ``vertex''
and the $\csp(\{\Imp\})$ instance $\IIv$
is used to define the vertices of the graph~$\Hve$ and the vertices of the digraph~$\Hvfb$.
The subscript ``e'' stands for ``edge'' and
the CSP instance $\IIe$ is used to define the edges of $\Hve$.
The instance $\IIf$   gives the ``forward'' constraints for each directed edge of $\Hvfb$ and 
the instance $\IIb$ gives the corresponding ``backward'' constraints.
We will use $\Cv$, $\Ce$, $\Cf$, and $\Cb$ to
denote the  constraint sets
of the instances $\IIv$, $\IIe$, $\IIf$, and $\IIb$,  respectively.

\begin{defn}\label{def:Hve}
Let $\IIv=(X,\Cv)$ and $\IIe=(X,\Ce)$ be instances of $\csp(\{\Imp\})$.
We define the   undirected graph $\Hve$ as follows. 
The vertices of $\Hve$ are the satisfying assignments of $\IIv$. 
Given any  assignments $\sigma$ and $\sigma'$ in $V(\Hve)$, there is an edge $\{\sigma, \sigma'\}$ in $\Hve$ if and only if 
the following holds:
For every constraint $\Imp(x,y)$ in $\IIe$,
we have  $\sigma(x) \Rightarrow \sigma'(y)$  and $\sigma'(x) \Rightarrow \sigma(y)$.
\end{defn}

The definition of the digraph $\Hvfb$ is similar.

\begin{defn}\label{def:dirHve}
Let $\IIv=(X,\Cv)$, $\IIf=(X,\Cf)$ and $\IIb = (X,\Cb)$ 
be instances of $\csp(\{\Imp\})$.
We define the directed graph $\Hvfb$ as follows.
The vertices of $\Hvfb$ are the satisfying assignments of $\IIv$. 
Given any  
assignments $\sigma$ and $\sigma'$ in $V(\Hvfb)$,
there is a (directed) edge $(\sigma, \sigma')$ in $\Hvfb$ if and only if 
the following holds:
\begin{itemize}
\item For every constraint $\Imp(x,y)$ in $\IIf$,
we have  $\sigma(x) \Rightarrow \sigma'(y)$, and
\item for every constraint $\Imp(x,y)$ in $\IIb$,
we have $\sigma'(x) \Rightarrow \sigma(y)$.
\end{itemize} 
\end{defn}

\begin{lem}\label{lem:OALHomToCSPImplies}
Let 
$\IIv=(X,\Cv)$, $\IIe=(X,\Ce)$, $\IIf=(X,\Cf)$ and $\IIb = (X,\Cb)$ 
be instances of $\csp(\{\Imp\})$.
 Then $\Ret{\Hve} \leap \csp(\{ \Imp, \delta_0, \delta_1\})$ and $\DirRet{\Hvfb} \leap \csp(\{ \Imp, \delta_0, \delta_1\})$.
\end{lem}

\begin{proof}

{\bf Undirected case: }We first show the 
reduction from $\Ret{\Hve}$ to $\csp(\{ \Imp, \delta_0, \delta_1\})$  and extend this to the directed result afterwards. The reductions we show are parsimonious.
From an instance $(G,\boldS)$ 
of $\Ret{\Hve}$
we create an instance $I$ of $\csp(\{ \Imp, \delta_0, \delta_1\})$ as follows. The set of variables of $I$ is $V(G) \times X$ and the set of constraints $C$ of~$I$ is constructed as follows.
\begin{enumerate}[(1)]
\item
For each $v\in V(G)$ and each constraint $\Imp(x,y) \in \IIv$, we add the
constraint $\Imp((v,x),(v,y))$ to $C$.\label{eq:csp1}
\item For each edge $\{u,v\} \in E(G)$ and each constraint $\Imp(x,y) \in \IIe$,
we add the constraints 
$\Imp((u,x) ,(v,y))$ and $\Imp(   (v,x),(u,y))$ to~$C$. \label{eq:csp2}
\item For each $v\in V(G)$ with $|S_v| = 1$ let $\tau$ be the  (only) element of $S_v$.
If $\tau(x)=0$ then add the constraint $\delta_0((v,x))$ to $C$.
Otherwise, add the constraint $\delta_1((v,x))$ to $C$. \label{eq:csp3} 
\end{enumerate}

To complete the reduction from $\Ret{\Hve}$ to $\csp(\{ \Imp, \delta_0, \delta_1\})$,
we will 
show that there is a bijection between homomorphisms from $(G,\boldS)$ to $\Hve$ and   satisfying assignments of $I$.
This  bijection ensures that  the number of satisfying assignments of $I$
is  equal to $ \hom{(G,\boldS)}{\Hve}$.
Hence the approximation to $ \hom{(G,\boldS)}{\Hve}$ can be achieved using a single oracle call to $\csp(\{ \Imp, \delta_0, \delta_1\})$ with the desired accuracy~$\epsilon$.

To establish the bijection, we present an (invertible) map from  
satisfying assignments of~$I$ to
homomorphisms from $(G,\boldS)$ to $\Hve$.
The map is constructed as follows. Let $\sigma$ be any satisfying assignment of~$I$.
\begin{itemize}
\item For every vertex $v\in V(G)$, 
define a function $\sigma_v \from X \to \{0,1\}$ as follows.
For all $x\in X$, let $\sigma_v(x) = \sigma((v,x))$.
The constraints added to~$C$ in item~\eqref{eq:csp1} ensure that, since $\sigma$ is a satisfying assignment of $I$, the assignment $\sigma_v$ is a satisfying assignment of $\IIv$. Thus, $\sigma_v$ is a vertex of $\Hve$.
\item Next, we will argue that the function from $V(G)$ to $V(\Hve)$ that maps every vertex $v\in V(G)$ to $\sigma_v$
is a homomorphism from $(G,\boldS)$ to $\Hve$.
\begin{itemize}
\item Consider an edge $\{u,v\}$ of~$G$. We must show that $\{\sigma_u,\sigma_v\}$ is an edge of $\Hve$.
Using Definition~\ref{def:Hve}, this is equivalent to showing that, for every constraint 
$\Imp(x,y)$ in $\IIe$,
we have  $\sigma_u(x) \Rightarrow \sigma_v(y)$  and $\sigma_v(x) \Rightarrow \sigma_u(y)$.
Using the construction of
$\sigma_u$ and $\sigma_v$, this is equivalent to showing that, for every constraint
$\Imp(x,y)$ in $\IIe$,
we have $\sigma(u,x) \Rightarrow \sigma(v,y)$ and
$\sigma(v,x) \Rightarrow \sigma(u,y)$.
This is ensured by the fact that $\sigma$ is a satisfying assignment of~$I$, so it satisfies the constraints added in item~\eqref{eq:csp2}.
\item Consider a vertex $v\in V(G)$ with $S_v = \{\tau\}$.
We must show that $\sigma_v = \tau$.
This is ensured by the constraints added in item~\eqref{eq:csp3}.

\end{itemize}
 
\end{itemize}

Starting from the satisfying assignment~$\sigma$ of~$I$,  we produced a
homomorphism from~$(G,\boldS)$ to~$\Hve$, namely the homomorphism
that maps every vertex $v\in V(G)$ to $\sigma_v$.
To finish the proof, we need only note that this construction 
is invertible --- given any homomorphism from~$(G,\boldS)$ 
to~$\Hve$ we can let $\sigma_v$ denote the image of~$v$ under this homomorphism.
Given the collection $\{\sigma_v \mid v\in V(G)\}$,
we construct an assignment $\sigma$ 
from $V(G) \times X$ to $\{0,1\}$ by inverting
the above construction: 
For every $v\in V(G)$ and $x\in X$, let $\sigma((v,x)) = \sigma_v(x)$.
We must then check that $\sigma$ is satisfying.
\begin{itemize}
\item For each $v\in V(G)$, the assignment $\sigma$ satisfies the relevant constraints
added in item~\eqref{eq:csp1} because $\sigma_v$ is a vertex of $\Hve$, hence a satisfying assignment of $\IIv$.
\item For each $\{u,v\} \in E(G)$
and each pair of constraints 
$\Imp((u,x) ,(v,y))$ and $\Imp((v,x),(u,y))$ 
added to~$C$ in item~\eqref{eq:csp2},
$\sigma$ satisfies the constraints  because $\{\sigma_u,\sigma_v\}$ is an edge of $\Hve$
(so $\sigma_u(x) \implies\sigma_v(y)$ and $\sigma_v(x) \implies \sigma_u(y)$).
\item Finally, for any $s\in \{0,1\}$,
consider a constraint $\delta_s((v,x))$ introduced in item~\eqref{eq:csp3}.
The procedure in item~\eqref{eq:csp3} ensures that, for some $\tau$ with $S_v=\{\tau\}$,
we have $\tau(x)=s$.
Our homomorphism has $\sigma_v = \tau$.
Thus, the constraint $\sigma((v,x))=s$ is satisfied by~$\sigma$.
\end{itemize}

{\bf Directed Case:} The 
reduction from $\DirRet{\Hvfb}$
to $\csp(\{ \Imp, \delta_0, \delta_1\})$
is similar to the one given in the undirected case.
Starting with an instance
 $(G,\boldS)$ of $\DirRet{\Hvfb}$ 
 we  create an instance  $I$ of $\csp(\{ \Imp, \delta_0, \delta_1\})$ as follows. 
 The set of variables of $I$ is $V(G) \times X$, as in the undirected reduction.
 The set     of constraints $C$ of~$I$ is constructed in the same way as 
 in the undirected reduction, except that item~\eqref{eq:csp2} is replaced with the following.

\begin{enumerate}  
\item[(2)']  For each (directed) edge $(u,v) \in E(G)$, we  add the following constraints to~$C$.
For  each constraint $\Imp(x,y) \in \IIf$,
we add the constraint
$\Imp((u,x) ,(v,y))$ to~$C$. For each   constraint $\Imp(x,y) \in \IIb$,
we add the constraint 
  $\Imp((v,x),(u,y))$ to~$C$.  \label{eq:csp2prime}
 \end{enumerate}

 As in the undirected case, we complete the proof by establishing a bijection  from
 satisfying assignments of~$I$ to homomorphisms from $(G,\boldS)$ to $\Hvfb$.
Let $\sigma$ be any satisfying assignment of~$I$.
The construction of~$\sigma_v$ from~$\sigma$ is the same as in the undirected case.
Only one difference arises in the verification that the function from $V(G)$ to $V(\Hvfb)$ that maps
every vertex $v\in V(G)$ to $\sigma_v$ is a homomorphism from $(G,\boldS)$ to $\Hvfb$.
Consider any directed edge $(u,v)$ of $G$.
We must show that $(\sigma_u,\sigma_v)$ is an edge of $\Hvfb$.
Using Definition~\ref{def:dirHve} and the construction of~$\sigma_u$ and
$\sigma_v$, this is equivalent to showing 
\begin{itemize}
\item For every constraint $\Imp(x,y)$ in $\IIf$,
we have  $\sigma(u,x) \Rightarrow \sigma(v,y)$, and
\item for every constraint $\Imp(x,y)$ in $\IIb$,
we have $\sigma(v,x) \Rightarrow \sigma(u,y)$.
\end{itemize}
This is ensured by the constraints added in item~\eqref{eq:csp2}'.

As in the undirected case, we next show that we have a bijection by
starting with a homomorphism from $(G,\boldS)$ to $\Hvfb$
and letting $\sigma_v$ denote the image of~$v$ under this homomorphism.
Given the collection $\{\sigma_v \mid v\in V(G)\}$ we construct an assignment $\sigma$ from $V(G)\times X$ to $\{0,1\}$
exactly as in the undirected case.
We must check that $\sigma$ is a satisfying assignment of~$I$.
This is the same as the undirected case except when checking that $\sigma$ satisfies the constraints
added in item~\eqref{eq:csp2}'.
For $(u,v)\in E(G)$ 
and a constraint 
$\Imp((u,x) ,(v,y))$ added to~$C$ because $\Imp(x,y) \in \IIf$,
note that, since $(\sigma_u,\sigma_v)$ is an edge of~$\Hvfb$,
by Definition~\ref{def:dirHve},
we have $\sigma_u(x) \implies \sigma_v(x)$,
so the constraint is satisfied.
Similarly, for a constraint 
 $\Imp((v,x),(u,y))$ added to~$C$ because
 $\Imp(x,y) \in \IIb$
 we again have $\sigma_v(x) \implies \sigma_u(y)$, so the constraint is satisfied.
 
 So we have a bijection from 
 satisfying assignments of~$I$ to homomorphisms from $(G,\boldS)$ to $\Hvfb$
 and the reduction from 
$\DirRet{\Hvfb}$
to $\csp(\{ \Imp, \delta_0, \delta_1\})$ follows. 
\end{proof} 
 
Although, to be general, we have presented 
undirected and directed reductions in  Lemma~\ref{lem:OALHomToCSPImplies},
our goal in Lemma~\ref{lem:RetBristledPath}
is to prove 
$\bis$-easiness of $\Ret{H}$ for
  a partially bristled reflexive path, which is an undirected graph.
So we will use the undirected reduction from Lemma~\ref{lem:OALHomToCSPImplies}
for this.
The rough idea will be to take a  partially bristled reflexive path~$H$
and show
how to set up the corresponding instances 
$\IIv$ and $\IIe$
of $\csp(\{\Imp\})$ so
that $\Hve = H$.
Then Lemma~\ref{lem:OALHomToCSPImplies}
shows that  $\Ret{H}$ reduces to $\csp(\{ \Imp, \delta_0, \delta_1\})$, so $\Ret{H}$ is $\bis$-easy.
Unfortunately, we can't precisely achieve this goal, but
we can set up the corresponding instances
$\IIv$ and $\IIe$ so that
$\Hve$ is equal to $H$, together with some additional small connected components,
which turn out not to matter. 
The fact that these small connected components don't cause trouble was
 first observed by Kelk~\cite{KelkThesis} in the context of counting homomorphisms. In Lemma~\ref{lem:noSingletons} we show that this is also true when counting retractions.
Lemma~\ref{lem:Kelk6.6} states the well-known fact that subtracting polynomial-size entities does not spoil an AP-reduction, which is, for instance, pointed out in~\cite[Lemma 6.6]{KelkThesis}. For the sake of completeness we give a short proof.

\begin{lem} \label{lem:Kelk6.6}
Let $H$ and $H'$ be graphs.
For any graph~$G$,
let $f(G) = \hom{G}{H} - \hom{G}{H'}$.
If $f(G)$ is non-negative and bounded from above by a polynomial in $|V(G)|$, and
can be computed in polynomial time, then
$\Hom{H'}\leap \Hom{H}$. 
\end{lem}
\begin{proof}
Let $G$ be an instance of $\Hom{H'}$ and let~$\eps\in(0,1)$ be the desired precision. 
To shorten notation, let $N = \hom{G}{H}$. From the  definition of~$f$ in the statement of the lemma, $\hom{G}{H'} = N - f(G)$. 
First, the algorithm computes $k=f(G)$ in polynomial time. If $k=0$, then $\hom{G}{H'}= \hom{G}{H}$, and the algorithm simply returns the result of a $\Hom{H}$ oracle call with precision $\eps$.  

Suppose instead that $k\geq 1$. In this case, the algorithm makes a $\Hom{H}$ oracle call with input $G$ and precision 
$\delta\le \frac{\eps}{16k}$. Let $R$ be the integer solution returned by this oracle 
call (note that $R$ is an approximation to $N$ satisfying $e^{-\delta} N \le R \le e^{\delta} N$). The algorithm returns $R-k$. We show that this output  approximates $\hom{G}{H'}$ with the desired precision.

If $\hom{G}{H'}=0$ then $N=k$ and $e^{-\delta} k \le R \le e^{\delta} k$. By Observation~\ref{obs:accuracy} and the facts that $\eps<1$ and $k\ge 1$ this implies $R \in (k-1/4, k+1/4)$ and since $R$ is integer this gives $R=k$. Thus, in this case the algorithm returns $0$, which is the exact solution.

Suppose instead that $\hom{G}{H'}\ge 1$. In this case, $N \ge k+1$ and by Observation~\ref{obs:accuracy} we have
\[
R-k \le e^\delta N - k \le (1+2\delta)N -k = (1+2\delta)(N-k) + 2k\delta.
\]
Since $N\ge k+1$ and $2\delta\le \eps/8$ we have $2k\delta\le \eps/8 \le \eps/8\cdot(N-k)$ and consequently
\[
(1+2\delta)(N-k) + 2k\delta \le (1+\eps/4)(N-k).
\]
Analogously, we obtain $R-k \ge (1-\delta)(N-k) -k\delta \ge (1-\eps/8)(N-k)$.
Finally, by Observation~\ref{obs:accuracy}, this implies $e^{-\eps} (N-k)\le R-k \le e^{\eps} (N-k)$ and thus returning $R-k$ has the desired precision.
\end{proof}

\begin{lem}\label{lem:noSingletons}
Let $H'$ be a graph
and let $H$ be the graph consisting of a connected component that is isomorphic to $H'$ together
with some additional connected components $C_1,\ldots,C_k$.
Suppose that, for each $i\in [k]$, $C_i$ is
one of the following graphs: a singleton vertex, with or without a loop, or an unlooped edge. 
Then $\OALHom{H'}\leap\OALHom{H}$.
\end{lem}
\begin{proof}
Recall the definition of $\cOALHom{H}$ from Section~\ref{sec:Preliminaries}.
To prove the lemma we show
\begin{equation}
\OALHom{H'} \leap \cOALHom{H'} \leap \cOALHom{H} \leap \OALHom{H}.
\end{equation}
The first and the trivial third reduction follow from Observation~\ref{obs:connectedOALHom}. It remains to show that $\cOALHom{H'} \leap \cOALHom{H}$. Let $(G,\boldS')$ be 
an input to $\cOALHom{H'}$  and let $\eps \in(0,1)$ be
the desired precision. From the problem definition, $G$ is connected. 
Now define the lists $S_v$ for $v\in V(G)$ as follows.
If $S'_v = V(H')$ then let $S_v = V(H)$. Otherwise, let $S_v = S'_v$. Let $\boldS=\{S_v \mid v\in V(G)\}$
 
First, the algorithm tests whether there  is a 
list $S'_v \in \boldS'$ with $|S'_v|=1$. 
If there is such a list, 
then there is a particular component of $H'$ 
with the property that every homomorphism from~$G$ to~$H'$
maps all vertices of~$G$ to this component, and every homomorphism from~$G$ to~$H$
maps all vertices of~$G$ to this component.
Thus, $\hom{(G,\boldS')}{H'}=\hom{(G,\boldS)}{H}$. 
So  a single oracle call with precision $\eps$ gives the sought-for approximation. 

If there is no list $S'_v \in \boldS'$ with $|S'_v|=1$ then
every list $S'_v$ is equal to $V(H')$
and every list $S_v$ is equal to $V(H)$. Thus, $\hom{(G,\boldS')}{H'} = \hom{G}{H'}$ and $\hom{(G,\boldS)}{H} = \hom{G}{H}$. As $G$ is connected we also have $\hom{G}{H}= \hom{G}{H'} + \sum_{i =1}^k \hom{G}{C_i}$. As $C_1,\dots C_k$ are either singleton vertices or unlooped edges, the algorithm can compute $\sum_{i =1}^k \hom{G}{C_i}$ efficiently. Also, for each $i\in [k]$,  $\hom{G}{C_i}\le 2$. Setting $f(G)=\sum_{i =1}^k \hom{G}{C_i}$ in Lemma~\ref{lem:Kelk6.6} gives the sought-for AP-reduction.
\end{proof}

As noted at the beginning of this section, 
approximately counting homomorphisms to partially bristled reflexive paths is shown to be $\bis$-easy in~\cite[Appendix A.8]{KelkThesis}. Using the same construction and our Lemma~\ref{lem:OALHomToCSPImplies} we can now prove Lemma~\ref{lem:RetBristledPath}, which is the generalisation for counting retractions.  
We restate the lemma and recast the construction in our setting for convenience.

{\renewcommand{\thethm}{\ref{lem:RetBristledPath}}
\begin{lem}\label{lem:BristledPathEasiness}
\LemRetBristledPath
\end{lem}
\addtocounter{thm}{-1}
}
\begin{proof}
The $\bis$-hardness part of the statement is inherited from $\Hom{H}$ using Theorem~\ref{thm:HomBIS} and the reduction $\Hom{H} \leap \OALHom{H}$ from Observation~\ref{obs:HomToRetToLHom}. We now show $\bis$-easiness.

Matching the notation from Definition~\ref{def:PBRP}, the partially bristled reflexive path $H$ can be described as follows. 
There exists a positive integer $Q$ and a be a subset $S$ of $[Q]$
such that  $V(H) = \{c_0,\ldots,c_{Q+1}\} \cup \bigcup_{i\in S} \{g_i\}$
and $E(H) = \bigcup_{i=0}^{Q} \{c_i,c_{i+1}\} \cup \bigcup_{i=0}^{Q+1} \{c_i,c_i\} \cup \bigcup_{i\in S} \{c_i,g_i\}$. Note that $S$ can be empty.
 
Let   $X= \{x_0,\ldots,x_{Q}\}$.
Define the instances 
$\IIv=(X,\Cv)$ and $\IIe=(X,\Ce)$ 
of $\csp(\{\Imp\})$ as follows.
\begin{itemize}
\item
For each $i\in[Q] \setminus S$,
we add a constraint $\Imp(x_i,x_{i-1})$ to $\Cv$.
\item
For each pair $(i,j)$ satisfying $0\leq i < j \leq Q$,
we add a constraint $\Imp(x_j, x_i)$ to $\Ce$.
\end{itemize}

We claim that the graph $\Hve$, 
as defined in Definition~\ref{def:Hve}, has a connected component that is isomorphic to 
$H$ and that  all other connected components of $\Hve$ are singleton vertices (without loops).
Given the claim, the reduction from $\Ret{H}$ to $\bis$
follows from Lemmas~\ref{lem:noSingletons}, \ref{lem:OALHomToCSPImplies} and~\ref{lem:CSPImpliesEqBIS} (applied in that order). 
 
 We conclude the proof by showing the claim. 
 For each $i\in \{0,\ldots,Q+1\}$ let 
 $\sigma_i \from X \to \{0,1\}$ be the following assignment 
 of Boolean values to variables in~$X$.
 $$\sigma_i(x_j) = 
 \begin{cases}
1, \quad\text{if $j<i$}\\
0, \quad\text{otherwise}.
\end{cases}$$   
 Note that $\sigma_0$    maps all arguments to $0$
 and $\sigma_{Q+1}$ maps all arguments to~$1$.
 
The indices of $\sigma_0, \dots, \sigma_{Q+1}$ are chosen this way to match the indices of the vertices $c_0$ to $c_{Q+1}$ of the graph $H$. 
Note that $\sigma_0,\dots, \sigma_{Q+1}$ are satisfying assignments of $\IIv$ and, therefore, they are vertices of $\Hve$. 
By the definition of $\IIe$,
 these vertices are looped in~$\Hve$.
Also, for all $i\in [Q]$, we have $\{\sigma_{i},\sigma_{i+1}\}\in E(\Hve)$. 
Hence, the vertices $\sigma_0,\dots, \sigma_{Q+1}$ form a reflexive path in $\Hve$.  

Now for each $i\in [Q]$ let
$\sigma'_i \from X \to \{0,1\}$ be the following assignment of Boolean values
to variables in~$X$.
$$
\sigma'_i(x_j) =  
\begin{cases}
1, \quad\text{if $j\le i$  and $j\neq i-1$}\\
0, \quad\text{otherwise}.
\end{cases}$$  
For every $i\in [Q]$, we have $\sigma'_i(x_{i-1})=0$ and $\sigma'(x_i)=1$,
so $\sigma'_i$ is not equal to any $\sigma_{i'}$. 

Consider a vertex $c_i$ of~$H$ with $i\in [Q]$.
\begin{itemize}
\item If $i\in S$:
In this case, 
$\sigma'_{i}$ is a satisfying assignment of $\IIv$.
 By the definition of $\IIe$,  $\sigma'_{i}$ has degree $1$ and is adjacent to $\sigma_{i}$ in $\Hve$.
 Thus, the vertex $\sigma'_{i}$ 
 of $\Hve$
 corresponds to the  vertex $g_i$ of~$H$.
\item If $i\notin S$:   In this case, as $i\ge 1$, the  constraint $\Imp(x_i, x_{i-1})$ 
in $\IIv$
ensures that $\sigma'_i$ is not a satisfying assignment of $\IIv$ and, therefore,
$\sigma'_i$ is not a vertex of $\Hve$.
\end{itemize} 
   
We will next show that the
edges that we have already described constitute all of the edges of $\Hve$.
This means that the rest of the vertices of $\Hve$ have degree~$0$, so we are finished.

To this end,
 let $\sigma$ be any function from $X$ to $\{0,1\}$. 
 From the definition of $\IIe$, we obtain the following necessary condition for $\sigma$ to have a neighbour in $\Hve$: 
 Let $i\in \{0,\ldots,Q\}$ be the largest index for which $\sigma(x_i)=1$. If $\psi$ is a neighbour of $\sigma$ then, for all $j \le i-1$, $\psi(x_j)=1$ and hence, for all $j \le i-2$, $\sigma(x_j)=1$. Thus, for $\sigma$ to have a neighbour in $\Hve$ it has to be of the form $\sigma_i$ or $\sigma'_i$.
\end{proof}

\begin{remark}
One interesting feature of 
Lemma~\ref{lem:RetBristledPath}  
is that it
 shows that there  are graphs $H$ for which $\OALHom{H}$ is $\bis$-equivalent, whereas $\LHom{H}$ is $\sat$-hard. 
 Thus, subject to the complexity assumption that $\bis$ is not $\sat$-equivalent,
 there is a graph~$H$ for which   the complexity of $\OALHom{H}$ differs from that of $\LHom{H}$. 
 The smallest example from the class of partially bristled reflexive paths for which this separation holds is the so-called $2$-Wrench, depicted in Figure~\ref{fig:2Wrench}.
The fact that $\sat\leap\LHom{\text{$2$-Wrench}}$ follows from Theorem~\ref{thm:LHomTricho}.
\end{remark}

\begin{figure}[h!]\centering
{\def\scaleFactor{1}
\begin{tikzpicture}[scale=1, every loop/.style={min distance=10mm,looseness=10}]

	\filldraw (0,0) node (r1){} circle[radius=3pt] --++ (0:1.5cm) node (b){} circle[radius=3pt] --++ (0:1.5cm) node (r2){} circle[radius=3pt];
	\filldraw (b) --++ (-90:1.5cm) node (g){} circle[radius=3pt];

	\path[-] (r1.center) edge  [in=125,out=55,loop] node {} ();
	\path[-] (b.center) edge  [in=125,out=55,loop] node {} ();
	\path[-] (r2.center) edge  [in=125,out=55,loop] node {} ();
	
\end{tikzpicture}
}
\caption{The graph $2$-Wrench.}
\label{fig:2Wrench}
\end{figure}
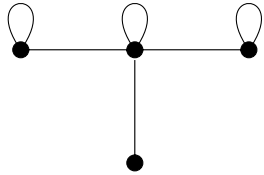

\subsubsection{$\sat$-Hardness Results for Graphs with Loops}\label{sec:MixedTreesHardness}

The goal of this section is to prove the hardness results given in Lemmas~\ref{lem:MixedTreeHardness2a},~\ref{lem:MixedTreeHardness2b} and~\ref{lem:MixedTreeHardness3}. In order to show $\sat$-hardness results we will prove that certain neighbourhood structures induce hardness. To this end consider the following easy and well-known observation proved here for completeness.

\begin{obs}
\label{obs:PinNeighbourhood}
Let $H$ be a graph and let $u$ be a vertex of $H$. Then $\Hom{H[\Gamma(u)]} \leap \OALHom{H}$.
\end{obs}
\begin{proof}
Let $G$ be an input to $\Hom{H[\Gamma(u)]}$ and let $v_1, \dots, v_n$ be the vertices of $G$. Let $w$ be a vertex distinct from the vertices in $G$. Then we construct the graph $G'$ with vertices $V(G')=V(G)\cup \{w\}$ and edges $E(G')=E(G) \cup \{\{w,v_i\} \mid i\in [n]\}$.
We set $S_{w}=\{u\}$ and $S_v=V(H)$ for all remaining vertices of $G'$. Let $\boldS=\{S_v \mid v\in V(G')\}$. Then $\hom{G}{H[\Gamma(u)]} = \hom{(G',\boldS)}{H}$.
\end{proof}

First we combine some known results to show hardness that is derived from the analysis of distance-$1$ neighbourhoods (Lemmas~\ref{lem:MixedTreeHardness2a} and~\ref{lem:MixedTreeHardness2b}). Then we show hardness results derived from the analysis of distance-$2$ neighbourhoods in the more difficult Lemma~\ref{lem:MixedTreeHardness3}, which is the main result of this section.

For Lemmas~\ref{lem:MixedTreeHardness2a} and~\ref{lem:MixedTreeHardness2b} we use gadgets based on complete bipartite graphs where two states dominate (see, e.g.,~\cite[Lemma 25]{DGGJApprox},~\cite[Section 5]{GKP2004} and~\cite[Lemma 5.1]{KelkThesis}). We use the version of Kelk~\cite{KelkThesis}.
Let $F(H)=\{u \in V(H) \mid \Gamma(u) = V(H)\}$. For a set of vertices $S$ recall the set of common neighbours $\Gamma(S)$ from Section~\ref{sec:Preliminaries}.
\begin{lem}[{\cite[Lemma 5.1]{KelkThesis}}]\label{lem:Kelk5.1}
Let $H$ be a graph with $\emptyset \subsetneq F(H) \subsetneq V(H)$. Suppose that, for every pair $(S,T)$ with $\emptyset \subseteq S,T \subseteq V(H)$ satisfying $S\subseteq \Gamma(T)$ and $T\subseteq \Gamma(S)$, at least one of the following holds:
\begin{enumerate}[(1)]
\item $S=F(H)$.\label{eq:Kelk1}
\item $T=F(H)$.\label{eq:Kelk2}
\item $\abs{S}\cdot \abs{T} < \abs{F(H)}\cdot \abs{V(H)}$.\label{eq:Kelk3}
\end{enumerate}
Then $\sat \leap \Hom{H}$.
\end{lem}

Lemma~\ref{lem:Kelk5.1} is not difficult to prove. A homomorphism from a complete bipartite graph to $H$ will typically map one side to $F(H)$ and the other to $V(H)$. So it is easy to reduce from counting independent sets.

Let $\WR{q}$ be a reflexive star with $q$ leaves. (The name is not relevant here but it comes from the Widom-Rowlinson model~\cite{WidomRowlinson} from statistical physics.) The non-leaf vertex of $\WR{q}$ is called its centre.

\begin{lem}\label{lem:MixedTreeHardness2a}
Let $H$ be a graph that has a looped vertex $b$ such that $H[\Gamma(b)]$ is isomorphic to $\WR{q}$ for some $q\ge 3$. Then $\sat\leap \OALHom{H}$.
\end{lem}
\begin{proof}
The problem $\Hom{\WR{q}}$ is the same as $\Hom{H[\Gamma(b)]}$, and by Observation~\ref{obs:PinNeighbourhood} we obtain $\Hom{H[\Gamma(b)]}\leap \OALHom{H}$.
For $q\ge 4$ Dyer et al.~\cite[Lemma 26]{DGGJApprox} show $\sat\leap\Hom{\WR{q}}$. For $q=3$ this fact is due to Kelk~\cite[Section 2.3]{KelkThesis}.
Summarising we obtain
\[
\sat\leap\Hom{\WR{q}}\eqap\Hom{H[\Gamma(b)]}\leap \OALHom{H}.
\] 
\end{proof}

Recall the $2$-Wrench as given in Figure~\ref{fig:2Wrench}.

\begin{lem}\label{lem:MixedTreeHardness2b}
Let $H$ be a triangle-free graph that has a looped vertex $b$ which has an unlooped neighbour. If $H[\Gamma(b)]$ is not isomorphic to a $2$-Wrench, then $\sat\leap \OALHom{H}$.
\end{lem}
\begin{proof}
By Observation~\ref{obs:PinNeighbourhood} we know $\Hom{H[\Gamma(b)]} \leap \OALHom{H}$. We show $\sat \leap \Hom{H[\Gamma(b)]}$ to obtain $\sat \leap \OALHom{H}$.

To shorten the notation let $H_b=H[\Gamma(b)]$. We consider different cases depending on the graph $H_b$. By assumption the vertex $b$ is looped and has at least one unlooped neighbour. First consider the case where $H_b$ has at most $4$ vertices. Since, by assumption, $H_b$ is triangle-free and not isomorphic to a $2$-Wrench, it has to be isomorphic to one of the graphs depicted in Figure~\ref{fig:MixedTrees0}. Approximately counting homomorphisms to the first graph in Figure~\ref{fig:MixedTrees0} is well-known to be equivalent to $\is$ (the problem of approximately counting independent sets in a graph) which is $\sat$-equivalent~\cite[Theorem 3]{DGGJApprox}. The second and fourth graphs correspond to weighted versions of $\is$ which are known to be $\sat$-equivalent~\cite[Lemma 2.3]{KelkThesis}. The third graph is the so-called $1$-Wrench and the corresponding $\sat$-hardness is shown in~\cite[Theorem 21]{DGGJApprox}. Finally, approximately counting homomorphisms to the fifth graph in Figure~\ref{fig:MixedTrees0} is shown to be $\sat$-hard in~\cite[Section 2.3]{KelkThesis}.

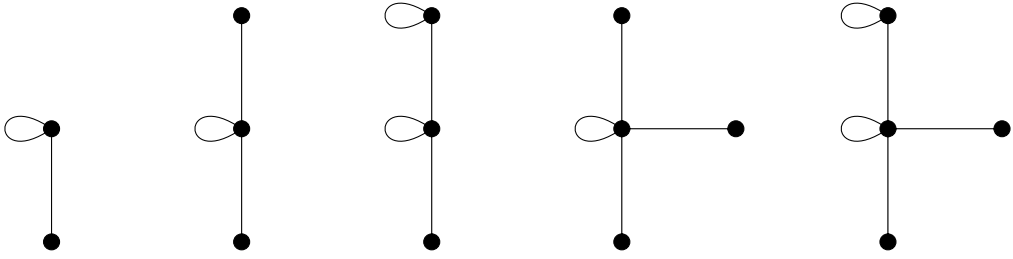
\begin{figure}[ht]
	\centering
	\begin{tikzpicture}[scale=1, every loop/.style={min distance=10mm,looseness=10}]

			\filldraw (0,0) node(a){} circle[radius=3pt] --++ (90:1.5cm) node(b){} circle[radius=3pt];
			
b			\path[-] (b.center) edge  [in=215,out=145,loop] node {} ();
	
			\filldraw (2.5,0) node(a){} circle[radius=3pt] --++ (90:1.5cm) node(b){} circle[radius=3pt] -- ++(90:1.5cm) node(c){} circle[radius=3pt];
	
			\path[-] (b.center) edge  [in=215,out=145,loop] node {} ();

			\filldraw (5,0) node(a){} circle[radius=3pt] --++ (90:1.5cm) node(b){} circle[radius=3pt] -- ++(90:1.5cm) node(c){} circle[radius=3pt];
	
			\path[-] (b.center) edge  [in=215,out=145,loop] node {} ();
			\path[-] (c.center) edge  [in=215,out=145,loop] node {} ();
			
			\filldraw (7.5,0) node(a){} circle[radius=3pt] --++ (90:1.5cm) node(b){} circle[radius=3pt] -- ++(90:1.5cm) node(c){} circle[radius=3pt] ++(90:-1.5cm) -- ++(0:1.5cm) node(d){} circle[radius=3pt];
	
			\path[-] (b.center) edge  [in=215,out=145,loop] node {} ();

			\filldraw (11,0) node(a){} circle[radius=3pt] --++ (90:1.5cm) node(b){} circle[radius=3pt] -- ++(90:1.5cm) node(c){} circle[radius=3pt] ++(90:-1.5cm) -- ++(0:1.5cm) node(d){} circle[radius=3pt];
	
			\path[-] (b.center) edge  [in=215,out=145,loop] node {} ();
			\path[-] (c.center) edge  [in=215,out=145,loop] node {} ();
			
	\end{tikzpicture}
	\caption{Possible graphs $H_b$ with at most $4$ vertices.}
	\label{fig:MixedTrees0}
\end{figure}

Now consider the case where $H_b$ has $5$ or more vertices. We claim that, under this assumption, Lemma~\ref{lem:Kelk5.1} gives $\sat \leap \Hom{H_b}$. To see this, note that $F(H_b)=\{b\}$ and $\abs{F(H_b)}\abs{V(H_b)} \ge 5$. Consider any pair $(S,T)$ with $\emptyset\subseteq S,T \subseteq V(H_b)$, $S\subseteq \Gamma(T)$ and $T\subseteq \Gamma(S)$. We distinguish between different cases depending on the cardinalities of $S$ and $T$ and show that in each case the conditions of Lemma~\ref{lem:Kelk5.1} are fulfilled. 
\begin{itemize}
\item If $\abs{S}=1$ then either item~\eqref{eq:Kelk1} or item~\eqref{eq:Kelk3} of Lemma~\ref{lem:Kelk5.1} are satisfied.
\item If $\abs{T}=1$ then either item~\eqref{eq:Kelk2} or item~\eqref{eq:Kelk3} of Lemma~\ref{lem:Kelk5.1} are satisfied.
\item If $\abs{S}\ge 3$ then $T=\{b\}$ since $T\subseteq \Gamma(S)$ and $H$ is triangle-free. So $\abs{T}=1$.
\item If $\abs{T}\ge 3$ then $S=\{b\}$ since $S\subseteq \Gamma(T)$ and $H$ is a triangle-free. So $\abs{S}=1$.
\item If $\abs{S}=\abs{T}=2$ then $\abs{S}\cdot \abs{T} = 4$ and item~\eqref{eq:Kelk3} of Lemma~\ref{lem:Kelk5.1} is satisfied.
\end{itemize}
So Lemma~\ref{lem:Kelk5.1} gives $\sat \leap \Hom{H_b}$.
\end{proof}

The goal of the remainder of this section is to show Lemma~\ref{lem:MixedTreeHardness3}, in which we prove $\sat$-hardness using distance-$2$ neighbourhoods of vertices in $H$. In order to show $\sat$-hardness we use a reduction from counting large cuts in a graph $G$. We use graph gadgets to model these cuts. We replace each vertex of $G$ by a graph $J$ such that the number of homomorphisms from $J$ to $H$ is dominated by exactly two ``types'' of homomorphisms. These two types encode the two parts of a cut. In Table 1 we give all types that represent a significant share of the set of homomorphisms. In Lemma 44 we show how to choose parameters of the graph $J$ to ensure that only 2 significant types remain. In the proof of Lemma 45 we verify another desired property, which is that the two types interact in an “anti-ferromagnetic” way to ensure that large cuts dominate.

At this point we introduce the gadget graph $J$ and introduce some of its properties. Note that a similar but simpler gadget has been used in~\cite{DGGJApprox} and~\cite{KelkThesis}.

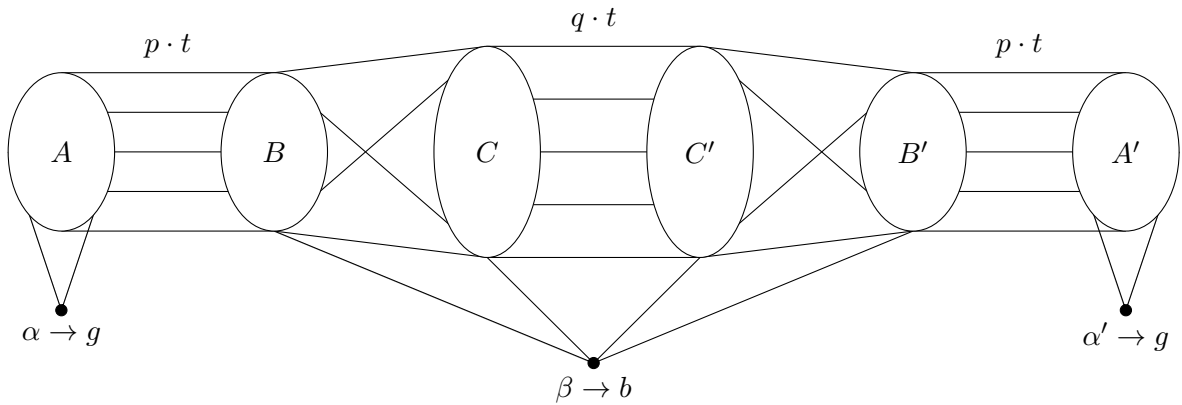
\begin{figure}[h!]\centering
{\def\scaleFactor{.7}
\begin{tikzpicture}[scale=\scaleFactor, every loop/.style={min distance=10mm,looseness=10}]

\foreach \i in {-2,...,2}
	\draw (0,\i)--(4,\i);

\foreach \i in {-2,...,2}
	\draw (-8,.75*\i)--(-4,.75*\i);

\foreach \i in {-2,...,2}
	\draw (8,.75*\i)--(12,.75*\i);
	
\node at ($(2,2.5)$) {$q\cdot t$};
\node at ($(-6,2)$) {$p\cdot t$};
\node at ($(10,2)$) {$p\cdot t$};

\draw (-4,1.5) -- (0,2);
\draw (-4,1.5) -- (0,-2);
\draw (-4,-1.5) -- (0,-2);
\draw (-4,-1.5) -- (0,2);

\draw (8,1.5) -- (4,2);
\draw (8,1.5) -- (4,-2);
\draw (8,-1.5) -- (4,-2);
\draw (8,-1.5) -- (4,2);

\coordinate (g) at (-8,-3);
\filldraw (-9,0) -- (g) circle[radius=3pt] -- (-7,0);
\node at ($(g)+(0,-.5)$) {$\alpha\rightarrow g$};

\coordinate (g) at (12,-3);
\filldraw (11,0) -- (g) circle[radius=3pt] -- (13,0);
\node at ($(g)+(0,-.5)$) {$\alpha'\rightarrow g$};

\filldraw (-4,-1.5) --  (2,-4) node (b){} circle[radius=3pt];
\draw (0,-2) to (b.center);
\draw (4,-2) to (b.center);
\draw (8,-1.5)  to  (b.center);
\node at ($(b)+(0,-.5)$) {$\beta\rightarrow b$};

\draw (-8,0) ellipse (1cm and 1.5cm) [fill=white];
\draw (-4,0) ellipse (1cm and 1.5cm) [fill=white];
\draw (0,0) ellipse (1cm and 2cm) [fill=white];
\draw (4,0) ellipse (1cm and 2cm) [fill=white];
\draw (8,0) ellipse (1cm and 1.5cm) [fill=white];
\draw (12,0) ellipse (1cm and 1.5cm) [fill=white];

\node at (-8,0) {$A$};
\node at (-4,0) {$B$};
\node at (0,0) {$C$};
\node at (4,0) {$C'$};
\node at (8,0) {$B'$};
\node at (12,0) {$A'$};

\end{tikzpicture}
}
\caption{The graph $J$. A label of the form $v \rightarrow u$ means that the vertex $v\in V(J)$ is pinned to $u\in V(H)$ since $S_v=\{u\}$.}
\label{fig:AdHocGadget}
\end{figure}

\begin{defn}
For sets $X$ and $Y$ we define $\ucp{X}{Y}= \{\{x,y\} \mid x\in X, y\in Y\}$ as an undirected version of the usual definition of the Cartesian product.
\end{defn}
\begin{defn}\label{def:GadgetJ}
We now define the graph $J$, as visualised in Figure~\ref{fig:AdHocGadget}. 
Let $p$, $q$ and $t$ be positive  integers --- these are parameters of~$J$. Let
$A$, $A'$, $B$ and $B'$ be  independent sets
of size $p\cdot t$ and  let
$C$ and $C'$ be independent sets
of size $q\cdot t$. These six sets are pairwise disjoint. In addition, we introduce vertices $\alpha$, $\alpha'$ and $\beta$ that are distinct from each other and the remaining vertices. The vertex set of $J$ is the union of $\{\alpha, \alpha',\beta\}$ and the sets $A$, $A'$, $B$, $B'$, $C$ and $C'$. As displayed in Figure~\ref{fig:AdHocGadget}, the edge set of $J$ is defined as follows. The set of edges $\calM_1$ between the vertices of $A$ and $B$ forms a perfect matching (every vertex in $A$ is adjacent to exactly one vertex in $B$ and vice versa). The set of edges $\calM_2$ between the vertices of $C$ and $C'$ and the edges $\calM_3$ between the vertices of $A'$ and $B'$ form perfect matchings respectively. Then
\begin{align*}
E(J)= &\bigcup_{i\in[3]} \calM_i \cup \Bigl(\ucp{B}{C}\Bigr) \cup \Bigl(\ucp{B'}{C'}\Bigr)\\
&\cup \Bigl(\ucp{\{\alpha\}}{A}\Bigr) \cup \Bigl(\ucp{\{\alpha'\}}{A'}\Bigr) \cup \Bigl(\ucp{\{\beta\}}{\left(B\cup C\cup C'\cup B'\right)}\Bigr).
\end{align*}
This completes the definition of the graph $J$.
\end{defn}

\begin{figure}[h!]\centering
{\def\scaleFactor{1}
\begin{tikzpicture}[scale=1, every loop/.style={min distance=10mm,looseness=10}]

	\draw [red, ultra thick]   ($(0,0)+(0cm, 0cm)$) to[out=45,in=135] ($(8,0)+(0cm, 0cm)$);
	\draw [red, ultra thick]   ($(.6,-1.6)+(0cm, 0cm)$) to[out=-55,in=165] ($(3,-2.6cm)+(0cm, 0cm)$);
	\draw [red, ultra thick]   ($(7.4,-1.6)+(0cm, 0cm)$) to[out=-145,in=15] ($(5,-2.6cm)+(0cm, 0cm)$);

	\draw [red, thick] (4,-2.614) ellipse (1.5cm and 1cm) [fill=white];
	\draw[red, thick, rotate around={70:(7.75,-.75)}] (7.75,-.75) ellipse (1.5cm and .8cm) [fill=white];
	\draw[red, thick, rotate around={110:(0.25,-.75)}] (0.25,-.75) ellipse (1.5cm and .8cm) [fill=white];

	\filldraw (0,0) node (w1){} circle[radius=3pt] --++ (0:2cm) node (r1){} circle[radius=3pt] --++ (0:2cm) node (b){} circle[radius=3pt] --++ (0:2cm) node (r2){} circle[radius=3pt] --++ (0:2cm) node (w2){} circle[radius=3pt];
	\filldraw (b) --++ (-90:1.414cm) node (g){} circle[radius=3pt] --++ (-130:1.5cm) node(y1) {} circle[radius=3pt];
	\filldraw (g.center) --++ (-50:1.5cm) node(yk) {} circle[radius=3pt];
	\filldraw (r1.center) --++ (-135:2cm) node (d1){} circle[radius=3pt];
	\filldraw (r2.center) --++ (-45:2cm) node (d2){} circle[radius=3pt];

	\path[-] (r1.center) edge  [in=125,out=55,loop] node {} ();
	\path[-] (b.center) edge  [in=125,out=55,loop] node {} ();
	\path[-] (r2.center) edge  [in=125,out=55,loop] node {} ();
	\path[-] (w1.center) edge  [in=125,out=55,loop] node {} ();
	\path[-] (w2.center) edge  [in=125,out=55,loop] node {} ();
	\path[-] (y1.center) edge  [in=-125,out=-55,loop] node {} ();
	\path[-] (yk.center) edge  [in=-125,out=-55,loop] node {} ();
	\node at ($(g)+(0,-1.2cm)$) {$\ldots$};
	\node at ($(w1)+(0,-.4cm)$) {$w_1$};
	\node at ($(r1)+(0,-.4cm)$) {$r_1$};
	\node at ($(b)+(-.3,-.4cm)$) {$b$};
	\node at ($(r2)+(0,-.4cm)$) {$r_2$};
	\node at ($(w2)+(0,-.4cm)$) {$w_2$};
	\node at ($(d1)+(0cm,-.4cm)$) {$d_1$};
	\node at ($(g)+(-.3cm,0cm)$) {$g$};
	\node at ($(d2)+(0cm,-.4cm)$) {$d_2$};
	\node at ($(y1)+(0cm,.4cm)$) {$y_1$};
	\node at ($(yk)+(0cm,.4cm)$) {$y_k$};
	
\end{tikzpicture}
}
\caption{The graph $H_k$. Circled sets of vertices are independent sets of possibly looped vertices. Sets of vertices that are connected by a thick red edge have a complete set of edges between them.}
\label{fig:Hk}
\end{figure}
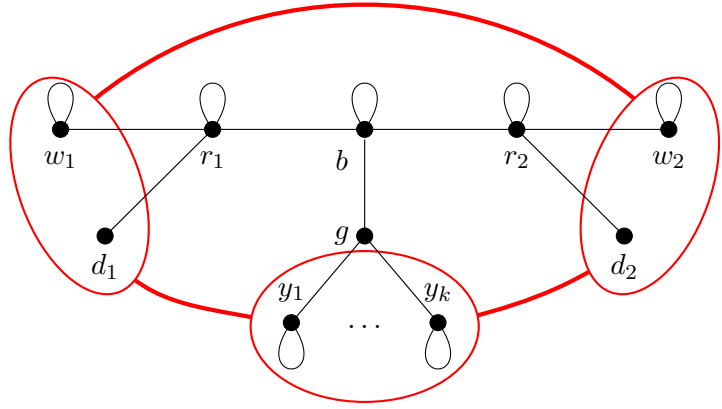

For  
any
positive integer $k$ let $H_k$ be the graph as shown in Figure~\ref{fig:Hk}. The vertex set of $H_k$ is $\{w_1, d_1, r_1, w_2, d_2, r_2, b, g, y_1,\dots, y_k\}$. All of these vertices are looped except for $d_1$, $d_2$ and $g$. The non-loop edges of $H_k$ are 
the edges in
$$\{\{w_1,r_1\}, \{w_2,r_2\}, \{d_1,r_1\}, \{d_2,r_2\}, \{r_1,b\}, \{r_2,b\}, \{b,g\}, \{g,y_1\}, \dots, \{g,y_k\}\},$$ together with 
those in $\ucp{\{w_1,d_1\}}{\{w_2,d_2\}}$, $\ucp{\{w_1,d_1\}}{\{y_1, \dots, y_k\}}$ and $\ucp{\{w_2,d_2\}}{\{y_1, \dots, y_k\}}$.
The significance of this graph will become clear in the proof of Lemma~\ref{lem:MixedTreeHardness3}.

For a graph $J$ we define the vertex lists $S_\alpha=\{g\}$, $S_{\alpha'}=\{g\}$, and $S_\beta=\{b\}$. Also,
for all $v\in V(J)\setminus\{\alpha, \alpha',\beta\}$,
we define
$S_v=V(H_k)$. Finally, we let $\boldS_J=\{S_v \mid v\in V(J)\}$. 
In order to investigate the number of homomorphisms from $(J,\boldS_J)$ to $H_k$,
we set up the following notation.
Suppose that $U$ and $V$ are subsets of $V(J)$ and  that $h$ is a homomorphism  
$h\in \calH((J,\boldS_J),H_k)$.  We define

\begin{myitemize}
\item $h(V)=\{h(x) \mid x\in V\}$ and
\item $h(U,V)=\{(h(x),h(y)) \mid x\in U, v\in V, \{x,y\}\in  E(J)\}$.\end{myitemize}

We say that $\left(h(A,B), h(C,C'), h(B',A')\right)$ is the \emph{type} of $h$. 
We will partition the set $\calH((J,\boldS_J),H_k)$ into different classes by type.
Formally, a type is a tuple $T=(T_1,T_2,T_3)$ where each $T_i$ is a subset of
$\{(x,y) \mid x\in V(H_k), y\in V(H_k), \{x,y\} \in E(H_k)\}$. 
The type of a homomorphism gives a lot of information.
Given a type~$T=(T_1,T_2,T_3)$, let
$A(T) = \{x \mid \exists y \,(x,y)\in T_1\}$,
$B(T) = \{y \mid \exists x \,(x,y)\in T_1\}$,
$C(T) = \{x \mid \exists y \,(x,y)\in T_2\}$,
$C'(T) = \{y \mid \exists x \,(x,y)\in T_2\}$,
$B'(T) = \{x \mid \exists y \,(x,y)\in T_3\}$, and
$A'(T) = \{y \mid \exists x \,(x,y)\in T_3\}$.
 If a homomorphism 
 $h\in \calH((J,\boldS_J),H_k)$  has type $T$
 then it is clear from the definition of~$J$ that
 $h(A)=A(T)$, $h(B)=B(T)$, $h(C)=C(T)$, $h(C')=C'(T)$, $h(B') = B'(T)$ and $h(A')=A'(T)$.
 A type $\type$ is called \emph{non-empty} if there exists a homomorphism from $(J,\boldS_J)$ to $H_k$ that has type $\type$, otherwise it is called empty. The following observation follows from the definition of~$J$.
 \begin{obs}\label{obs:emptytype}
 A type $T=(T_1,T_2,T_3)$ is non-empty if 
 and only if
 \begin{enumerate}[(1)]
 \item $T_1$, $T_2$ and $T_3$ are non-empty,
 \item $B(T) \cup C(T) \cup C'(T) \cup B'(T) \subseteq \Gamma(b)$,
 \item $A(T)\cup A'(T) \subseteq \Gamma(g)$,
 \item  $B(T)\times C(T) \subseteq E(H_k)$ and $B'(T) \times C'(T) \subseteq E(H_k)$.
 \end{enumerate}
 \end{obs}

Given a type $\type=(\type_1, \type_2, \type_3)$ we define $N(\type)$ to be the number of homomorphisms in $\calH((J,\boldS_J),H_k)$ that have type $\type$. We also set $\displaystyle\hatN(\type)=\abs{\type_1}^{pt} \abs{\type_2}^{qt} \abs{\type_3}^{pt}$. In Lemma~\ref{lem:hatN} we show that, for non-empty~$T$,
$\hatN(\type)$  is a close approximation to $N(\type)$.

We use the following technical fact. Let $\stirling{a}{b}$ be the Stirling number of the second kind, i.e.\ the number of surjective functions from a set of $a$ elements to a set of $b$ elements.
\begin{lem}[{\cite[Lemma 18]{DGGJApprox}}]\label{lem:DGGJ18}
If $a$ and $b$ are positive integers and $a \ge 2b \ln b$, then
\[
b^a \left( 1- \exp{\left(-\frac{a}{2b}\right)}\right) \le \stirling{a}{b}\le b^a.
\]
\end{lem}

\begin{lem}\label{lem:hatN}
Let $p$ and $q$ be positive integers. There exists a positive integer $t_0$ such that for all $t\ge t_0$ and all non-empty types $\type$ of the corresponding graph $J$ we have
\[
\frac{\hatN(\type)}{2} \le N(\type)\le \hatN(\type).
\]
\end{lem}
\begin{proof}
Let $\type = (\type_1,\type_2,\type_3)$ be a non-empty type. 
Then
\begin{equation}\label{eq:NofT}
N(\type)=\stirling{p\cdot t}{\abs{\type_1}}\cdot\stirling{q\cdot t}{\abs{\type_2}}\cdot\stirling{p\cdot t}{\abs{\type_3}}.
\end{equation}
For fixed $p$ and $q$ and sufficiently large $t_0$ 
we know from Lemma~\ref{lem:DGGJ18} that for all $t\ge t_0$ we have 
$$1-\exp\left(-\frac{p \cdot t}{2\abs{\type_1}}\right)\ge (1/2)^{1/3},$$ 
an analogous bound holds for the other two factors in Equation~(\ref{eq:NofT}). The statement of the lemma then directly follows from Lemma~\ref{lem:DGGJ18}.
\end{proof}

\begin{defn}
We say that a type $\type=(\type_1,\type_2,\type_3)$ is \emph{maximal} if it is non-empty and every type $\type'=(\type'_1,\type'_2,\type'_3)$ with $\type'\neq \type$, $\type_1\subseteq \type'_1$, $\type_2\subseteq \type'_2$ and $\type_3\subseteq \type'_3$ is empty. 
\end{defn}

Using this definition of maximality we prove that the number of homomorphisms in $\calH((J,\boldS_J),H_k)$ that have a maximal type is exponentially larger as a function of $t$ than the number of homomorphisms that have non-maximal types. Note that the precise value of the fraction $\frac{31+12k}{32+12k}$ that appears in the following lemmas is not important, we only need it to be smaller than $1$. This particular bound uses the fact that, for any type $(\type_1,\type_2,\type_3)$, 
the sets~$T_1$, $T_2$ and $T_3$
have cardinality at most $2\abs{E(H_k)}= 32+12k$. 
\begin{constraint}\label{constraint:pq}
In 
our proofs
we will need the fact that the parameters $p$ and $q$ of $J$ are sufficiently large with respect to the number of edges in $H_k$. 
In particular, we  require that
$p,q \ge 2\abs{E(H_k)}= 32+12k$.
\end{constraint}

\begin{lem}\label{lem:maxType}
Let $\type$ be a non-empty type that is not maximal. Then there exists a non-empty type $\type^*$ such that $\hatN(\type) \le \left(\frac{31+12k}{32+12k}\right)^t \hatN(\type^*)$.
\end{lem}
\begin{proof}
Let $\type=(\type_1,\type_2,\type_3)$ be a non-empty type that is not maximal. Then there exists a non-empty type $\type^*=(\type^*_1,\type^*_2,\type^*_3)$ with $\type^*\neq \type$ and $\type_i\subseteq \type^*_i$ for $i\in[3]$. Since $\type^*\neq \type$ there exists an index $i\in[3]$ such that $\type_i \subsetneq \type_i^*$, i.e.\ $\abs{\type_i}\le \abs{\type_i^*}-1$. Then (using the fact that $p,q\ge 1$)
\[
\frac{\hatN(\type)}{\hatN(\type^*)} = \frac{\abs{\type_1}^{pt}\abs{\type_2}^{qt}\abs{\type_3}^{pt}}{\abs{\type_1^*}^{pt}\abs{\type_2^*}^{qt}\abs{\type_3^*}^{pt}} \le \left(\frac{\abs{\type_i^*} -1}{\abs{\type_i^*}}\right)^{t}
\le \left(\frac{2\abs{E(H_k)} -1}{2\abs{E(H_k)}}\right)^{t}
\le \left(\frac{31+12k}{32+12k}\right)^t.
\]
\end{proof}

\begin{defn}\label{eq:EXY}
Let $E=E(H_k)$. For all $X\subseteq V(H_k)$ and $Y \subseteq V(H_k)$ we set $E(X,Y) = \{(x,y) \mid x\in X, y\in Y, \{x,y\} \in E\}$.
\end{defn}

For a set of vertices $S$ in a graph $H$ recall the definition of the set of common neighbours $\Gamma(S)$ 
and the set of all neighbours $\Phi(S)$
from Section~\ref{sec:Preliminaries}.

\begin{lem}\label{lem:necessary1}
Let   $\type=(\type_1, \type_2, \type_3)$ 
be a maximal type. Then
\begin{enumerate}[(1)] 
\item $T_1 = E(A(T),B(T))$, $T_2 = E(C(T),C'(T))$ and $T_3 = E(B'(T),A'(T))$. Also,
\item  
$C(T)=\Gamma\big(\,\Gamma(C(T)) \cap \Gamma(b)\,\big) \cap \Gamma(b)$ 
and  
$C'(T)=\Gamma\big(\,\Gamma(C'(T)) \cap \Gamma(b)\,\big) \cap \Gamma(b)$.
\item $B(T)=\Gamma(C(T)) \cap \Gamma(b)$ and
$B'(T)=\Gamma(C'(T)) \cap \Gamma(b)$ .
\item $A(T) = \Phi(B(T)) \cap \Gamma(g)$ and 
$A'(T) = \Phi(B'(T)) \cap \Gamma(g)$.
 \end{enumerate}
\end{lem}
\begin{proof}

Let $\type=(\type_1, \type_2, \type_3)$ be a non-empty type. 

{\bf Proof of (1):\quad}
It is clear from the definitions  that
$T_1 \subseteq E(A(T),B(T))$, $T_2 \subseteq E(C(T),C'(T))$ and $T_3 \subseteq E(B'(T),A'(T))$.  
Suppose that $T_2$ is a strict subset of
$E(C(T),C'(T))$.
We will show that $T$ is not maximal.
To this end,  consider the type $\type^*=(\type_1,E(C(T),C'(T)),\type_3)$. 
Note that $A(T^*) = A(T)$, $B(T^*)=B(T)$, $C(T^*)=C(T)$,
$C'(T^*)=C'(T)$, $B'(T^*)=B'(T)$ and $A'(T^*)=A'(T)$.
Using Observation~\ref{obs:emptytype} and the fact that $T$ is non-empty,
we conclude that $T^*$ is non-empty.
Using the definition of maximality (comparing $T$ to $T^*$) we conclude that $T$ is not maximal.
Similarly, if $T_1$ is a strict subset of $E(A(T),B(T))$ 
or if $T_3$ is a strict subset of $E(B'(T),A'(T))$ then $T$ is not maximal.
 
{\bf Proof of (2):\quad}
Let $X = \Gamma\big(\, \Gamma(C(T)) \cap \Gamma(b)\,\big)$
and $S = X   \cap \Gamma(b)$.
If $y\in \Gamma(C(T)) \cap \Gamma(b)$ 
then $y$ is certainly adjacent to everything in $C(T)$, so $C(T) \subseteq X$.
 Since $C(T)\subseteq \Gamma(b)$ by Observation~\ref{obs:emptytype}, we conclude that $C(T)$ is a subset of $S $. 
Similarly, 
defining $X' = \Gamma\big(\, \Gamma(C'(T)) \cap \Gamma(b)\,\big)$
and $S' =  X' \cap \Gamma(b)$, we have
$C'(T) \subseteq S'$.
  Thus, $\type_2\subseteq E(S,S')$.
Consider the type $\type^{*}=(\type_1,E(S,S'), \type_3)$. 

\begin{itemize}
\item
We first show that  $\type^{*}$ is non-empty. 
Note that $A(T^{*}) = A(T)$, $B(T^{*})=B(T)$, $A'(T^{*})=A'(T)$ and $B'(T^{*})=B'(T)$.
Also,  $C(T^{*}) \subseteq S \subseteq \Gamma(b)$
and $C'(T^{*}) \subseteq S' \subseteq \Gamma(b)$.
Using Observation~\ref{obs:emptytype} and the fact that $T$ is non-empty,
we must check that  
 $B(T)\times C(T^{*}) \subseteq E(H_k)$ and $B'(T) \times C'(T^{*}) \subseteq E(H_k)$.
 To do this, we will check that
  $B(T)\times  S \subseteq E(H_k)$ and $B'(T) \times  S' \subseteq E(H_k)$.
 
 We start with the first of these.
 Since $T$ is non-empty, Observation~\ref{obs:emptytype} guarantees that 
 $B(T) \subseteq \Gamma(C(T)) \cap \Gamma(b)$.
 So it suffices to show 
 that
 \big(\,$\Gamma(C(T)) \cap \Gamma(b)\,\big) \times S \subseteq E(H_k)$, which follows from the definition of~$S$.
 The proof that $B'(T) \times  S' \subseteq E(H_k)$ is similar.
 We have shown that $T^{*}$ is non-empty. 
 
 \item We next show that $C(T^*)=S$. We have already established that $C(T^*)\subseteq S$.
 The vertex $b$ is adjacent to everything in $\Gamma(b)$
 so it is  adjacent to everything in the subset
 $ \Gamma(C'(T)) \cap \Gamma(b)$ hence
 $b\in X'$. Since $b$ has a loop, this  implies $b\in S'$. By the definition of $T^*$ it follows that $S\subseteq C(T^*)$, and hence $C(T^*) = S$, as required.
 We can similarly show that $C'(T^*)=S'$.
 \end{itemize}
 
Suppose that $C(T)$ is a strict subset of $S$.
Comparing $T$ to $T^{*}$, we find that 
$T_2$ is a strict subset of $E(S,S')$ so
$T$ is not maximal.
Similarly, if  $C'(T)$ is a strict subset of $S'$ then $T$ is not maximal.
 
 {\bf Proof of (3):\quad} It is immediate from  Observation~\ref{obs:emptytype}
 that 
 $B(T)\subseteq \Gamma(C(T)) \cap \Gamma(b)$ and
$B'(T)\subseteq \Gamma(C'(T)) \cap \Gamma(b)$. 

Suppose that $B(T)$ is a strict subset of 
$\Gamma(C(T)) \cap \Gamma(b)$. We will show that $T$ is not maximal.
To this end, 
let $v$ be any vertex in $\Gamma(C(T)) \cap \Gamma(b)\setminus B(T)$ 
and
consider the type 
$T^* = (T_1 \cup \{(b,v)\},T_2,T_3)$.
 Observation~\ref{obs:emptytype} shows that $T^*$ is non-empty, so $T$ is not maximal.
  Similarly, if $B'(T)$ is a strict subset of $\Gamma(C'(T)) \cap \Gamma(b)$ then $T$ is not maximal. 
 
 {\bf Proof of (4):\quad} It is immediate from Observation~\ref{obs:emptytype} and the definition of a type
 that 
 $A(T) \subseteq \Phi(B(T)) \cap \Gamma(g)$ and 
$A'(T) \subseteq \Phi(B'(T)) \cap \Gamma(g)$. 
 If  either of these  subset inclusions is strict then, as in the proof of (3), it is straightforward to see that $T$ is not maximal.
\end{proof}

\begin{lem}\label{lem:necessary2}
Let $T $ be a maximal type.
Then $C(T)$ and $C'(T)$ are 
both in the set $$ \{  \{b\} ,  \{r_1,b\} ,  \{r_2,b\}   \{r_1, r_2, b,g\} \}.$$
\end{lem}
\begin{proof}
We will prove this for $C(T)$. The argument for $C'(T)$ is the same.
From Observation~\ref{obs:emptytype}, 
$C(T)$  is a (not necessarily strict) subset of $\Gamma(b) = \{r_1,r_2,b,g\}$
(and it is non-empty).
\begin{itemize}
\item If $g\in C(T)$ then  $\Gamma(C(T))\cap \Gamma(b) = \{b\}$
so, by  item~(2) of Lemma~\ref{lem:necessary1},
$C(T)=\Gamma(b) = \{r_1,r_2,b,g\}$.
\item If $r_1\in C(T)$ and $r_2\in C(T)$ then $\Gamma(C(T))\cap \Gamma(b) = \{b\}$
so,  again,
$C(T)=\Gamma(b) = \{r_1,r_2,b,g\}$.
\item If $C(T)= \{r_1\}$ 
then $\Gamma(C(T))\cap \Gamma(b) = \{r_1,b\}$
so, by  item~(2) of Lemma~\ref{lem:necessary1},
$C(T)= \{r_1,b\}$, which is a contradiction.
\item Similarly, the case $C(T)=\{r_2\}$ gives a contradiction.
\end{itemize}

This covers all possible cases.
\end{proof}

\begin{defn}
For $i\in [6]$ let $X_i\subseteq V(H_k)$. We say that the types $(E(X_1,X_2), E(X_3,X_4), E(X_5,X_6))$ and $(E(X_6,X_5), E(X_4,X_3), E(X_2,X_1))$ are \emph{symmetric} to each other.
\end{defn}
Note that if $\type$ and $\type'$ are symmetric to each other it holds that $N(\type)=N(\type')$.

\begin{table}[h!]
{\centering
\caption{Maximal types of the homomorphisms in $\calH(J,\boldS_J),H_k)$.}
\label{tab:configs}
\setlength{\tabcolsep}{3pt}
\newcommand\Tstrut{\rule{0pt}{2.6ex}} 
\newcounter{rowno}
\setcounter{rowno}{0}
\begin{tabular}{@{}>{$\type_{\stepcounter{rowno}\therowno}$}lcccccc|r@{}}
\multicolumn{1}{l}{}&$A(T)$ & $B(T)$ & $C(T)$ & $C'(T)$ & $B'(T)$ & $A'(T)$ & $\hatN(\type)$\\
	\hline
	\Tstrut
&$\{b\} \cup \calY$& $\{r_1,r_2,b,g\}$ 
& $\{b\}$ & $\{b\}$ 
& $\{r_1,r_2,b,g\}$ & $\{b\} \cup \calY$ 
& $(4+k)^{pt}\cdot 1^{qt}\cdot (4+k)^{pt}$\\
&$\{b\} \cup \calY$& $\{r_1,r_2,b,g\}$ 
& $\{b\}$ & $\{r_1,b\}$ 
& $\{r_1,b\}$ & $\{b\}$ 
& $(4+k)^{pt}\cdot 2^{qt}\cdot 2^{pt}$\\
&$\{b\} \cup \calY$& $\{r_1,r_2,b,g\}$ 
& $\{b\}$ & $\{r_2,b\}$ 
& $\{r_2,b\}$ & $\{b\}$ 
& $(4+k)^{pt}\cdot 2^{qt}\cdot 2^{pt}$\\
&$\{b\} \cup \calY$& $\{r_1,r_2,b,g\}$ 
& $\{b\}$ & $\{r_1,r_2,b,g\}$ 
& $\{b\}$ & $\{b\}$ 
& $(4+k)^{pt}\cdot 4^{qt}\cdot 1^{pt}$\\
&$\{b\}$& $\{r_1,b\}$ 
& $\{r_1,b\}$ & $\{r_2,b\}$ 
& $\{r_2,b\}$ & $\{b\}$ 
& $2^{pt}\cdot 3^{qt}\cdot 2^{pt}$\\
&$\{b\}$& $\{r_1,b\}$ 
& $\{r_1,b\}$ & $\{r_1,b\}$ 
& $\{r_1,b\}$ & $\{b\}$ 
& $2^{pt}\cdot 4^{qt}\cdot 2^{pt}$\\
&$\{b\}$& $\{r_2,b\}$ 
& $\{r_2,b\}$ & $\{r_2,b\}$ 
& $\{r_2,b\}$ & $\{b\}$ 
& $2^{pt}\cdot 4^{qt}\cdot 2^{pt}$\\
&$\{b\}$& $\{r_1,b\}$ 
& $\{r_1,b\}$ & $\{r_1,r_2,b,g\}$ 
& $\{b\}$ & $\{b\}$ 
& $2^{pt}\cdot 6^{qt}\cdot 1^{pt}$\\
&$\{b\}$& $\{r_2,b\}$ 
& $\{r_2,b\}$ & $\{r_1,r_2,b,g\}$ 
& $\{b\}$ & $\{b\}$ 
& $2^{pt}\cdot 6^{qt}\cdot 1^{pt}$\\
&$\{b\}$& $\{b\}$ 
& $\{r_1,r_2,b,g\}$ & $\{r_1,r_2,b,g\}$ 
& $\{b\}$ & $\{b\}$ 
& $1^{pt}\cdot 9^{qt}\cdot 1^{pt}$\\
\end{tabular}
}

{\it Note: }{Recall that $p$, $q$ and $t$ are the parameters of $J$ where $p$ and $q$ satisfy Constraint~\ref{constraint:pq}. Each line corresponds to a type 
		$\Bigl(E\bigl(A(T),B(T)\bigr), E\bigl(C(T),C'(T)\bigr), E\bigl( B'(T),A'(T)\bigr)\Bigr)$. To shorten the notation we set $\calY=\{y_i \mid i\in [k]\}$.}
\end{table}

\begin{lem}\label{lem:tableType}
All maximal types are listed in Table~\ref{tab:configs} (except for those that are symmetric to a listed type). Furthermore, for each listed type $\type$ we give the corresponding value for $\hatN(\type)$.
\end{lem}
\begin{proof}
First, Lemma~\ref{lem:necessary2} gives the $4$~possibilities for $C(T)$ and $C'(T)$.
Up to symmetry, this gives the 10~possibilities listed in the table.

Next, for each of the 10~possibilities, we use items~(3) and~(4) of Lemma~\ref{lem:necessary1} to compute
the corresponding sets $A(T)$, $B(T)$, $B'(T)$ and $A'(T)$.

Now item~(1) of Lemma~\ref{lem:necessary1} guarantees that
$T_1 = E(A(T),B(T))$, $T_2 = E(C(T),C'(T))$ and $T_3 = E(B'(T),A'(T))$. 
So 
$$ \hatN( T)=\abs{ E(A(T),B(T))}^{pt} \abs{E(C(T),C'(T)) }^{qt} \abs{ E(B'(T),A'(T))}^{pt}.$$
These quantities are all computed in the table.
 \end{proof}

Let $\type_1, \dots, \type_{10}$ be the types as given in Table~\ref{tab:configs}.
\begin{lem}\label{lem:T4dominance}
Let $k$ be a positive integer. Then there  is a $\gamma\in (0,1)$ and positive integers $p$ and $q$ that satisfy Constraint~\ref{constraint:pq} such that, for all $i\in[10]$ except $i= 4$ and all positive integers $t$, we have $\hatN(\type_i)\le \gamma^t \hatN(\type_4)$.
\end{lem}
\begin{proof}
We choose integers $p,q \ge 32+12k$ ($p$ and $q$ satisfy Constraint~\ref{constraint:pq}) such that 
\begin{equation}\label{eq:fracqp1}
\log_4 (4+k) < \frac{q}{p} < \log_{9/4}(4+k).
\end{equation}
This is possible as $\log_4 (4+k)<\log_{9/4}(4+k)$ for all $k>0$. Suppose that $\type$ and $\type'$ are types listed in Table~\ref{tab:configs} which are distinct from $\type_4$ and have the property that $\hatN(\type')<\hatN(\type)$. Then the sought-for bound automatically holds for $\type'$ if it holds for $\type$. 

We check the sought-for bound for each $i\in[10]$, $i\neq 4$:

\begin{list}{}%
{\renewcommand\makelabel[1]{#1:\hfill}%
   \settowidth\labelwidth{\makelabel{$\type_2$ (and $\type_3$)\quad}}%
   \setlength\leftmargin{\labelwidth}
   \addtolength\leftmargin{\labelsep}}
\item[$\type_1$]  $\frac{\hatN(\type_1)}{\hatN(\type_4)}=(4+k)^{pt}(1/4)^{qt} < \gamma^t$ is fulfilled for some sufficiently large $\gamma<1$ if and only if $(4+k)^{p}/4^{q}<1$. This is true as $\log_4 (4+k)<\frac{q}{p}$ by (\ref{eq:fracqp1}).
\item[$\type_2$ (and $\type_3$)]  $\frac{\hatN(\type_2)}{\hatN(\type_4)}=2^{pt}(1/2)^{qt} < \gamma^t$ is fulfilled for some sufficiently large $\gamma<1$ if and only if $2^{p}/2^{q}<1$. This is true as $q>p$ by (\ref{eq:fracqp1}).
\item[$\type_5$]  $\hatN(\type_5) <\hatN(\type_6)$.
\item[$\type_6$ (and $\type_7$)]  $\frac{\hatN(\type_6)}{\hatN(\type_4)}=(4/(4+k))^{pt} < \gamma^t$ is fulfilled for $4/(4+k)\le 4/5<\gamma<1$.
\item[$\type_8$ (and $\type_9$)]  $\frac{\hatN(\type_8)}{\hatN(\type_4)}=(2/(4+k))^{pt}(3/2)^{qt} < \gamma^t$ is fulfilled for some sufficiently large $\gamma<1$ if and only if $(2/(4+k))^{p}(3/2)^{q}<1$. This is true as $\frac{q}{p}<\log_{9/4} (4+k)<\log_{3/2} ((4+k)/2)$ by (\ref{eq:fracqp1}) and for all $k>0$.
\item[$\type_{10}$]  $\frac{\hatN(\type_{10})}{\hatN(\type_4)}=(1/(4+k))^{pt}(9/4)^{qt} < \gamma^t$ is fulfilled for some sufficiently large $\gamma<1$ if and only if $(1/(4+k))^{p}(9/4)^{q}<1$. This is true as $\frac{q}{p} < \log_{9/4} (4+k)$ by (\ref{eq:fracqp1}).
\end{list}
\end{proof}

Now we have collected all properties of the gadget graph $J$ that we need to prove Lemma~\ref{lem:MixedTreeHardness3}. We will see that this lemma is the final piece to show the classification for graphs of girth at least $5$ stated in Theorem~\ref{thm:RetMain}. In the statement of the lemma we refer to the graph $H_k'$ as depicted in Figure~\ref{fig:Hkprime}.

\begin{figure}[h!]\centering
{\def\scaleFactor{1}
\begin{tikzpicture}[scale=1, every loop/.style={min distance=10mm,looseness=10}]

	\filldraw (0,0) node (r1){} circle[radius=3pt] --++ (0:1cm) node (b){} circle[radius=3pt] --++ (0:1cm) node (r2){} circle[radius=3pt];
	\filldraw (b) --++ (-90:1cm) node (g){} circle[radius=3pt] --++ (-130:1cm) node (y1) {} circle[radius=3pt];
	\filldraw (g.center) --++ (-50:1cm) node(yw) {} circle[radius=3pt];

	\path[-] (r1.center) edge  [in=125,out=55,loop] node {} ();
	\path[-] (b.center) edge  [in=125,out=55,loop] node {} ();
	\path[-] (r2.center) edge  [in=125,out=55,loop] node {} ();
	\node at ($(g)+(0,-.8cm)$) {$\ldots$};
	
	\node at ($(g)+(0,-.8cm)$) {$\ldots$};
	\node at ($(r1)+(0,.8cm)$) {$r_1$};
	\node at ($(b)+(0,.8cm)$) {$b$};
	\node at ($(r2)+(0,.8cm)$) {$r_2$};
	\node at ($(g)+(-.3cm,0)$) {$g$};
	\node at ($(y1)+(0cm,-.4cm)$) {$y_1$};
	\node at ($(yw)+(0cm,-.4cm)$) {$y_k$};
	
\end{tikzpicture}
}
\caption{The graph $H'_k$.}
\label{fig:Hkprime}
\end{figure}
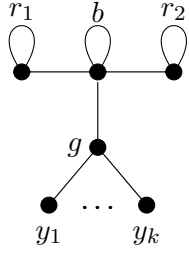

\begin{lem}\label{lem:MixedTreeHardness3}
Let $H$ be graph that has a looped vertex $b$ such that, for some positive integer $k$, $H'_k$ (see Figure~\ref{fig:Hkprime}) is a subgraph of $H[\Gamma^2(b)]$, and $H[\Gamma^2(b)]$ in turn is a subgraph of $H_k$ (see Figure~\ref{fig:Hk}). Then $\sat \leap \OALHom{H}$.
\end{lem}
\begin{proof}
We use a reduction from $\largecut$, which is known to be $\sat$-hard (see~\cite{DGGJApprox}). A \emph{cut} of a graph $G$ is a partition of $V(G)$ into two subsets (the order of this pair is ignored) and the size of a cut is the number of edges that have exactly one endpoint in each of these two subsets. 

\prob
{
$\largecut$.
}
{
An integer $K\ge 1$ and a connected graph $G$ in which every cut has size at most $K$.
}
{
The number of size-$K$ cuts in $G$.
}

Let $G$ and $K$ be an input to $\largecut$, $n$ be the number of vertices of $G$ and $\eps\in (0,1)$ be the parameter of the desired precision of approximation in the AP-reduction. From $G$ we construct an input $(G',\boldS)$ to $\OALHom{H}$ by introducing vertex and edge gadgets. 
By the assumption of the lemma, the vertex~$b$ of~$H$ has $\Gamma(b) = \{b,r_1,r_2,g\}$ where 
$b$, $r_1$ and $r_2$ are looped and $g$ is not
and $\Gamma(g) = \{b,y_1,\ldots,y_k\}$ with $k\geq 1$.

Let $p$, $q$ be positive integers that are chosen such that they fulfil Constraint~\ref{constraint:pq} and (\ref{eq:fracqp1}). Note that $p$ and $q$ only depend on $k$ which is a parameter of the fixed graph $H$ and therefore do not depend on the input $G$. 
We will define the parameter~$t$ of the gadget graph~$J$ to be $t=n^4$.
We also define a new parameter $s=n+1$. 

For each vertex $v\in V(G)$ we introduce a vertex gadget $G'_v$ which is a graph $J$ with parameters $p$, $q$ and $t$ as given in Definition~\ref{def:GadgetJ}. We denote the corresponding sets $A,B,C,C',B',A'$ by $A_v,B_v,C_v,C'_v,B'_v$ and $A'_v$, respectively. It is fine to keep the notation for the remaining vertices as $\alpha$, $\alpha'$ and $\beta$ as technically these vertices can be thought of as identical vertices over all gadgets because of their pinning. We say that two gadgets $G'_u$ and $G'_v$ are adjacent if $u$ and $v$ are adjacent in $G$.

We connect vertex gadgets as follows. For every edge $e=\{u,v\}\in E(G)$ we introduce an edge gadget as follows. We introduce two 
size-$s$ independent sets,
denoted by $S_e$ and $S'_e$. We set $V'_e=S_e\cup S'_e$. As shown in Figure~\ref{fig:EdgeGadget1} we construct the set of edges
\[
E'_e = (\ucp{C_u}{S_e}) \cup (\ucp{C'_u}{S'_e})\cup (\ucp{C_v}{S'_e}) \cup (\ucp{C'_v}{S_e})\cup (\ucp{\{\beta\}}{S_e}) \cup (\ucp{\{\beta\}}{S'_e}).
\]

\begin{figure}[h!]\centering
{\def\scaleFactor{1}
\begin{tikzpicture}[scale=.7, every loop/.style={min distance=10mm,looseness=10}]

\draw (-6,2) -- (0,1.5);
\draw (-6,2) -- (0,-1.5);
\draw (-6,-2) -- (0,-1.5);
\draw (-6,-2) -- (0,1.5);

\draw (6,2) -- (0,1.5);
\draw (6,2) -- (0,-1.5);
\draw (6,-2) -- (0,-1.5);
\draw (6,-2) -- (0,1.5);

\coordinate (b) at (0,-3);
\filldraw (-1,0) -- (b) circle[radius=3pt] -- (1,0);
\node at ($(b)+(-1cm,0)$) {$\beta\rightarrow b$};

\draw (-6,-4) -- (0,-4.5);
\draw (-6,-4) -- (0,-7.5);
\draw (-6,-8) -- (0,-7.5);
\draw (-6,-8) -- (0,-4.5);

\draw (6,-4) -- (0,-4.5);
\draw (6,-4) -- (0,-7.5);
\draw (6,-8) -- (0,-7.5);
\draw (6,-8) -- (0,-4.5);

\filldraw (-1,-6) -- (b) circle[radius=3pt] -- (1,-6);

\draw (-6,0) ellipse (1cm and 2cm) [fill=white];
\draw (0,0) ellipse (1cm and 1.5cm) [fill=white];
\draw (6,0) ellipse (1cm and 2cm) [fill=white];

\node at (-6,0) {$C_u$};
\node at (0,0) {$S_e$};
\node at (6,0) {$C'_v$};

\draw (-6,-6) ellipse (1cm and 2cm) [fill=white];
\draw (0,-6) ellipse (1cm and 1.5cm) [fill=white];
\draw (6,-6) ellipse (1cm and 2cm) [fill=white];

\node at (-6,-6) {$C'_u$};
\node at (0,-6) {$S'_e$};
\node at (6,-6) {$C_v$};
\end{tikzpicture}
}
\caption{The edge gadget for the edge $e=\{u,v\}$.}
\label{fig:EdgeGadget1}
\end{figure}
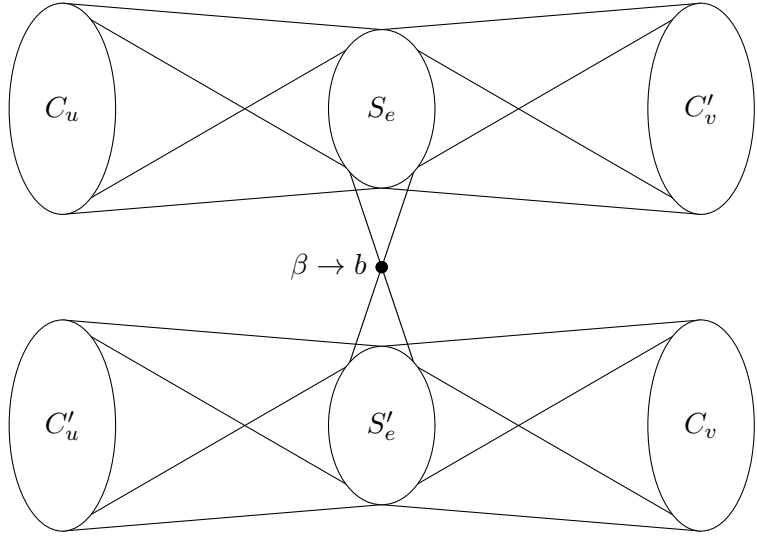

Putting the pieces together, $G'$ is the graph with 
\[
V(G')= \bigcup_{v\in V(G)} V(G'_v) \cup \bigcup_{e\in E(G)} V'_e
\qquad\text{ and }\qquad
E(G')= \bigcup_{v\in V(G)} E(G'_v) \cup \bigcup_{e\in E(G)} E'_e.
\]
Finally, we define the vertex lists $S_\alpha=S_{\alpha'}=\{g\}$, $S_\beta=\{b\}$ and $S_v=V(H)$ for all $v\in V(G')\setminus\{\alpha, \alpha', \beta\}$. Then $\boldS=\{S_v \mid v\in V(G')\}$. This completes the definition of the instance $(G',\boldS)$.

The pinning of the vertex $\beta$ (via the list $S_\beta$)
ensures that every homomorphism from $(G',\boldS)$ to $H$ is a homomorphism from $(G',\boldS)$ to $H[\Gamma^2(b)]$. By the assumption of the lemma, $H[\Gamma^2(b)]$ is a subgraph of $H_k$. We make a case distinction based on the graph $H[\Gamma^2(b)]$. 

\bigskip
\noindent {\bf Case 1: $H[\Gamma^2(b)]=H_k$.}
Let $h$ be a homomorphisms from $(G',\boldS)$ to $H$, $v$ be some vertex of $G$ and $G'_v$ be the corresponding vertex gadget. Then by our definition of $(G',\boldS)$ we observe that $h\vert_{V(G'_v)}$ corresponds to a homomorphism from $(J, \boldS_J)$ to $H_k$ and therefore has a type. 

We say that a homomorphism from $(G',\boldS)$ to $H$ is \emph{full} if its restriction to each vertex gadget is either of type $\type_4$ (from Table~\ref{tab:configs}) or of its symmetric type (let us call it $\type'_4$). Each full homomorphism $h$ defines a cut as it partitions $V(G)$ into those vertices $v$ for which $h\vert_{G'_v}$ has type $\type_4$ and those for which $h\vert_{G'_v}$ has type $\type'_4$. We say that a full homomorphism is \emph{$K$-large} if the size of the corresponding cut is equal to $K$, otherwise we say that the homomorphism is \emph{$K$-small}.
Consider a full homomorphism $h$ from $(G',\boldS)$ to $H_k$.
\begin{itemize}
\item
For an edge $e=\{u,v\}$ of $G$ suppose that $h\vert_{G'_u}$ has type $\type_4$ and $h\vert_{G'_v}$ has type $\type'_4$. Note that by the definition of the edge gadget, we have $h(S_e)\subseteq \Gamma(h(C_u))\cap \Gamma(h(C'_v))$. Then the vertices in $S_e$ can be mapped to any of the $4$ neighbours of $b$, whereas all vertices in $S_e'$ have to be mapped to $b$ (since $h(S_e')\subseteq \Gamma(h(C'_u))\cap \Gamma(h(C_v))$ where $C'_u=C_v =\{r_1, r_2, b,g\}$ and $b$ is the sole common neighbour of $r_1$, $r_2$, $b$ and $g$).
\item
Suppose instead that $h\vert_{G'_u}$ and $h\vert_{G'_v}$ have the same type $\type_4$ or $\type_4'$. Then the homomorphism $h$ has to map the vertices in both $S_e$ and $S_v$ to $b$. 
\end{itemize}
Thus, every pair of adjacent gadgets of different types contributes a factor of $4^s$ to the number of full homomorphisms, whereas every pair of adjacent gadgets of the same type only contributes a factor of $1$. Recall the definition of $N(\type)$ as the number of homomorphisms from $(J,\boldS_J)$ to $H_k$ that have type $\type$. Then for $\ell\ge 1$ every size-$\ell$ cut of $G$ arises in $2\cdot N(\type_4)^n\cdot 4^{s\ell}$ ways as a full homomorphism from $(G',\boldS)$ to $H_k$.

Let $N$ be the number of solutions to $\largecut$ with input $G$ and $K$ (our goal is to approximate this number). We partition the homomorphisms from $(G',\boldS)$ to $H_k$ into three different sets. $Z^*$ is the number of $K$-large (full) homomorphisms, $Z_1$ is the number of homomorphisms that are full but $K$-small and $Z_2$ is the number of non-full homomorphisms.
Then we have $N= Z^*/(2 N(\type_4)^n 4^{sK})$ and $\hom{(G',\boldS)}{H} = \hom{(G',\boldS)}{H_k} = Z^* +Z_1 +Z_2$. Thus it remains to show that $(Z_1 +Z_2)/(2 N(\type_4)^n 4^{sK})\le 1/4$  for our choice of $p$, $q$, $t$ and $s$. Under this assumption we then have $\displaystyle\hom{(G',\boldS)}{H}/(2 N(\type_4)^n 4^{sK}) \in [N,N+1/4]$ and a single oracle call to determine $\hom{(G',\boldS)}{H}$ with precision $\delta=\eps/21$ suffices to determine $N$ with the sought-for precision as demonstrated in~\cite[Proof of Theorem 3]{DGGJApprox}.

Now we prove $(Z_1 +Z_2)/(2 N(\type_4)^n 4^{sK})\le 1/4$.
As there are at most $2^n$ ways to assign a type $\type_4$ or $\type'_4$ to the $n$ vertex gadgets in $G'$ we have $Z_1 \le 2^n \cdot N(\type_4)^n\cdot 4^{s(K-1)}$.
We next obtain the following bound since $s=n+1$:
\[
\frac{Z_1}{2 N(\type_4)^n 4^{sK}}\le \frac{2^n N(\type_4)^n 4^{s(K-1)}}{2 N(\type_4)^n 4^{sK}} = \frac{2^n}{2\cdot 4^s}\le \frac{1}{8}.
\]
We obtain a similar bound for $Z_2$. From Lemmas~\ref{lem:maxType},~\ref{lem:tableType} and~\ref{lem:T4dominance} we know that for our choice of $p$ and $q$ there exists $\gamma\in(0,1)$ such that for every type $\type$ that is not $\type_4$ or $\type_4'$ we have $\hatN(\type)\le \gamma^t\hatN(\type_4)$. Using Lemma~\ref{lem:hatN} this gives $N(\type) \le 2\gamma^t N(\type_4)$ for sufficiently large $t$ with respect to $p$, $q$ and $k$ (which only depend on $H$ but not on the input $G$). Since $t=n^4$ we can assume that $t$ is sufficiently large with respect to $p$ and $q$ as otherwise the input size is bounded by a constant (in which case we can solve $\largecut$ in constant time).

For each type $T=(T_1,T_2,T_3)$, the cardinality of each set~$T_i$
is bounded above by $2\abs{E(H_k)}=32+12k$ and hence there are at most $\left(2^{32+12k}\right)^3$ different types. Furthermore, as $H_k$ has $8+k$ vertices, there are at most $(8+k)^{2sn^2}$ possible functions from the at most $2sn^2$ vertices in $\bigcup_{e\in E(G)} (S_e\cup S'_e)$ to vertices in $H_k$. Since $t=n^4$ and $s=n+1$ we obtain
\begin{align*}
\frac{Z_2}{2 N(\type_4)^n 4^{sK}}
&\le \frac{\left(2^{32+12k}\right)^{3n}\cdot N(\type_4)^{n-1}\cdot 2\gamma^tN(\type_4)\cdot(8+k)^{2sn^2}}{2 N(\type_4)^n 4^{sK}} \\
&= \gamma^t\cdot\frac{\left(2^{32+12k}\right)^{3n}(8+k)^{2sn^2}}{4^{sK}}\le \frac{1}{8}.
\end{align*}
The last inequality holds 
for sufficiently large~$n$
as \[\frac{\left(2^{32+12k}\right)^{3n}(8+k)^{2sn^2}}{4^{sK}}\le C^{n^3}\] for some positive constant $C$ that only depends on $H$,
  but not on the input $G$, whereas $t=n^4$.
{\bf(End of Case 1)}

\bigskip
\noindent {\bf Case 2: $H[\Gamma^2(b)]\neq H_k$.}
By the assumption of the lemma, $H[\Gamma^2(b)]$ is a subgraph of $H_k$.
Let $\calH'$ be the set of homomorphisms in $\calH((J,\boldS_J),H_k)$ that are homomorphisms from $J$ to $H[\Gamma^2(b)]$. Then for each type $T$ the number of homomorphisms in $\calH'$ of type $T$ is at most the number of homomorphisms in $\calH((J,\boldS_J),H_k)$ that have type $T$.

Note that the type $\type_4$ (and its symmetric type) only uses vertices and edges from $H'_k$ and we know that $H'_k$ is a subgraph of $H[\Gamma^2(b)]$ by the assumption of the lemma. Therefore each homomorphism which is of type $\type_4$ is also in $\calH'$ (their number remains unchanged). The analysis is then analogous to that of Case 1. (The number of $K$-large and $K$-small homomorphisms stays the same whereas the number of non-full homomorphisms can only decrease as we only need to consider a subset of the previous types and the number of homomorphisms that have a particular type can only decrease.)
{\bf(End of Case 2)}

\end{proof}

\subsection{Putting the Pieces together}\label{sec:MainTheorem}
Now finally we have all the tools at hand to prove the main classification result for counting retractions to graphs of girth at least $5$, which we restate at this point.
{\renewcommand{\thethm}{\ref{thm:RetMain}}
\begin{thm}
\ThmRetMain
\end{thm}
\addtocounter{thm}{-1}
}
\begin{proof}
As in the proof of Theorem~\ref{thm:RetIrreflexive}, the fact that the classification extends from connected graphs to graphs with multiple connected components follows from Remark~\ref{rem:Connectivity}. Now assume without loss of generality that $H$ is a connected graph. If $H$ is an irreflexive graph then the statement of the theorem follows from the slightly more general Theorem~\ref{thm:RetIrreflexive} (in the irreflexive case we only require $H$ to be square-free). 

Now suppose that $H$ has at least one looped vertex. From Observation~\ref{obs:HomToRetToLHom} we know that, in general, hardness results for $\Hom{H}$ carry over to $\Ret{H}$ and easiness results carry over from $\LHom{H}$. Then, by Theorem~\ref{thm:LHomTricho}, $\Ret{H}$ is in $\FP$ if $H$ is a single looped vertex or a single looped edge.  
Otherwise, since it is  triangle-free, $H$ cannot be a complete reflexive graph, so
$\Ret{H}$ is $\bis$-hard with respect to AP-reductions by Theorem~\ref{thm:HomBIS}. The $\bis$-easiness for partially bristled reflexive paths follows from our Lemma~\ref{lem:BristledPathEasiness}. Theorem~\ref{thm:LHomTricho} implies that $\Ret{H}$ is always $\sat$-easy.

It remains to show the $\sat$-hardness result for graphs $H$ that have at least one looped vertex but are not partially bristled reflexive paths. To this end we distinguish two disjoint cases:
\begin{enumerate}
\item \label{item:MainHardness1} Suppose that every unlooped vertex in $H$ has degree $1$. Let $H^*$ be the subgraph induced by the looped vertices of $H$. 
As all unlooped vertices have degree $1$, the fact that $H$ is connected implies that $H^*$ is connected. Recall that $\WR{q}$ is a reflexive star with $q$ leaves. Then, in general, $H^*$ is either a reflexive path, a reflexive cycle or it contains a subgraph $\WR{q}$ for some $q\ge3$. 
\begin{enumerate}
\item Suppose that $H^*$ is a reflexive path $u_1, \dots, u_t$. By the fact that $H$ is not a partially bristled reflexive path and all unlooped vertices have degree $1$, it follows that 
either some $u_i$ has more than one unlooped neighbour or
at least one of the endpoints $u_1$ or $u_t$ has an unlooped neighbour. Then $\sat$-hardness follows from Lemma~\ref{lem:MixedTreeHardness2b}.
\item If $H^*$ is a reflexive cycle, then by the fact that every unlooped vertex in $H$ has degree $1$, it holds that $H^*$ is the only cycle in $H$. Then $H$ is a pseudotree and, as $H$ has girth at least $5$, the reflexive cycle $H^*$ has length at least $5$. Therefore $\DRet{H}$ is $\NP$-complete by Theorem~\ref{thm:DRetPseudotree} and it follows that $\Ret{H}$ is $\sat$-hard under AP-reductions by~\cite[Theorem 1]{DGGJApprox}.
\item If $H^*$ contains a subgraph $\WR{q}$ for some $q\ge3$, then $H$ contains a looped vertex with at least $3$ looped neighbours apart from itself. As $H$ is triangle-free, the subgraph $\WR{q}$ is induced and $\sat$-hardness follows either from Lemma~\ref{lem:MixedTreeHardness2a} or Lemma~\ref{lem:MixedTreeHardness2b}.
\end{enumerate} 
\item \label{item:MainHardness2}Suppose there exists an unlooped vertex in $H$ that has degree at least $2$. As $H$ is connected and contains at least one looped vertex, it follows that there exists a looped vertex $b$ with an unlooped neighbour $g$, which has degree $k+1$ for some $k\ge 1$, i.e.~has neighbours $y_1,\dots, y_k$ apart from $b$. Then $H[\Gamma(b)]$ is isomorphic to a $2$-Wrench, or otherwise hardness follows from Lemma~\ref{lem:MixedTreeHardness2b}. Therefore $b$ has exactly $2$ looped neighbours apart from itself. Let us call them $r_1$ and $r_2$. Then, as $H$ has girth at least $5$, the vertices $\{r_1,r_2,b,g,y_1,\dots,y_k\}$ are distinct. This shows that $V(H_k')\subseteq V(H[\Gamma^2(b)])$ and $E(H_k')\subseteq E(H[\Gamma^2(b)])$, i.e.~that $H'_k$ (see Figure~\ref{fig:Hkprime}) is a subgraph of $H[\Gamma^2(b)]$. 

For $i=1,2$ the following hold:
\begin{enumerate}
\item \label{item: MainHardness2_1}Apart from $b$ and $r_i$ itself, the vertex $r_i$ has at most $1$ other looped neighbour, or otherwise hardness follows either from Lemma~\ref{lem:MixedTreeHardness2a} or from Lemma~\ref{lem:MixedTreeHardness2b}. 
\item \label{item: MainHardness2_2}If $r_i$ has an unlooped neighbour, then $H[\Gamma(r_i)]$ is isomorphic to a $2$-Wrench, or otherwise hardness follows from Lemma~\ref{lem:MixedTreeHardness2b}.
\end{enumerate}
From items~\ref{item: MainHardness2_1} and~\ref{item: MainHardness2_2} it follows that for $i=1,2$ the vertex $r_i$ has at most one looped and one unlooped neighbour apart from $b$ and $r_i$ itself. (If they exist let us call the looped neighbour $w_i$ and the unlooped neighbour $d_i$.) Therefore, $V(H[\Gamma^2(b)])\subseteq \{w_1,d_1,r_1,w_2,d_2,r_2, b, g,y_1,\dots,y_k\}\subseteq V(H_k)$.

Note that $d_1$, $d_2$ and $g$ are unlooped 
vertices in~$H$.
Furthermore, for $i=1,2$ we have shown the following
\begin{itemize}
\item $E(H_k') \subseteq E(H[\Gamma^2(b)])$.
\item $\{w_i,r_i\}\in E(H[\Gamma^2(b)])$ if $w_i\in V(H[\Gamma^2(b)])$.
\item $\{d_i,r_i\}\in E(H[\Gamma^2(b)])$ if $d_i\in V(H[\Gamma^2(b)])$.
\end{itemize}

The edges $E(H'_k)$ together with $\{w_1,r_1\}$, $\{w_2,r_2\}$, $\{d_1,r_1\}$, $\{d_2,r_2\}$ (if these exist) form a tree on the vertices $\Gamma^2(b)$ (Recall that a tree might have loops but no cycles). By the fact that $H$ has girth at least $5$, all named vertices are distinct, and it follows that $E(H[\Gamma^2(b)])\subseteq E(H_k)$, which shows that $H[\Gamma^2(b)]$ is a subgraph of $H_k$.

Summarising, $H_k'$ is a subgraph of $H[\Gamma^2(b)]$ and $H[\Gamma^2(b)]$ is a subgraph of $H_k$ and we can apply Lemma~\ref{lem:MixedTreeHardness3} to obtain $\sat$-hardness.
\end{enumerate}

Items~\ref{item:MainHardness1} and~\ref{item:MainHardness2} cover all graphs $H$ that have at least one looped vertex but are not partially bristled reflexive paths. (Note that item~\ref{item:MainHardness1} includes the case where $H$ is reflexive.)
\end{proof}

\section{Approximately Counting Retractions is as least as hard as Counting Surjective Homomorphisms or Compactions}\label{sec:ComplexityLandscape}

This section studies the place of the problem $\Ret{H}$ within the landscape of a number of closely related counting problems. We now give formal definitions for these problems, which are parameterised by a graph~$H$.

Let $G$ be an irreflexive graph. A homomorphism~$h\from G \to H$ is said to be \emph{surjective} if for every vertex $v\in V(H)$ there is a vertex $u\in V(G)$ such that $h(u)=v$. We use $\sur{G}{H}$ to denote the number of surjective homomorphisms from~$G$ to~$H$.
Similarly, the homomorphism~$h$ is a \emph{compaction} if it is surjective and for every non-loop edge $\{v_1,v_2\} \in E(H)$ there is is an edge $\{u_1,u_2\}\in E(G)$ such that $h(u_1)=v_1$ and $h(u_2)=v_2$. We use $\comp{G}{H}$ to denote the number of compactions from~$G$ to~$H$.

\noindent
\begin{tabular}{@{}p{.475\linewidth}p{.475\linewidth}@{}}
\prob
{
$\SHom{H}.$
}
{
An irreflexive graph $G$.
}
{
$\sur{G}{H}$.
}
&
\prob
{
$\Comp{H}$.
}
{
An irreflexive graph $G$.
}
{
$\comp{G}{H}$.
}
\end{tabular}

We also define the corresponding list versions of these two problems. We use $\sur{(G,\boldS)}{H}$ and $\comp{(G,\boldS)}{H}$ to denote the number of surjective homomorphisms from~$(G,\boldS)$ to~$H$ and the number of compactions from~$(G,\boldS)$ to~$H$, respectively. Note that the list version of the problem $\Ret{H}$ is simply the problem $\LHom{H}$.

\noindent
\begin{tabular}{@{}p{.475\linewidth}p{.475\linewidth}@{}}
\prob
{
$\LSHom{H}$.
}
{
An irreflexive graph $G$ and a collection of lists $\boldS=\{S_v\subseteq V(H)\mid v\in V(G)\}$.
}
{
$\sur{(G,\boldS)}{H}$.
}
&
\prob
{
$\LComp{H}$.
}
{
An irreflexive graph $G$ and a collection of lists $\boldS=\{S_v\subseteq V(H)\mid v\in V(G)\}$.
}
{
$\comp{(G,\boldS)}{H}$.
}
\end{tabular}

Furthermore, we define a generalisation of the problems $\Hom{H}$, $\OALHom{H}$ and $\LHom{H}$. Let $2^{V(H)}= \{S \mid S\subseteq V(H)\}$ be the power set of $V(H)$. For a fixed graph $H$ and a set $\calL\subseteq 2^{V(H)}$ we define

\prob
{
$\Hom{H,\calL}$.
}
{
An irreflexive graph $G$ and a collection of lists $\boldS=\{S_v\in \calL\mid v\in V(G)\}$.
}
{
$\hom{(G,\boldS)}{H}$.
}

As a measure of distance between two distributions $\pi$ and $\pi'$ on a finite universe $\Omega$ we use the \emph{total variation distance} $\tv{\pi}{\pi'}=\frac12 \sum_{\omega\in \Omega} \abs{\pi(\omega)-\pi'(\omega)}$. For a set $A\subseteq \Omega$, $\mathrm{Uni}(A)$ is the uniform distribution on $A$.
Furthermore, $\mathrm{Be}(p)$ is the Bernoulli distribution with parameter $p$. In general, we write $X \sim D$ if a random variable $X$ has distribution $D$.

\subsection{Reductions using a Monte Carlo Approach}\label{sec:MonteCarlo}
The main goal of this section is to prove Corollaries~\ref{cor:CompToRet} and~\ref{cor:SHomToRet}. Together they constitute Theorem~\ref{thm:SHomCompToRet} which states that both $\SHom{H}$ and $\Comp{H}$ are AP-reducible to $\Ret{H}$.

In the following two lemmas we prove some necessary ingredients that we use in the proof of Lemma~\ref{lem:MCAlgo}. 
From Section~\ref{sec:Preliminaries} recall that a RAS for $\Hom{H, \calL}$ is an $(\epsilon,\delta)$-approximation 
for $\Hom{H, \calL}$ with $\delta=1/4$.
First, we point out the well-known fact that 
this can be powered to obtain an $(\epsilon,\delta)$-approximation for smaller~$\delta$.

\begin{lem}\label{lem:Powering}
Let $H$ be a graph and $\calL\subseteq 2^{V(H)}$. 
There is an algorithm $\textsc{CountHom}_{H,\calL}$
which uses oracle access to a RAS for $\Hom{H, \calL}$
and has the following properties.  
\begin{itemize}
\item  It is given an input $(G,\boldS)$ to $\Hom{H, \calL}$ together with accuracy parameters $\eps$
and $\delta$ in $(0,1)$.
\item It returns a natural number~$X$ with $\Pr\left(e^{-\eps} \le \frac{X}{\hom{(G,\boldS)}{H}}\le e^{\eps}\right)\ge 1-\delta$.
\item Its running time is bounded by a polynomial in 
$\eps^{-1}$, $\log\delta^{-1}$, and
the number of vertices of $G$.
\end{itemize}
\end{lem}
\begin{proof}
The lemma is basically the same as \cite[Lemma 6.1]{JVV1986} applied to the problem $\Hom{H, \calL}$. The only difference is that \cite[Lemma 6.1]{JVV1986} gives a precision guarantee of the form \[\Pr\left((1-\eps) \le \frac{X}{\hom{(G,\boldS)}{H}}\le (1+\eps)\right)\ge 1-\delta.\] However, by Observation~\ref{obs:accuracy}, using accuracy parameter $\eps/2$ instead of $\eps$ in \cite[Lemma 6.1]{JVV1986} suffices to obtain the desired result.
 \end{proof}

Second, we point out that if $\calL$ contains the set of singletons $\{ \{v\} \mid v\in V(H)\}$ then
$\Hom{H, \calL}$ is \emph{self-reducible}.
So the technique of Jerrum, Valiant and Vazirani~\cite{JVV1986}
reduces the problem of approximately sampling 
homomorphisms with lists in~$\calL$  to the problem of approximately counting them.
The original notion of self-reducibilty, due to Schnorr~\cite{Schnorr},
relies on careful encodings of instances,
so we use instead the 
more general \emph{self-partitionability} notion 
  of Dyer and Greenhill.  
Dyer and Greenhill show~\cite{DGRndmWalks} that 
the technique of Jerrum, Valiant and Vazirani applies to every self-partitionable problem.
Thus, we get the following lemma.

\begin{lem}\label{lem:JVV}
Let $H$ be a graph and $\calL\subseteq 2^{V(H)}$ such that $\{ \{v\} \mid v\in V(H)\}\subseteq \calL$. 
There is an algorithm $\textsc{SampleHom}_{H,\calL}$
which uses oracle access to a RAS for $\Hom{H,\calL}$ and
has the  following properties.
\begin{itemize}
\item It is given an input $(G,\boldS)$ to $\Hom{H,\calL}$ together with
an accuracy parameter $\eps\in (0,1)$.
\item The distribution $D$ of its outputs satisfies
 $\tv{D}{\mathrm{Uni}\bigl(\calH((G,\boldS),H)\bigr)}\le \eps$.
\item  Its running time is bounded by a polynomial in 
$\log\eps^{-1}$ and the number of vertices of~$G$.
\end{itemize}  
\end{lem}
\begin{proof}
Rather than repeating the (lengthy) formal definition of \emph{self-partitionability} from~\cite{DGRndmWalks},
we state the (self-evident) relevant properties of
$\Hom{H,\calL}$ which imply that $\Hom{H,\calL}$ is self-partitionable.
The lemma follows immediately from~\cite{DGRndmWalks}.

Let $(G,\boldS)$ be an input to $\Hom{H,\calL}$. If $v\in V(G)$ and $s\in S_v$,
then let $\boldS^{v\to s} = \{S^{v\to s}_u \mid u\in V(G)\}$
be defined as follows. 
\[
S^{v\to s}_u = \begin{cases} \{s\}  &\text{ if } u=v   \\ 
S_u &\text{ otherwise.}\end{cases}
\]
Note that $(G,\boldS^{v\to s})$ is a valid input to $\Hom{H,\calL}$ as $\{ \{v\} \mid v\in V(H)\}\subseteq \calL.$

The relevant properties are
\begin{enumerate}
\item If for all $v\in V(G)$ we have $S_v \in \{ \{u\} \mid u\in V(H)\}$
then the function $\tau$ which maps each vertex $v\in V(G)$ to the single element in the corresponding list $S_v$ is the only mapping from~$G$ to~$H$ that respects the lists. It is then easy to check whether $\tau$ is a homomorphism.
Therefore, computing $\hom{(G,\boldS)}{H}$ and sampling from the set of list homomorphisms from $(G,\boldS)$ to $H$ is trivial.
\item If $v\in V(G)$ then 
\begin{equation}
\label{eq:partition}
\calH((G,\boldS),H) = \bigcup_{s\in S_v} \calH((G,\boldS^{v\to s}),H).
\end{equation}
\end{enumerate}
Note that the right-hand-side of~\eqref{eq:partition}
is a union of disjoint sets since  
all of the homomorphisms in $\calH((G,\boldS^{v\to s}),H)$ map $v$ to~$s$.

These properties imply that $\Hom{H,\calL}$ is self-partitionable
in the sense of Dyer and Greenhill,  thus the lemma follows from the technique
of Jerrum, Valiant and Vazirani, as demonstrated in~\cite{DGRndmWalks}.

This completes the proof of the lemma, but for the reader who wants to relate the above
properties to the notation of Dyer and Greenhill, we take the size of
an instance $(G,\boldS)$ to be $\abs{\{v\in V(G) \mid \abs{S_v} > 1 \}}$.
The set of smaller instances $\Xi(G,\boldS)$ considered in~\cite{DGRndmWalks}
can be constructed by fixing any $v\in V(G)$ with $\abs{S_v} > 1$
and then 
setting $\Xi(G,\boldS) = \{ (G,\boldS^{v\to s}) \mid s\in S_v\}$.
The functions $k_\xi$ mentioned in~\cite{DGRndmWalks} can all be taken to be constant
functions, with output~$1$.
The injection $\phi_{ (G,\boldS^{v\to s}) }$ is the identity.
Finally $W( (G,\boldS), \tau)$ is just the indicator function that is~$1$ if $\tau$ is a homomorphism from $(G,\boldS)$ to $H$ and $0$ otherwise.
\end{proof}

Our first goal is Corollary~\ref{cor:CompToRet} which is an AP-reduction from $\Comp{H}$ to $\Ret{H}$. The reduction uses a Monte Carlo Algorithm 
(Algorithm~\ref{algo:CoverageMC}). The algorithm is presented more generally, with lists,
so that we can also use it in the reductions of Corollaries~\ref{cor:LCompToLHom},
\ref{cor:SHomToRet} and \ref{cor:LSHomToLHom2}.
The following observation provides the basis for the algorithm. Let $H$ be a graph, $G$ be an irreflexive graph and $\boldS$ be a corresponding set of lists.
If there is a compaction from $(G,\boldS)$ to $H$ then there exists a set $U\subseteq V(G)$ with $|U| \leq \abs{V(H)} + 2\abs{E(H)}$ and a compaction~$\tau$ from~$G[U]$ to~$H$. Accordingly, we define
\begin{align}
\TGS=\{(U,\tau) \mid\ &U\subseteq V(G), \abs{U}\le \abs{V(H)} + 2\abs{E(H)},\label{eq:TGS}\\\nonumber
&\tau \text{ is a compaction from }G[U]\text{ to }H \text{ such that }\forall u\in U, \tau(u)\in S_u\}
\end{align}
and $\tGS=\abs{\TGS}$. Let $(U_i,\tau_i)_{i\in[\tGS]}$ be an arbitrary indexing of the elements of $\TGS$. For $i\in [\tGS]$ we define
\begin{align}
&\OmGSi=\left\{\sigma \in \calH((G,\boldS),H) \mid \enspace\sigma\vert_{U_i}=\tau_i\right\},\label{eq:OmGSi}\\
&\OmPlus=\left\{(i,\sigma) \mid i\in [\tGS]\text{ and }\sigma\in \OmGSi\right\}, \text{and}\label{eq:OmPlus}\\
&\OmCup=\Bigl\{(i,\sigma) \in \OmPlus \mid \sigma\notin \bigcup_{k=1}^{i-1}\OmGSk\Bigr\}.\label{eq:OmCup}
\end{align}
Note that $\abs{\OmPlus} = \sum_{i\in [\tGS]} \abs{\OmGSi}$. As every element of a set $\OmGSi$ is a compaction from $(G,\boldS)$ to $H$ and every such compaction is contained in a set $\OmGSi$, we have
\begin{align*}
\abs{\OmCup} &= \Bigl\vert\bigcup_{i\in [\tGS]} \OmGSi\Bigr\vert = \Bigl\vert\{\sigma \in \calH((G,\boldS),H) \mid \exists i\in [\tGS] \text{ such that } \sigma\in \OmGSi\}\Bigr\vert \\
&= \comp{(G,\boldS)}{H}.
\end{align*}
It is clear from the definitions that $|\OmCup| \geq |\OmPlus|/\tGS$. Thus, 
\begin{equation}\label{eq:CoverageLB}
\comp{(G,\boldS)}{H} = \abs{\OmCup} \ge \frac{\abs{\OmPlus}}{\tGS}.
\end{equation}
Intuitively, for some fixed graph $H$ and $\calL\subseteq 2^{V(H)}$ we use this lower bound to construct a Monte Carlo algorithm (Algorithm~\ref{algo:CoverageMC}) in the style of~\cite[Algorithm 11.2]{Mitzenmacher2017}, which approximately samples from $\OmPlus$ to approximately compute $\abs{\OmCup}=\comp{(G,\boldS)}{H}$. To this end the algorithm relies on access to a RAS oracle for $\Hom{H,\calL^*}$ where $\calL^*=\calL\cup \{ \{v\} \mid v\in V(H)\}$.

\begin{algorithm}
\caption{Approximate Computation of $\abs{\OmCup}$. Let $H$ be a fixed graph, $\calL\subseteq 2^{V(H)}$ and $\calL^*=\calL\cup \{ \{v\} \mid v\in V(H)\}$. Then $\displaystyle\textsc{CountHom}_{H, \calL^*}$ and $\textsc{SampleHom}_{H, \calL^*}$ are the routines from Lemma~\ref{lem:Powering} and~\ref{lem:JVV}, respectively. Let $\TGS$, $\tGS$, $\OmGSi$, $\OmPlus$ and $\OmCup$ be defined as in  Equations~(\ref{eq:TGS})-(\ref{eq:OmCup}). Note that $(U_i, \tau_i)$ is the $i$'th element of $\TGS$.}
\label{algo:CoverageMC}
\begin{algorithmic}
\Require Irreflexive graph $G$ with lists $\boldS=\{S_v\in \calL \mid v\in V(G)\}$ and $\eps,\ \delta\in (0,1)$.
	\If {$\tGS=0$}
		\State $Y=0$.
	\Else
		\State $\eps'=\frac{\eps}{12}$, $\delta' =\frac{\delta}{2}$, $\delta'' = \frac{\delta'}{\tGS}$.
		\For{$i=1, \dots, \tGS$}
			\State For all $v\in V(G)$, if $v \in U_i$, set $S^i_v =\{\tau_i(v)\}$, otherwise set $S^i_v=S_v$.
			\State $\boldS^i= \{S^i_v \mid v\in V(G)\}$
			\State $\omega_i = \countHom{G,\boldS^i,\eps',\delta''}$. 
		\EndFor
		\State $\displaystyle\omega = \sum_{i=1}^{\tGS} \omega_i$.
		\State $\displaystyle m = \left\lceil 6\tGS\cdot   \frac{\ln(2/\delta')}{{\eps'}^2}\right\rceil$.
		\For{$j=1, \dots, m$}
			\State Choose $i\in [\tGS]$ with probability $\frac{\omega_i}{\omega}$.
			\State $\displaystyle\sigma_j = \sampleHom{G,\boldS^i,\eps'/(2\abs{V(H)}^n)}.$
			\State Let $X_j$ be $1$ in the event $(i,\sigma_j) \in \OmCup$
			 and $0$ otherwise.
		\EndFor
		\State $Y=\frac{\omega}{m}\sum_{j=1}^m X_j$.
	\EndIf
\Ensure $Y$
\end{algorithmic}
\end{algorithm}

\begin{lem}\label{lem:MCAlgo}
Algorithm~\ref{algo:CoverageMC} returns an $(\eps,\delta)$-approximation of $\abs{\OmCup}$ if it has access to a RAS oracle for $\Hom{H,\calL^*}$ and every list in $\boldS$ is an element of $\calL$. For fixed $\delta$, the algorithm runs in time polynomial in $n=\abs{V(G)}$ and $\eps^{-1}$.
\end{lem}
\begin{proof}
First note that given oracle access to a RAS for $\Hom{H,\calL^*}$, the routines $\textsc{CountHom}_{H,\calL^*}$ and $\textsc{SampleHom}_{H, \calL^*}$ exist as shown in Lemmas~\ref{lem:Powering} and~\ref{lem:JVV} (by definition $\calL^*$ contains $\{ \{v\} \mid v\in V(H)\}$). Furthermore, the input to these routines is valid: A list $S_v^i\in \boldS^i$ is either of the form $\{\tau_i(v)\}$ or otherwise $S_v^i=S_v\in\calL$. Therefore, in general, $S_v^i\in \calL^*$. Thus Algorithm~\ref{algo:CoverageMC} is well-defined.

Next we show that the runtime condition is met. Assume $\delta$ to be fixed. Note that we can determine $\TGS$ exactly by enumerating all possible assignments of at most $\abs{V(H)}+2\abs{E(H)}$ vertices of $G$ to the vertices of $H$ and checking whether the resulting assignment is a compaction. Checking can be done in polynomial time and $\tGS=\abs{\TGS}\le \sum_{k=1}^{\abs{V(H)}+2\abs{E(H)}}n^k\in \text{poly}(n)$. It follows that $m\in \text{poly}(n, \eps^{-1})$. The 
runtime of the routine $\countHom{G,\boldS^i,\epsilon', \delta''}$ is in $\text{poly}(n, 1/{\epsilon'})$ by Lemma~\ref{lem:Powering}. Finally, the runtime of $\sampleHom{G,\boldS^i,\eps'/(2\abs{V(H)}^n)}$ is in $\text{poly}(n, \log(1/\epsilon'))$ by Lemma~\ref{lem:JVV}. It is   essential here that the runtime of $\textsc{SampleHom}_{H, \calL^*}$ has logarithmic dependence on the precision parameter as the precision we use is $\eps'/(2\abs{V(H)}^n)$, which is exponential in $n$.

If $\abs{\TGS}=\tGS= 0$ then $\abs{\OmCup}=0$ and the algorithm returns an exact solution. To prove the correctness of the algorithm it remains to show that otherwise it is an $(\eps,\delta)$-approximation. 

Note that by the definition of the $\boldS^i$ in the first part of the algorithm, $\OmGSi=\calH((G,\boldS^i),H)$. Then, by Lemma~\ref{lem:Powering} and the definition of $\delta''$, $\countHom{G,\boldS^i,\eps',\delta''}$ returns a number $\omega_i$ with $\Pr(e^{-\eps'}\abs{\OmGSi} \le \omega_i\le e^{\eps'}\abs{\OmGSi})\ge 1-\delta'/\tGS$. By the union bound, with probability of at least $1-\delta'$, we have
\begin{equation}\label{eq:omAccuracy}
e^{-\eps'}\abs{\OmGSi} \le \omega_i\le e^{\eps'}\abs{\OmGSi} \qquad (\forall i\in[\tGS]).
\end{equation}
The following two subclaims are based on this assumption. Let  $p=\abs{\OmCup}/\abs{\OmPlus}$.

\medskip
\noindent {\bf Subclaim 1:} Assume that (\ref{eq:omAccuracy}) holds. Then for all $j\in[m]$ we have $X_j\sim \mathrm{Be}(p')$ where $e^{-3\eps'}p\le p'\le e^{3\eps'}p$.
\medskip

\noindent {\bf Proof of Subclaim 1:} 
Consider fixed $\omega_1,\ldots,\omega_{\tGS}$
satisfying~\eqref{eq:omAccuracy}.
 Note that the distribution of $(i,\sigma_j)$, conditioned on these, 
does not depend on the index $j$. 
Let $D$ be the distribution (conditioned on 
$\omega_1,\ldots,\omega_{\tGS}$)
such that for all $j\in [m]$ we have $(i,\sigma_j)\sim D$. 
By Lemma~\ref{lem:JVV} and the fact that $\OmGSk=\calH((G,\boldS^k),H)$ we have
\[
D((k,\sigma))= \Pr(\sigma_j=\sigma\,\vert\,i=k)\cdot \Pr(i=k) \le \left(\frac{1}{\abs{\OmGSk}}+\frac{\eps'}{2\abs{V(H)}^n}\right)\cdot\Pr(i=k)
\]
First using $\abs{\OmGSk}\le\abs{V(H)}^n$ and then using Observation~\ref{obs:accuracy} it follows
\[
D((k,\sigma))\le \left(1+\frac{\eps'}{2}\right)\frac{1}{\abs{\OmGSk}}\cdot\Pr(i=k) \le e^{\eps'}\frac{1}{\abs{\OmGSk}}\cdot\Pr(i=k) = e^{\eps'}\frac{1}{\abs{\OmGSk}}\cdot\frac{\omega_k}{\sum_{i=1}^{\tGS} \omega_i}.
\] 
Using the assumption of this subclaim, we obtain
\[
D((k,\sigma))\le e^{\eps'}\frac{1}{\abs{\OmGSk}}\cdot e^{2\eps'}\frac{\abs{\OmGSk}}{\sum_{i\in [\tGS]} \abs{\OmGSi}} = e^{3\eps'}\frac{1}{\abs{\OmPlus}}
\]
and  thus
\[
p'= \Pr(X_j=1) = \sum_{(i,\sigma)\in \OmCup} D((i,\sigma))
\le e^{3\eps'} \sum_{(i,\sigma)\in \OmCup} \frac{1}{\abs{\OmPlus}} = e^{3\eps'}p.
\]
Analogously we obtain the lower bound $e^{-3\eps'}p\le p'$.
{\bf (End of the proof of Subclaim 1.)}

\medskip
\noindent {\bf Subclaim 2:} Assume that (\ref{eq:omAccuracy}) holds. Then $e^{-4\eps'}\abs{\OmCup}\le \E[Y]\le e^{4\eps'}\abs{\OmCup}$.
\medskip

\noindent {\bf Proof of Subclaim 2:}
Consider fixed 
$\omega_1,\ldots,\omega_{\tGS}$
satisfying~\eqref{eq:omAccuracy}.
Conditioned on these,
we have $X_j\sim \mathrm{Be}(p')$ and
\[
\E[Y] = \frac{\sum_{i\in[\tGS]} \omega_i}{m} \sum_{j\in [m]}\E[X_j] = \sum_{i\in[\tGS]} \omega_i \cdot p'.
\]
We now use~\eqref{eq:omAccuracy} as well as Subclaim 1 to obtain
\[
\E[Y] \le e^{\eps'} \sum_{i\in[\tGS]} \abs{\OmGSi} \cdot e^{3\eps'}p = e^{4\eps'}\abs{\OmCup}
\]
and
\[
\E[Y] \ge e^{-\eps'} \sum_{i\in[\tGS]} \abs{\OmGSi} \cdot e^{-3\eps'}p = e^{-4\eps'}\abs{\OmCup}.
\]
\noindent{\bf (End of the proof of Subclaim 2.)} 

Next we show that,
conditioned on computing $\omega_i$'s that satisfy (\ref{eq:omAccuracy}), 
the number of samples~$m$ is sufficiently large.
 First, by Subclaim 1, $X_1,\dots, X_m$ are independent indicator random variables that have distribution $\mathrm{Be}(p')$ and expected value $p'$. By Subclaim 1 and Observation~\ref{obs:accuracy} we have 
\[
 (1-6\eps')p\le e^{-3\eps'}p\le p'\le e^{3\eps'}p\le (1+6\eps')p.
\]
From the definition of $\eps'$ it follows that $\abs{p'-p}\le 6\eps'p\le\eps p/2$ and consequently $p'\ge p/2$. Using this fact and taking into account that by Equation~\eqref{eq:CoverageLB} we have $\tGS\geq 1/p$, it follows that
\begin{align*}
m &= \left\lceil 6\tGS\cdot  \frac{\ln(2/\delta')}{{\eps'}^2}\right\rceil
\ge 6\frac{\ln(2/\delta')}{{\eps'}^2 p}\ge 3 \frac{\ln(2/\delta')}{{\eps'}^2 p'}.
\end{align*}
Thus, we can use~\cite[Theorem 11.1]{Mitzenmacher2017} to obtain
$\Pr\left(\abs{Y-\E[Y]}\ge \eps'\E[Y]\right)\le \delta'$ which is conditioned on the fact that (\ref{eq:omAccuracy}) holds. Now taking into account the fact that, with probability at least $1-\delta'$, 
$\omega_1,\ldots,\omega_{\tGS}$
satisfy~\eqref{eq:omAccuracy}, we have shown that, with probability of at least $(1-\delta')^2\ge 1-\delta$, we have 
\[
\abs{Y-\E[Y]}\le \eps'\E[Y]=\frac{\eps}{12}\E[Y].
\]
By Subclaim 2 and Observation~\ref{obs:accuracy} we also know that 
\[
\abs{\E[Y]-\abs{\OmCup}}\le 8\eps'\abs{\OmCup}= \frac{2\eps}{3}\abs{\OmCup}.
\]
Summarising, with probability of at least $1-\delta$, we have
\begin{align*}
\abs{Y-\abs{\OmCup}}&\le \abs{Y-\E[Y]} + \abs{\E[Y]-\abs{\OmCup}} \le \frac{\eps}{12}\E[Y] + \frac{2\eps}{3}\abs{\OmCup}\\ 
&\le \frac{\eps}{12}\left(\abs{\OmCup} + \frac{2\eps}{3}\abs{\OmCup}\right) + \frac{2\eps}{3}\abs{\OmCup} \le \eps \abs{\OmCup}.
\end{align*}
Hence, $Y$ is an ($\eps,\delta)$-approximation of $\abs{\OmCup}$.
\end{proof}

\begin{cor}\label{cor:CompToRet}
Let $H$ be a graph. Then $\Comp{H} \leap \Ret{H}$.
\end{cor}
\begin{proof}
Let $\calL=\{V(H)\}$. Then the problem $\Hom{H,\calL^*}$ is identical to $\OALHom{H}$. Furthermore, given an irreflexive graph $G$ and a set $\boldS=\{S_v \in \calL \mid v\in G\}$, for this choice of $\calL$ it holds that $\comp{(G,\boldS)}{H}=\comp{G}{H}$.

Then, by Lemma~\ref{lem:MCAlgo}, given a RAS oracle for $\OALHom{H}$, Algorithm~\ref{algo:CoverageMC} computes an $(\eps,\delta)$-approximation of $\abs{\OmCup}=\comp{(G,\boldS)}{H}=\comp{G}{H}$. If we choose $\delta=1/4$ then the algorithm is an FPRAS for $\Comp{H}$. 
\end{proof}

\begin{cor}\label{cor:LCompToLHom}
Let $H$ be a graph. Then $\LComp{H} \leap \LHom{H}$.
\end{cor}
\begin{proof}
Let $\calL=2^{V(H)}$. Then the problem $\Hom{H,\calL^*}=\Hom{H,\calL}$ is identical to $\LHom{H}$. 

From Lemma~\ref{lem:MCAlgo} it follows that given a RAS oracle for $\LHom{H}$, Algorithm~\ref{algo:CoverageMC} returns an $(\eps,\delta)$-approximation of $\abs{\OmCup}=\comp{(G,\boldS)}{H}$. In particular, as $\calL$ is unrestricted, it does so for any valid input $(G,\boldS)$ of the problem $\LComp{H}$. Thus, if we choose $\delta=1/4$, the algorithm is an FPRAS for $\LComp{H}$. 
\end{proof}

To obtain Corollaries~\ref{cor:CompToRet} and~\ref{cor:LCompToLHom}, the only property of compactions we use is the fact that for every compaction from $G$ to $H$ there exists a preimage $U$ of polynomial size, i.e.\ a set $U\subseteq V(G)$ with $|U| \leq \abs{V(H)} + 2\abs{E(H)}$ and a compaction~$\tau$ from~$G[U]$ to~$H$. (This is used in  Equation~(\ref{eq:TGS}).)

Similarly, for every surjective homomorphism from $G$ to $H$ there exists a set $U\subseteq V(G)$ with $|U| = \abs{V(H)}$ such that there exists a surjective homomorphism~$\tau$ from~$G[U]$ to~$H$. If we substitute
\begin{align*}\label{eq:TG}
\TGS=\{(U,\tau) \mid\ &U\subseteq V(G), \abs{U}= \abs{V(H)},\\
&\tau \text{ is a surjective homomorphism from }G[U]\text{ to }H \text{ such that }\forall u\in U, \tau(u)\in S_u\}
\end{align*}
for  Equation~(\ref{eq:TGS}), Lemma~\ref{lem:MCAlgo} still holds and now $\abs{\OmCup}=\sur{(G,\boldS)}{H}$.

Therefore, analogously to the previous two corollaries we obtain the following.
\begin{cor}\label{cor:SHomToRet}
Let $H$ be a graph. Then $\SHom{H} \leap \Ret{H}$.
\end{cor}

\begin{cor}\label{cor:LSHomToLHom2}
Let $H$ be a graph. Then $\LSHom{H} \leap \LHom{H}$.
\end{cor}

Corollaries~\ref{cor:CompToRet} and~\ref{cor:SHomToRet} together constitute Theorem~\ref{thm:SHomCompToRet}.

\subsection{Additional Reductions and Consequences}\label{sec:AdditionalReductions}
The following simple reductions complete our current knowledge of the complexity landscape given in Figure~\ref{fig:ApproxCountingLandscape}.

\begin{lem}\label{lem:APLHomHardness}
Let $H$ be a graph. Then $\LHom{H}\leap \LSHom{H}$ and $\LHom{H}\leap \LComp{H}$.
\end{lem}
\begin{proof}
Let $v_1,\dots, v_q$ be the vertices of $H$ and let $(G,\boldS)$ be an input to $\LHom{H}$. Further, let $H'$ be a copy of $H$ and let $u_1,\dots,u_q$ be the vertices of $H'$ ordered in the same way as $v_1,\dots, v_q$. For $i\in [q]$ let $S_{u_i}=\{v_i\}$ and let $\boldS'=\boldS\cup\{S_{u_i} \colon i\in[q]\}$. Let $G'$ be the disjoint union of $G$ and $H'$. Then $\hom{(G,\boldS)}{H} = \sur{(G', \boldS')}{H}=\comp{(G', \boldS')}{H}$.
\end{proof}

From Corollaries~\ref{cor:LCompToLHom} and~\ref{cor:LSHomToLHom2} as well as Lemma~\ref{lem:APLHomHardness} we immediately obtain Theorem~\ref{thm:LHomEquivalent} which we restate at this point.
{\renewcommand{\thethm}{\ref{thm:LHomEquivalent}}
\begin{thm}
\ThmLHomEquivalent
\end{thm}
\addtocounter{thm}{-1}
}
\bibliography{\jobname}
\end{document}